\documentclass[11pt]{article}
\usepackage[hmargin=1in,vmargin=1in]{geometry}
\usepackage{hchang}

\usepackage{subfigure}
\usepackage{thm-restate,cleveref}
\usepackage[many]{tcolorbox}
\tcbset{enhanced jigsaw,opacityback=0,boxsep=0mm,left=\leftmargin-2mm,right=\leftmargin-2mm,leftrule=2mm,rightrule=2mm}

\nolinenumbers
\def\note#1{}

\title{Cutting Planarians: Planar Emulators for String Graphs}
\date{}

\author{%
Hsien-Chih Chang%
\thanks{Department of Computer Science, Dartmouth College. Email: {\tt hsien-chih.chang@dartmouth.edu}. Supported by the U.S. National Science Foundation Grant No.\ CCF-2443017.}
\and 
Jonathan Conroy%
\thanks{Department of Computer Science, Dartmouth College. Email: {\tt jonathan.conroy.gr@dartmouth.edu}. Supported by the U.S. National Science Foundation Grant No.\ CCF-2443017.}
\and
Zihan Tan%
\thanks{Department of Computer Science and Engineering, University of Minnesota. Email: {\tt ztan@umn.edu}}
\and
Da Wei Zheng%
\thanks{Institute of Science and Technology Austria. Email: {\tt dzheng@ista.ac.at}. This project has received funding from the Austrian Science Fund (FWF) grant  \href{https://www.doi.org/10.55776/I5982}{DOI 10.55776/I5982}. 
For open access purposes, the author has applied a CC BY public copyright license to any author-accepted manuscript version arising from this submission.}
}

\newtheorem{theorem}{Theorem}[section]
\newtheorem{lemma}[theorem]{Lemma}
\newtheorem{claim}[theorem]{Claim}
\newtheorem{observation}[theorem]{Observation}
\newtheorem{remark}[theorem]{Remark}
\newtheorem{corollary}[theorem]{Corollary}
\newtheorem{definition}[theorem]{Definition}

\newcommand{\cR}{\mathcal{R}}
\newcommand{\cN}{\mathcal{N}}
\newcommand{\cS}{\mathcal{S}}

\newcommand{\cC}{\mathcal{C}}
\newcommand{\cT}{\mathcal{T}}
\newcommand{\cO}{\mathcal{O}}
\newcommand{\cF}{\mathcal{F}}
\newcommand{\cG}{\mathcal{G}}
\newcommand{\cP}{\mathcal{P}}

\newcommand{\out}{\textrm{out}}

\newcommand{\len}[1]{\operatorname{len}(#1)}
\newcommand{\lenR}[2]{\operatorname{len}_{#1}(#2)}
\newcommand{\canon}[2]{\ensuremath{#1{\rightarrow}#2}}
\newcommand{\cpath}[2]{\ensuremath{#1{\rightsquigarrow}#2}}
\newcommand{\mThreat}{t}
\def\norm#1{\len {#1}}

\newcommand{\dom}{\operatorname{dom}}
\newcommand{\dist}{\delta}
\newcommand{\bdry}{\partial\!}
\newcommand{\shatter}[2]{\textsc{Shatter}({#1 - \bdry #2})}

\newcommand{\act}[1]{\operatorname{Active}(#1)}

\begin{document}
\maketitle

\begin{abstract}
    In this paper we construct distance sketches for intersection graphs of \emph{arbitrary} path-connected regions in the plane (known as the \emph{string graphs}) in the constant and $1+\varepsilon$ distortion regimes.  Furthermore, the distance sketches themselves are \emph{planar graphs}. 
    First, we show that every unweighted string graph $G$ has an $O(1)$-distortion \emph{planar  emulator}: that is, there exists an edge-weighted planar graph $H$ containing every vertex in $G$, such that every pair of vertices $(u,v)$ satisfies $\delta_G(u,v) \le \delta_H(u,v) \le O(1) \cdot \delta_G(u,v)$.
    Furthermore, we show that for any constant $\varepsilon > 0$, 
    there is an edge-weighted planar graph $H'$ such that
    every pair of vertices $(u,v)$ satisfies $\delta_G(u,v) \le \delta_{H'}(u,v) \le (1+\varepsilon) \cdot \delta_G(u,v) + O(\varepsilon^{-4}\textrm{poly}\log n)$.
    No previous constructions of sparse distance sketches were known even for intersection graphs of simple shapes like axis-parallel rectangles or fat convex polygons.
    
    As applications, we construct the first $(1+\varepsilon, +O(1))$ mixed-distortion tree cover and distance oracle for arbitrary string graphs, as well as the first additive $+(\varepsilon\Delta+O(1))$-distortion embedding of string graphs $G$ with diameter $\Delta$ into graphs of constant treewidth $O(\varepsilon^{-4})$.
\end{abstract}

\newpage
\tableofcontents
\newpage

\section{Introduction}
Given a set of geometric objects $\cS$, the  \EMPH{intersection graph} of $\cS$, denoted \EMPH{$G_{\cS}$}, is a graph with vertex set~$\cS$, such that an edge exists between two objects if and only if they intersect.
The most general form of geometric intersection graph in the plane is the class of \emph{string graphs}, where vertices are arbitrary path-connected regions.   One can assume without loss of generality that the regions are connected arcs, thus the name ``string'' graphs.
One may also study intersection graphs of more restricted class of objects, such as disks, fat objects, or axis-aligned rectangles. 

Our goal is to study \emph{distances} between objects on a geometric intersection graph $G$. 
The \EMPH{distance} between two vertices $u$ and $v$ in $G$, denoted \EMPH{$\dist_{G}(u,v)$}, is the number of edges in the shortest path between $u$ and $v$. 
(Note that we assume edges are unweighted; if we allowed arbitrary edge weights then all metrics are realizable as string graph metrics, as the complete graph is a string graph.\footnote{For some classes of geometric intersection graphs, one can instead impose a natural weight on the edges---for example, a unit-disk graph may have edge weights set to be the Euclidean distance between the disk centers. We are primarily concerned with unweighted string graphs.})
Distance problems on geometric intersection graphs have attracted a great deal of attention, including diameter computation \cite{CS19,CS19_apsp,BKKNP22,CGL24,CCGKLZ25}, clustering \cite{BFI23,friggstad2025qptas}, spanners \cite{YXD12, catusse2010planar, biniaz2020plane, CT23,CH23}, low-diameter decomposition and its applications to LP rounding \cite{lee2017separators, kumar2021constant, kumar2022point, lokshtanov20241}, and much more.
However, few results have been obtained for the most general string graphs beyond restricted classes of shapes, such as unit-disk graphs (UDGs) and axis-parallel unit-square graphs.

In partial explanation of the success enjoyed on unit-disk graphs, several results on UDGs can be traced back to the existence of an $O(1)$-distortion \emph{planar spanner}: every unit-disk graph $G$ has a subgraph $H$ which is planar, such that for every pair of vertices $(u, v)$, we have $\dist_G(u,v) \le \dist_H(u,v) \le O(1) \cdot \dist_G(u,v)$.
In other words, the techniques that enables the results on UDGs relies on the similarity between unit-disk metrics and planar metrics.
Planar metrics are well studied, and tools from planar graphs (in particular, shortest-path separators) can be combined with the planar spanner to design algorithms for UDGs. 
The $O(1)$-distortion planar spanner for UDGs is helpful for distance problems even in the $1+\e$ regime; for example, it is a crucial ingredient for a near-linear time $(1+\e)$-approximate diameter algorithm and a compact $(1+\e)$-approximate distance oracle \cite{CS19}, and in the design of $(1+\e)$-approximate coresets for clustering \cite{BFI23}. (Note that one can also obtain shortest-path separators on UDG directly, not using a planar spanner \cite{YXD12, HZ24}; this is another way in which the metric structure of UDGs is similar to that of planar metrics.) 

The use of planar techniques for geometric intersection graphs so far has seemed somewhat limited to UDGs and closely related classes: for example, the only classes of intersection graphs with planar spanners are Euclidean weighted UDGs \cite{li2002distributed}, unweighted UDGs \cite{biniaz2020plane}, and Euclidean weighted unit-square graphs \cite{BFI23}.
General string graphs could have much more complex interactions than UDGs. 
Unlike disks and squares, objects like segments and rectangles may be \emph{piercing} and thus permit non-local interactions; moreover,
an $n$-vertex string graph may even require a representation (i.e.\ a drawing on the plane) where the number of crossings is exponential in $n$~\cite{KM91}.\footnote{We note a subtlety about the representation: although the number of crossings in the drawing may be exponential, there exists a compressed representation of the drawing via straight-line programs that has polynomial size~\cite{SSS03,SSS07,SSS11}. In this paper, we do not work with the compressed representation.} 

\paragraph{Main results.}
It might come as a surprise then, that in this paper we demonstrate how to construct a planar distance sketching structure for \emph{general} string graphs, even when the objects are arbitrarily complex and piercing.
Our first main result vastly generalizes the previous spanner constructions on UDGs: every string graph admits an $O(1)$-distortion \EMPH{planar emulator}.
No similar construction was known before even for intersection graphs of very restricted shapes, like axis-parallel rectangles or fat convex polygons.%
\footnote{After the conference version of this paper was accepted to STOC'26, a paper by Kisfaludi-Bak and Marx~\cite{KM26} appearing in the same conference announced a construction of an $O(1)$-distortion planar emulator for intersection graphs of similarly-sized fat objects. They constructed this emulator en route to designing a PTAS for Steiner Tree and Subset TSP in such graphs.}
Philosophically speaking, our result suggests that metrics on string graphs are not so different from planar metrics in the $O(1)$-distortion regime.

\begin{theorem}
\label{thm:main}
    For an absolute constant $c$, the following is true:
    For any unweighted string graph $G$, there is a weighted
    planar graph $H$ with $V(G) \subseteq V(H)$, such that for every pair of vertices $(u,v)$ in $G$,
    \[
    \dist_G(u,v) \le \dist_H(u,v) \le c \cdot \dist_G(u,v).
    \]
    Moreover, if one is given a representation of $G$, the graph $H$ can be computed in time polynomial in $|V(G)|$ and  the number of crossings in the representation.%
    \footnote{A careful reading of the proofs shows that the constant $c$ is at most $(126+1)\cdot (2\cdot 753+1) = 191389$. } 
    
\end{theorem}

Note that \Cref{thm:main} allows for the number of vertices in $H$ to be much larger than the number of vertices in $G$. 
Indeed, the number of vertices in $H$ scales linearly in the complexity of the representation of $G$, which is perhaps unsatisfying because some string graphs require exponential-size representations \cite{KM91}.
However, this is simple to fix: we can reduce the number of vertices in $H$ by
applying $O(1)$-distortion Steiner point removal \cite{CCLMST24} on the set of terminals $V(G) \subseteq V(H)$. 
This produces an edge-weighted planar graph $H'$ with $V(H') = V(G)$, which is an $O(1)$-distortion planar emulator of $G$; see \Cref{cor:small-emulator} for details.
Note that our algorithm to \emph{construct} such an $H'$ still requires time polynomial in the size of the representation of $G$.

Our second main result might come as an even bigger surprise.  Not only do $O(1)$-distortion planar emulator exists for arbitrary $n$-vertex string graphs, 
by changing the weight of the edges of the emulator, it is possible to preserve distances with a $(1+\e)$ multiplicative factor, with only an extra $\poly\log n$ additive distortion.
Equivalently, this is a $(1+\e)$-distortion planar emulator that preserves distances between pairs of vertices that are at least polylogarithmically away.
We say that arbitrary string graph admits an \EMPH{$(\alpha,+\beta)$-mixed-distortion} planar emulator for $\alpha = 1+\e$ and $\beta = \poly\log n$.
We focus on string graphs,
but a similar result actually holds more generally for every graph class that admits $O(1)$-distortion planar emulators; see \Cref{thm:multToAdd}.

\begin{theorem}
\label{thm:main2}
    Let $G$ be a $n$-vertex string graph.
    For any small constant $\e > 0$,
    there is a weighted planar graph $H$ with $O(n)$ vertices and $V(G) \subseteq V(H)$, such that for every pair of vertices $(u,v)$ in $G$,
    \[
    \dist_G(u,v) \le \dist_H(u,v) \le (1+\e) \cdot \dist_G(u,v) + O(\e^{-4}\cdot \log^{17} n).
    \]
    Moreover, if one is given a representation of $G$, the graph $H$ can be computed in time polynomial in $|V(G)|$ and the number of crossings in the representation.
\end{theorem}

This result provides the first compelling evidence towards a conjecture of Agelos Georgakopoulos (as described by Nguyen, Scott, and Seymour~\cite{nguyen2025asymptotic}) in coarse graph theory:  
Every unweighted graph that can be embedded into a planar graph with $(\alpha, +\beta)$-distortion can also be embedded into a planar graph with pure additive distortion $+\beta'$, for some (possibly bigger) constant $\beta'$.
(The paper \cite{nguyen2025asymptotic} actually gives a more general conjecture for arbitrary class of graphs instead of just planar graphs, and in terms of \emph{quasi-isometry}. A formal definition of quasi-isometry can be found in Section~\ref{SS:approx-reduction} when we prove the existence of $(1+\e,
+\poly\log n)$-mixed-distortion emulator.)
The conjecture may seem bold, but %
there are no known counter-examples.
It is known to hold 
in other very restricted classes of graphs, such as trees~\cite{cdn+-caaeg-2012,kerr23} and bounded-pathwidth graphs~\cite{nguyen2025asymptotic}.
While we remain neutral on the validity of the original conjecture, we 
leave as an open question whether a $(1+\e,+O(1))$ or a $(1,+\poly\log n)$-distortion planar emulator exists.

\paragraph{Applications.}
In addition to the $(1+\e, +\poly\log n)$-mixed-distortion planar emulators, we have two other applications on constructing distance-sketching structures for any graph class $\cG$ (closed under induced subgraphs) that admits $O(1)$-distortion planar emulators.
First, we can construct $(1+\e, +O(1))$-tree covers of size $O(\e^{-3} \log (\e^{-1}))$ for every graph in $\cG$.
A tree $T$ is a \EMPH{dominating tree} of a graph $G$ if $V(G) \subseteq V(T)$, and $\dist_G(u,v) \le \dist_T(u,v)$ for every pair of vertices $(u,v)$ in $G$. An \EMPH{$(\alpha, \beta)$-tree cover} for a graph $G$ is a family of dominating trees $\cT$, such that for every pair of vertices $(u,v)$ in $G$, there exists some tree $T \in \cT$ with 
\[\dist_T(u,v) \le \alpha \cdot \dist_G(u,v) + \beta.\]
The \EMPH{size} of the tree cover is the cardinality $|\cT|$.
Tree covers are widely studied for both general metrics and restricted families such as Euclidean/doubling metrics, planar metrics, and minor-free metrics \cite{awerbuch1992routing, arya1995euclidean,gupta2005traveling,bartal2019covering,CCLMST23, CCLMST24}, as well as on
unit-disk graphs~\cite{weng2025thesis}.
In particular, if $G$ is planar, then for any constant $\e > 0$, there is a $(1+\e, +0)$-tree cover of $G$ with constant size $O(\e^{-3} \log (\e^{-1}))$ \cite{CCLMST23}.
We can use our $O(1)$-distortion planar emulator for string graphs to recover a similar result.
(Again we focus on string graphs, but the proofs do apply more generally to any suitable class of graphs $\mathcal{G}$ that admits $O(1)$-distortion planar emulators.)
Notice that the additive distortion needed for the tree cover result in an absolute constant $O(1)$, instead of $\poly\log n$ as in the result of $(1+\e)$-mixed-distortion planar emulator.

\begin{theorem}
\label{thm:tree-cover}
    Let $\mathcal{G}$ be a class of unweighted graphs that is closed under induced subgraphs, such that every $G \in \mathcal{G}$ has an $O(1)$-distortion planar emulator. Then for any $\e > 0$, $G$ has a $(1+\e, +O(1))$-tree cover of size $O(\e^{-3} \log (\e^{-1}))$.
\end{theorem}

One key application of tree covers is an easy construction of \emph{compact distance oracle}: here one aims to construct a data structure that can quickly return an (approximate) distance when queried on a pair of vertices. If a graph has small tree cover, constructing a distance oracle reduces to lowest-common-ancestor data structures in trees \cite[Theorem 1.6]{CCLMST23}.
\begin{corollary}
Let $\mathcal{G}$ be class of unweighted graphs that is closed under induced subgraphs, such that every $G \in \mathcal{G}$ has an $O(1)$-distortion planar emulator.
    Let $G$ be an $n$-vertex graph in $\mathcal{G}$. For any $\e >0$, there is a data structure with size $O_\e(n)$ which, when queried on a pair of vertices $(u,v)$ in $G$, returns an estimate $\tilde\dist(u,v)$ such that $\dist_G(u,v) \le \tilde\dist(u,v) \le (1+\e) \cdot \dist_G(u,v) + O(1)$.
    \begin{itemize}
        \item In the word RAM model with word size $\Omega(\log n)$, queries can be answered in $O_\e(1)$ time.
        \item In the pointer machine model, queries can be answered in $O_\e(\log \log n)$ time.
    \end{itemize}
\end{corollary}

For our next application, an active line of work attempts to embed planar (or minor-free) graphs $G$ into graphs of small \emph{treewidth}, while incurring a small additive distortion of $+\e \cdot \diam(G)$, for any constant $\e$ \cite{fox2019embedding,CFKL20,filtser2022low,CCLMST23, CCLMST24}. 
In our language, we seek a $(1, +\e \diam(G))$-distortion planar emulator $H$ such that the treewidth ${\rm tw}$ of $H$ is minimized. In planar graphs, one can take ${\rm tw} = O(\e^{-4})$ \cite{CCLMST23}. The proof of \cite{CCLMST23} uses the small-size tree cover; we show that their construction can be extended to string graphs, with the help of the $O(1)$-distortion planar emulator.

\begin{theorem}
\label{thm:tw-embedding}
    Let $G$ be an unweighted graph with diameter $\Delta$, such that $G$ has a $O(1)$-distortion planar emulator. For any $\e > 0$, there is a $(1, +(\e \Delta + O(1)))$-distortion emulator $G'$ of $G$ with treewidth $O(\e^{-4})$.
\end{theorem}

\paragraph{Parallel work.}
An equivalent result to our first main theorem
(\Cref{thm:main}) has recently been announced independently
by Davies~\cite{davies2025string}: every unweighted string graph is \emph{quasi-isometric} to an unweighted planar graph.
A quasi-isometry implies the existence of an $O(1)$-distortion planar emulator, and one can readily check that our emulator indeed gives a quasi-isometry.%
\footnote{Specifically, \cite{davies2025string} constructs an unweighted planar graph $H$ with $V(H) = V(G)$ such that $\dist_G(u,v)/O(1) \le \dist_H(u,v) \le O(1) \cdot \dist_G(u,v)$. To obtain this result from our \Cref{thm:main}, we first apply Steiner point removal to our emulator to produce a weighted graph $H$ with $V(H) = V(G)$. Without loss of generality, every edge has weight at most $O(1)$ (see \Cref{obs:bounded-edges}), so viewing $H$ as an unweighted graph (ie, ignoring the weights) yields a similar distortion guarantee to \cite{davies2025string}.}
The proof of Davies was announced publicly on arXiv in October 2025, and our proof of \Cref{thm:main} was posted on arXiv two days later. Subsequently, we learned through personal communication with Davies that he had obtained his proof earlier, though the result was not published until 2025. 
There was no communication between the authors and Davies before both works appeared on arXiv.

Our core technique for proving \Cref{thm:main} is an adaptation of the \emph{shortcut partition} construction using gridtrees developed in \cite{CCLMST23}, whereas the work by Davies builds the quasi-isometry from first principles. %
Davies proves several applications and refinements of \Cref{thm:main} in his paper \cite{davies2025string}.
This includes settling the Assouad-Nagata dimension of string graphs, an important notion in structural graph theory and geometry group theory related to sparse covers; connections to central conjectures in coarse graph theory via Georgakopoulos-Papasoglu;
and establishing the quasi-isometry between complete Riemannian planes and locally finite planar graphs.
To the best of our knowledge, our second main result (\Cref{thm:main2}) and the applications on tree covers and embedding into bounded-treewidth graphs in the $1+\e$ regime are completely new.

\section{Technical Overview}
\subsection{O(1)-Distortion Planar Emulator}
\label{SS:O(1)-emulator}
\begin{wrapfigure}{r}{0.35\textwidth}
\centering
\vspace{-1pt}  
\includegraphics[width=0.35\textwidth]
{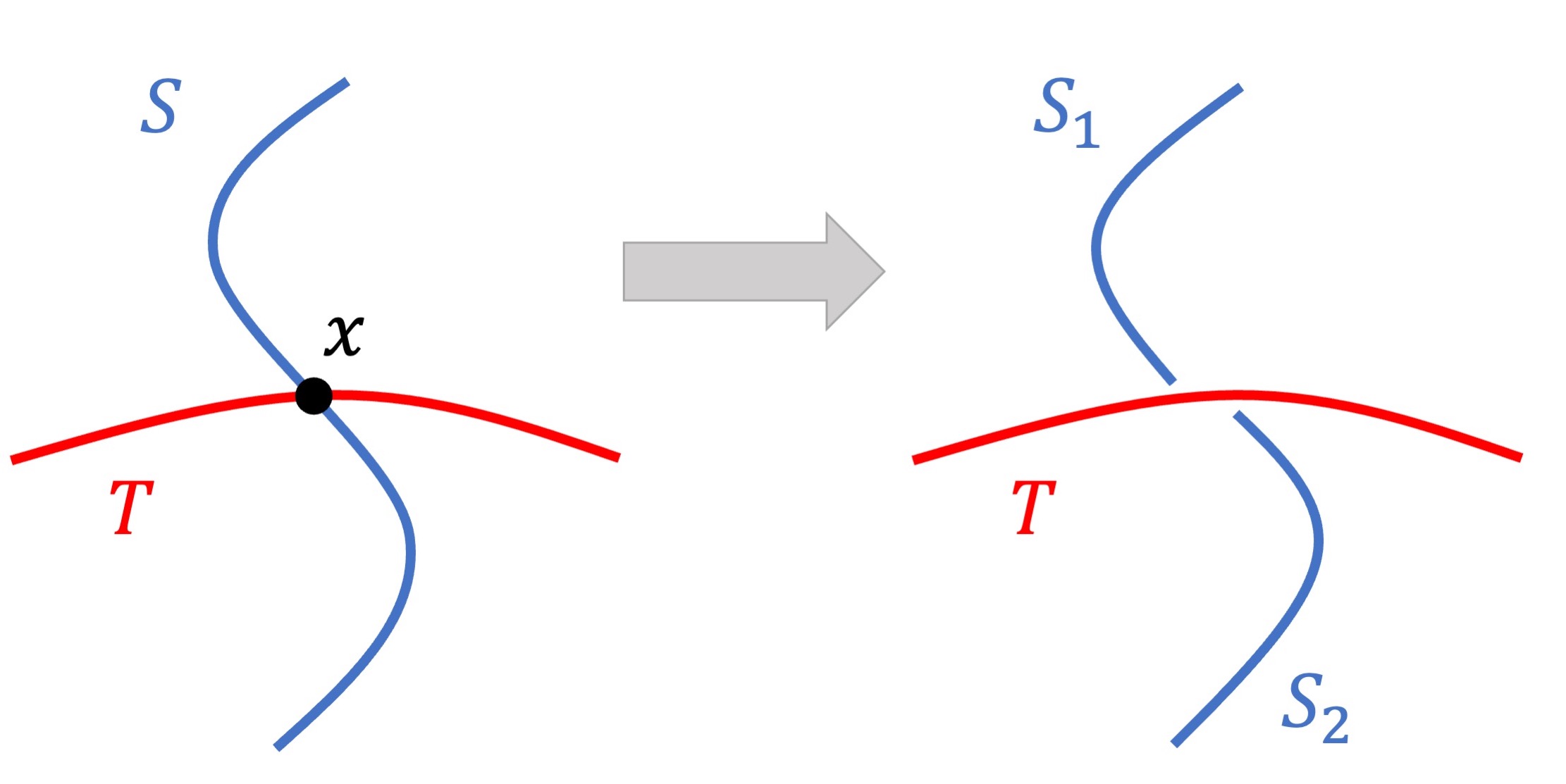}
\caption{Shattering $S$ at $x$}
\vspace{-10pt}
\end{wrapfigure}
We begin by describing a natural, but \emph{flawed} approach at creating a planar emulator for a string graph $G_\cS$ on a set of strings $\cS$.
Let $S$ and $T$ be two strings that cross at some point $x \in \R^2$.
The operation of \EMPH{shattering} of $S$ at $x$ can be thought of as cutting apart $S$ at $x$ to produce two new strings $S_1$ and $S_2$, both of which \emph{touch} $T$ but neither of which \emph{crosses} $T$ at $x$.
Given a set of strings $\cS$, we could transform their intersection graph $G_\cS$ into a planar graph by the following procedure: for every pair of strings that cross, shatter one of them to remove the crossing.\footnote{When a string $S$ is shattered into $S_1$ and $S_2$, one of the strings $S_1$ or $S_2$ is chosen arbitrarily to represent $S$, and the other string is treated as a Steiner vertex.}
Each shattering procedure removes a crossing at the cost of changing distances by an O(1) amount --- the single string $S$ is turned into two strings $S_1$ and $S_2$ which are within distance 2 of each other in the new intersection graph (as witnessed by the path $[S_1, T, S_2]$).
After all crossings are removed by the shattering procedure, it is not difficult to see that the resulting intersection graph on the shattered strings $\cS'$ is planar. However, distances on the shattered graph $G_{\cS'}$ could look very different than on the original $G_{\cS}$. 
When the intersection graph $G$ is a clique (for example), some string $S$ might cross every other string, so $S$ could be shattered into $\Theta(n)$ pieces, and these pieces could be at distance $\Theta(n)$ from each other. In this case, shattering does not seem helpful in producing a planar emulator.

On the other hand, it is easy to find an $O(1)$-planar emulator $H$ when $G$ is a clique: just take $H$ to be a star graph. 
This gives us the following hope: we could try to remove some crossings using the shattering procedure (only shattering each string $O(1)$ times), and hope that the remaining intersections happen in local ``clusters'' that can be replaced with stars. 

\newpage
\begin{quote}
    \textbf{Goal.} Given strings $\cS$, find a partition
    of $\R^2$ into connected clusters and shatter all strings in $\cS$ by the boundary of the clusters such that (1) each string is only shattered by the boundary of $O(1)$ clusters, and (2) for each cluster $C$, the (shattered) strings in $C$ are all within $O(1)$ distance of each other. %
\end{quote}
Given such a partition, it is fairly straightforward to find a planar emulator, 
by replacing each cluster with a star; see \Cref{clm:constant-emulator} below for the details, and \Cref{fig:partition} for an example.
Now, a key insight is that our stated goal is reminiscent of the notion of \EMPH{scattering partition} introduced by Filtser~ \cite{filtser2024scattering}.
\begin{quote}
    \textbf{Scattering partition.} A graph $G$ has an $O(1)$-scattering partition if, for any distance $\Delta > 0$, there is a partition of $V(G)$ into connected clusters such that: (1) any \emph{path} of length $\Delta$ intersects only $O(1)$ clusters, and (2) each cluster has diameter $O(\Delta)$.
\end{quote}
Our intuition is that a scattering partition with $\Delta = O(1)$ is similar to the partition in our goal. Unfortunately, scattering partitions are only known for some very restricted families of graphs, such as trees, cactus graphs~\cite{filtser2024scattering}, and series-parallel graphs~\cite{hershkowitz20221}. 
Recently, Chang~\etal~\cite{CCLMST23} constructed a very similar (albeit slightly weaker) object---the \EMPH{shortcut partition}---for \emph{planar graphs} (and later for minor-free graphs \cite{CCLMST24}). 
Our main technical contribution is to delicately adapt the planar construction of \cite{CCLMST23} to (almost) achieve our stated goal. 
Below, we formally state the partition we aim to construct (\Cref{lem:good-clustering}) and describe how to construct a planar emulator from this partition. We then sketch some intuition for the proof of \Cref{lem:good-clustering}.

\begin{figure}[t]
	\centering
	\subfigure[A region intersection graph $G$.]
	{\scalebox{0.065}{\includegraphics{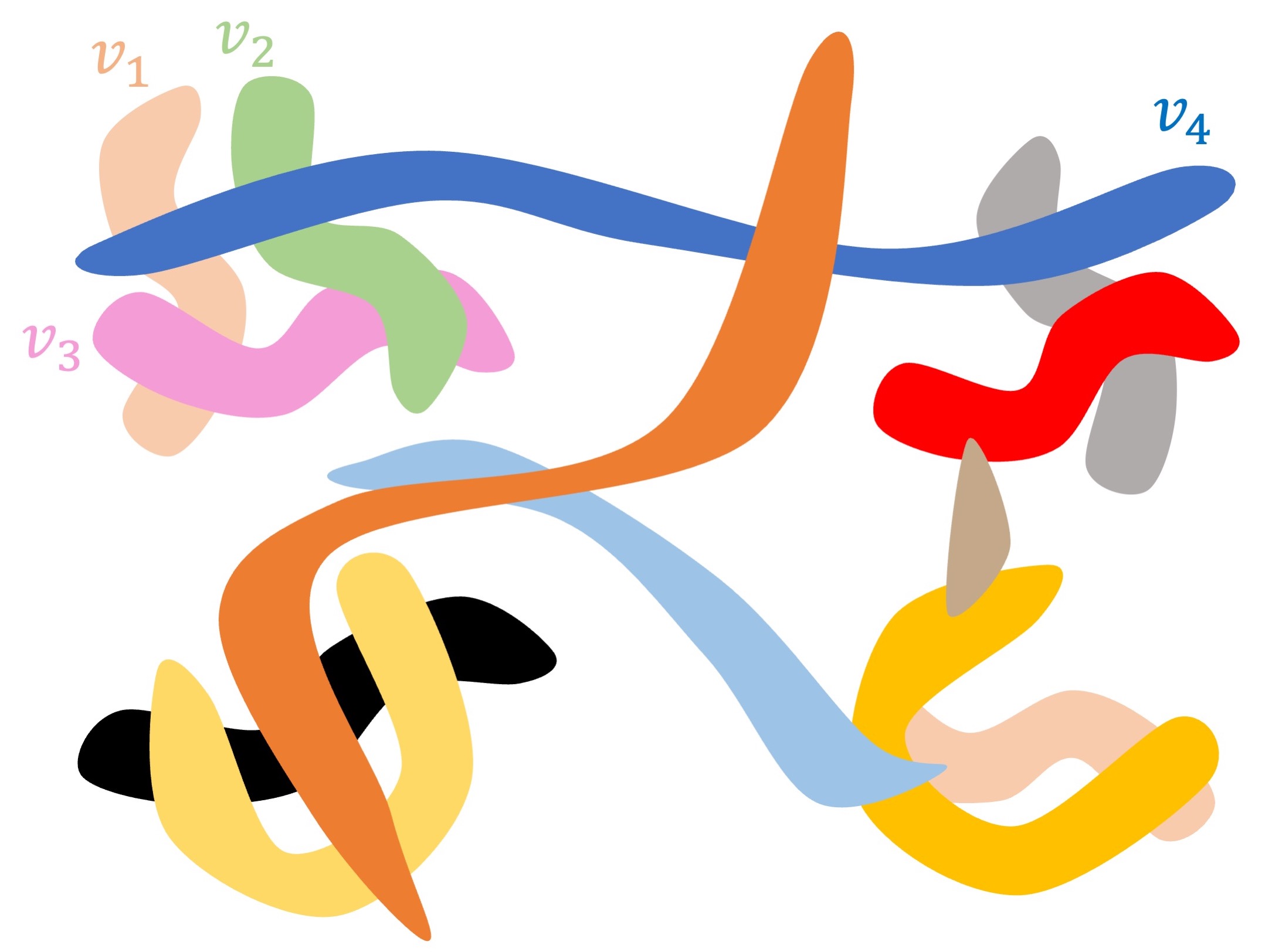}}}
	\hspace{0.3cm}
	\subfigure[A partition of $\R^2$ into clusters.]
	{\scalebox{0.062}{\includegraphics{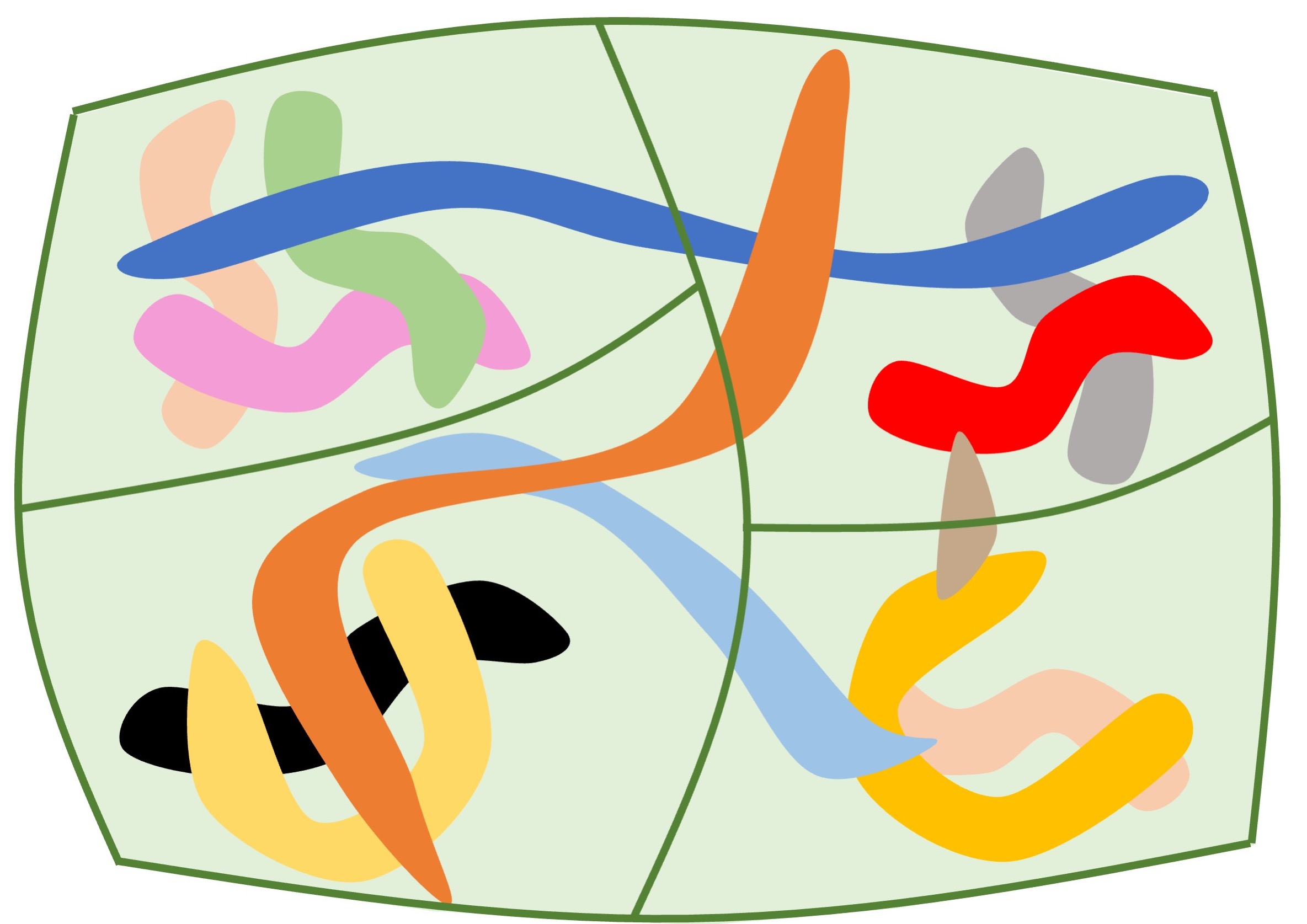}}}
	\hspace{0.3cm}
	\subfigure[A planar emulator for $G$ by contracting all clusters and adding stars.]
	{\scalebox{0.08}{\includegraphics{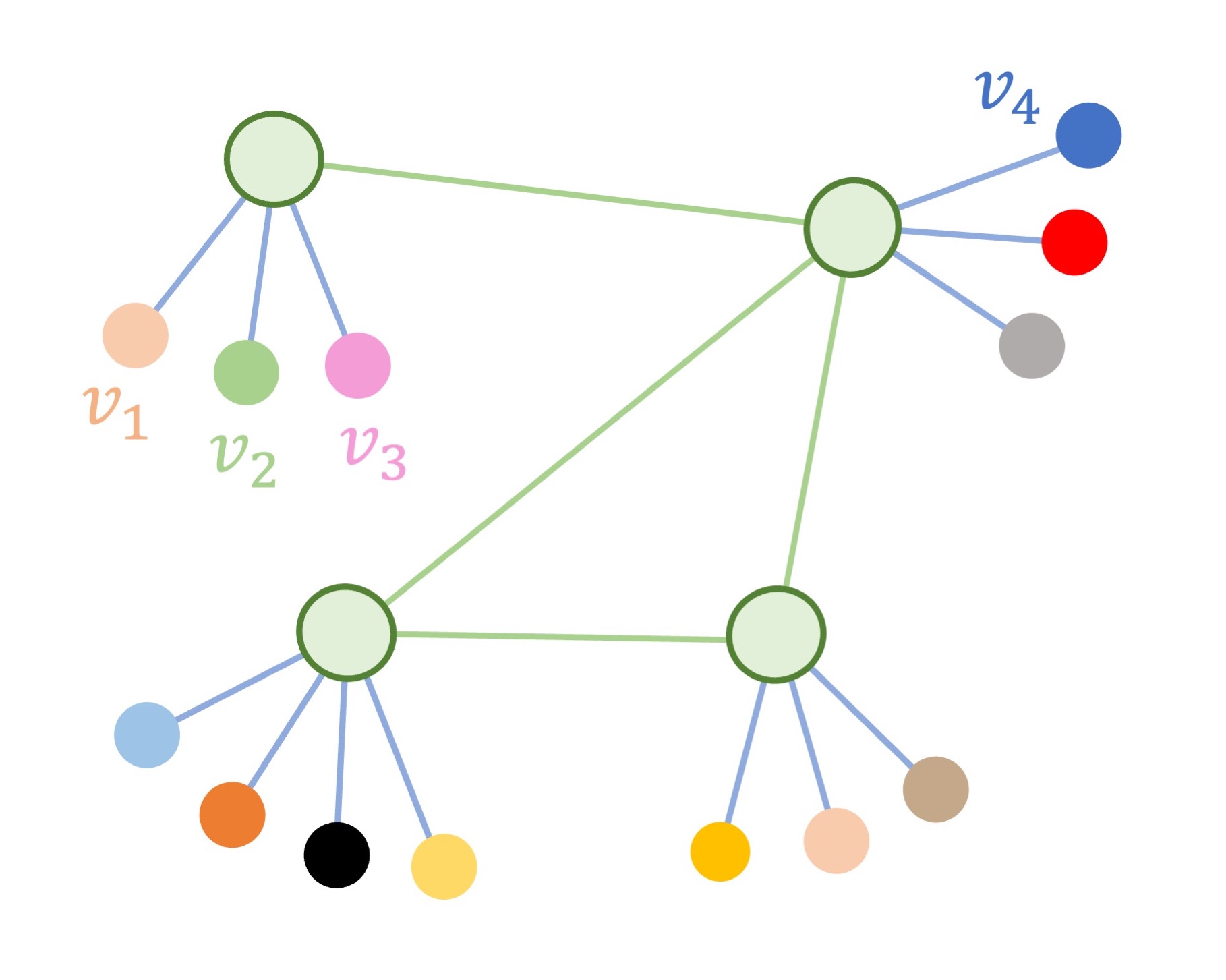}}}
	\caption{An illustration of planar emulator construction by contracting clusters.\label{fig:partition}}
\end{figure}

\subsubsection{Constructing an emulator from a clustering}
\label{SSS:constant-emulator-construction}
Rather than working directly with string graphs drawn in $\R^2$, we work with region intersection graph on planar graphs, as defined by Lee~\cite{lee2017separators}. Let \EMPH{$G$} be an arbitrary planar graph, which we call the \EMPH{base graph}. %
A \EMPH{region} is a connected subset of vertices in $G$; these are the analog of strings. %
Let \EMPH{$\cR$} be a set of regions on $G$. 
The \EMPH{region intersection graph} of $\cR$ is the graph with vertex set $\cR$, and an edge is added between every two regions $R_1$ and $R_2$ that share at least one vertex in $G$ (that is, there exists vertices $v_1 \in R_1$ and $v_2 \in R_2$ with $v_1 = v_2$). 
Clearly every string graph is a region intersection graph on a planar graph, and vice versa  \cite[Lemma 1.4]{lee2017separators}.
For technical reasons, described below, we actually will work with something we call the \EMPH{region contact graph}, denoted \EMPH{$G_\cR$}: this is the graph with vertex set $\cR$, where there is an edge is added regions $R_1$ and $R_2$ if and only if there are vertices $v_1 \in R_1$ and $v_2 \in R_2$ such that either $v_1 = v_2$ or \ul{$v_1$ is adjacent to $v_2$ in $G$}. It is not to hard to see that every string graph is a region contact graph on a planar graph, and vice versa; see \Cref{clm:string-to-contact} for the proof of the forward direction.
A \EMPH{cluster} of a graph $G$ is a set of vertices $C \subseteq V(G)$ such that $G[C]$ is connected. 
\begin{lemma}
\label{lem:good-clustering}
    There are absolute constants $\alpha_{\rm diam} = 126$ and $\alpha_{\rm hop} = 753$
    such that the following holds. 
    For any planar (base) graph $G$ and any set of regions $\cR$ on $G$,  there is a partition of $V(G)$ into disjoint clusters $\cC = \{C_1, C_2, \dots, C_k\}$ such that
    \begin{enumerate}
        \item \textnormal{[Diameter.]} For every cluster $C\in \cC$, for every pair of regions $R_1, R_2 \in \cR$ that intersects $C$, we have $\dist_{G_\cR}(R_1, R_2) \le \alpha_{\rm diam}$.
        \item \textnormal{[Scattering.]} For every region $R \in \cR$ and every two vertices $v_1, v_2 \in R$, there is a path in $G$ between $v_1$ and $v_2$ that intersects at most $\alpha_{\rm hop}$ clusters.
    \end{enumerate}
\end{lemma}
Notice the minor but crucial change in the [scattering] condition, compared to our goal %
stated earlier in the technical overview: 
instead of asking each string/region $R$ to intersect only $O(1)$ many clusters, we instead guarantee the existence of \emph{some path} between every pair of vertices in $R$ that intersects $O(1)$ many clusters.
(This change is similar to the difference between shortcut partition and scattering partition.)
Equipped with \Cref{lem:good-clustering}, we can prove \Cref{thm:main} by constructing an emulator $H$ as~follows.

Given a string graph, we find a representation of the string graph as a region contact graph \EMPH{$G_\cR$} on some planar graph $G$ (see \Cref{clm:string-to-contact}). Assume without loss of generality that the regions $\cR$ collectively cover all vertices in $G$ (otherwise, uncovered vertices can be deleted without changing the contact graph $G_\cR$). Apply \Cref{lem:good-clustering} on the base graph $G$ and set of regions $\cR$, and let $\cC$ be the resulting partition of $V(G)$ into clusters.
Let \EMPH{$\hat H$} be the \EMPH{cluster graph} obtained from $G$ by contracting every (connected) cluster in $\cC$ into a single vertex. More precisely, $\hat H$ is the graph with vertex set $\cC$, where there is an edge between two clusters $C_1, C_2 \in \cC$ if and only if some vertex in $C_1$ is adjacent to some vertex in $C_2$ in $G$.
Initialize $\EMPH{$H$} \gets \hat H$. For every region $R \in \cR$, create a new vertex in $H$ representing $R$, and add an edge between this vertex and an arbitrary cluster $C \in V(\hat H)$ that intersects $R$; this cluster $C$ is called the \EMPH{representative cluster} of $R$. See \Cref{fig:partition}.
We will consider $H$ to be a weighted graph where each edge has weight $\alpha_{\rm diam}+1$.
Clearly $H$ is planar: the contracted graph $\hat H$ is planar as it is obtained by contractions of the planar base graph $G$, and attaching new vertices by a single edge to $\hat H$ maintains planarity. We claim $H$ is an $O(1)$-distortion planar emulator for $G_\cR$; we defer the proof to \Cref{SS:const-emu}. 

\begin{restatable}{claim}{constantEmulator}
\label{clm:constant-emulator}
$H$ is a $(2\alpha_{\rm hop}+1)\cdot (\alpha_{\rm diam}+1)$-distortion planar emulator of the region contact graph $G_\cR$.
\end{restatable}

\subsubsection{Constructing the clustering: A comparison with \cite{CCLMST23} and \cite{BLT14}}

How do we prove \Cref{lem:good-clustering}? As discussed above, we delicately adapt (and at some point completely deviate from) the shortcut partition construction of \cite{CCLMST23} in planar graphs.

\paragraph{Shortcut partition of \cite{CCLMST23}.} 
We begin by summarizing the partition of \cite{CCLMST23}, which itself builds on the construction of \cite{BLT14} for sparse covers. 
Let $\Delta > 0$, and let $G$ be a plane graph (i.e. a planar graph $G$ given along with an embedding of $G$ in the plane).
We say that the pair $(u,v)$ in $G$ with $\dist_G(u,v) \le \Delta$ is \EMPH{$\tau$-scattered} by a set of clusters $\cC$ if there is a path of length $O(\Delta)$ between $u$ and $v$ in $G$ that intersects at most $\tau$ clusters. 
A \EMPH{shortcut partition} is (roughly speaking\footnote{Actually, we are stating the closely-related definition of \emph{approximate scattering partition} \cite{CCLMST24}, which is slightly weaker than shortcut partition; this definition suffices for our exposition.}) a partition of $V(G)$ into clusters $\cC$ each of diameter $O(\Delta)$, such that any pair of vertices $(u,v)$ in $G$ with $d_G(u,v) \le \Delta$ is $O(1)$-scattered by $\cC$.

If $G$ consists of a shortest path $\pi$ (called the \EMPH{spine}) together with some vertices within $\Delta$ distance of $\pi$, we call $G$ a \emph{spined supernode} (or just \EMPH{supernode}, for short). Observe that it is easy to construct a scattering partition in the special case where $G$ is a spined supernode.
Indeed, one can first 
pick a maximal set of vertices $\cP$ on $\pi$ that are $\Theta(\Delta)$ apart,
and then partition $G$ into clusters according to a Voronoi partition with respect to $\cP$
where each vertex is assigned to the cluster of its closest point in $\cP$.
It is clear that every cluster has diameter $O(\Delta)$, and (because $\pi$ is a shortest path) any
pair of vertices at distance $O(\Delta)$ in $G$ are $O(1)$-scattered.
If $G$ is a general plane graph, \cite{CCLMST23} first partitions $G$ into subgraphs, each of which is a spined supernode, such that every pair of vertices at distance $O(\Delta)$ is $O(1)$-scattered by the set of spined supernodes.
A shortcut partition then follows from partitioning each spined supernode into clusters as described above.

The set of supernodes is constructed recursively: in each iteration \cite{CCLMST23} assign some of the vertices of $G$ to spined supernodes, then recurse on the subgraph induced by the unassigned vertices. They guarantee that, in each iteration, every pair of vertices at distance $O(\Delta)$ is $O(1)$-scattered \emph{by the supernodes created at the current iteration}.
Moreover, they guarantee that in each iteration, all vertices within distance $O(\Delta)$ are assigned to some supernode; using this fact, they prove that any path of length $\Delta$ only intersects supernodes created in two iterations, and thus any pair of vertices at distance $O(\Delta)$ is $O(1)$-scattered by the final set of supernodes.

We now describe the \cite{CCLMST23} construction of supernodes in each iteration.
Intuitively, they sweep across the plane graph $G$ from left to right while maintaining a \EMPH{separating supernode} $S$; the spine of $S$ has endpoints on the outer face of $G$, and it separates the left side of $G$ (which has already been processed) from the right side (which has yet to be processed).
In a bit more detail: the supernode $S$ is initialized to be the $O(\Delta)$ neighborhood around an arbitrary vertex on the outer face. Then the following process is repeated: add $S$ to the set of spined clusters $\cC$; then delete all vertices in $S$ from $G$; then move $S$ one step forward on the outer face (roughly, let $v_1'$ and $v_2'$ be two vertices adjacent to $S$ on the outer face, and update $S$ to be the supernode comprising the shortest path from $v_1'$ to $v_2'$ in $G-S$ plus an $O(\Delta)$-neighborhood around the path).
In this way, $S$ is a ``thick path'' that acts as a separator as it sweeps across all vertices in the outer face.
The separation property of $S$ is used to show that any pair of vertices at distance $O(\Delta)$ is $O(1)$-scattered by the supernodes.
Observe that the construction guarantees all vertices on the outer face are assigned to a supernode.
After sweeping across the entire graph $G$, \cite{CCLMST23} perform an \EMPH{expansion} step: every vertex within distance $O(\Delta)$ of the outer face is assigned to its closest supernode (that is, each supernode grows, in a Voronoi manner, by distance $O(\Delta)$.

To summarize, the construction of \cite{CCLMST23} involves four steps\footnote{A historical note: The first and fourth steps appeared in \cite{BLT14} in the context of sparse cover, and the second and third were new to \cite{CCLMST23}.}: (1) sweep along the outer face to select spined supernodes; (2) expand the supernodes by $O(\Delta)$ distance; (3) recurse on the subgraph induced by unassigned vertices; (4) chop each spined supernode into $O(\Delta)$-diameter clusters. 

\paragraph{Adaptation to string graphs.}
Our approach to proving \Cref{lem:good-clustering}, in the setting of region contact graphs, also involves the four major steps of \cite{CCLMST23}. 
We would like to sweep across the contact graph $G_\cR$ (computing shortest paths in $G_\cR$, and maintaining a separating supernode $S$ which is a subset of regions $\cR$), then expand the supernodes and partition them into clusters, and then recurse on the regions not assigned to any clusters.
However, there are significant technical obstacles to this plan, and we must deviate from \cite{CCLMST23}. Here we highlight three differences, with increasing complexity.
\begin{enumerate}
    \item When \cite{CCLMST23} sweep across the graph and maintain a separating supernode $S$, they essentially rely on the \EMPH{Jordan curve theorem} to prove that $S$ separates the already-processed part of the graph from the yet-to-be-processed part: roughly, if $S$ contains a cycle in $G$, then no vertex inside the cycle is adjacent to a vertex outside the cycle.
    In region intersection graphs, the Jordan curve theorem does not immediately give a separator: even if $S$ is set of regions in $G_\cR$ that contains a cycle in $G$, there may be a region that both contains a vertex of $G$ inside the cycle and outside the cycle. However, any such region must contain at least one vertex of $G$ on the cycle. Thus the \emph{1-neighborhood} of $S$ is a separator for the intersection graph $G_\cR$. We use this fact to prove that we can maintain a separating supernode while sweeping across $G_\cR$.
    \item When sweeping across $G$, \cite{CCLMST23} delete the vertices in $S$ and then continue sweeping across the remaining part of the graph. This guarantees that the supernodes are disjoint, and so (later on) each one can be independently partitioned into clusters. However, when we sweep across $G_\cR$, it is \emph{not} sufficient to just delete the regions in $S$ from $\cR$ and continue. This is because we eventually want a partition of $V(G)$ into clusters, not a partition of $\cR$. Thus, if $S_1$ and $S_2$ are supernodes, it is not enough to guarantee that $S_1 \cap S_2 = \varnothing$, that is, there is no region of $\cR$ contained in both $S_1$ and $S_2$; rather, we need to guarantee that $S_1$ and $S_2$ are \EMPH{vertex disjoint}, no vertex of $V(G)$ is contained in both a region of $S_1$ and a region of $S_2$. To achieve this stronger requirement, we \EMPH{shatter}\footnote{as described at the beginning of the technical overview; see \Cref{sec: O(1)} for a precise definition} 
    all regions in $\cR$ along the boundary of supernode $S$, before deleting $S$ and continuing to sweep. 
    We show that (roughly speaking) each region only gets shattered $O(1)$ times.
    The region contact graph (unlike the region intersection graph) behaves nicely under the shatter operation, and we use this to bound the distortion incurred by shattering.
    \item Finally, the last issue arises during the expansion step of \cite{CCLMST23}. This procedures requires computing a \EMPH{Voronoi partition}: we assign each unassigned vertex to its closest supernode, producing a set of larger, vertex-disjoint, supernodes. The same Voronoi procedure is also used when partitioning each supernode into clusters. 
    Unfortunately, in region contact graphs, we have no comparable way to produce \emph{vertex-disjoint} clusters: performing a Voronoi partition on the regions $\cR$ in $G_\cR$ produces clusters that do not share regions, but again these clusters may not be vertex-disjoint. To get around this, we need to combine the step of supernode expansion with the step of partitioning each supernode. The combined expand-and-cluster procedure carves out clusters in a delicate order to simultaneously guarantee the bounded diameter and expansion properties.
\end{enumerate}

\subsection{\boldmath$(1+\e, +\poly\log n)$-Distortion Planar Emulator}

\paragraph{A toy example.}
Next we describe our plan to turn an $O(1)$-distortion planar emulator into another planar emulator with $(1+\e, +\polylog n)$-mixed-distortion guarantee.
Like in the previous section, we begin with a toy example to see why such a reduction might even be possible. Let $C = O(1)$ be the constant so that every string graph has a $C$-distortion planar emulator.
Let us imagine a string graph%
\footnote{For ease of exposition, in this section we focus on string graphs. But there is nothing special about the choice of $G$ being a string graph: as mentioned, our reduction from multiplicative distortion to $(1+\e)$-mixed-distortion works for any graph class that admits $O(1)$-distortion planar emulators. Our formal lemmas are all stated in this more general language.} $G$ where every vertex of the graph is in the $O(C)$-neighborhood of a single shortest path $\pi$; that is, $G$ is a \emph{spined supernode}. In this case, the \EMPH{path-straightening} approach of Nguyen, Scott, and Seymour~\cite{nguyen2025asymptotic} shows that there is a planar emulator for $G$ with purely additive $+O(1)$ distortion.%
\footnote{Actually, \cite{nguyen2025asymptotic} get something slightly different than emulator: they allow distances to \emph{contract} by a small amount, but note that it is easy to remove this contraction by adding some Steiner vertices.}
Indeed, let $H$ be a $C$-distortion planar emulator of $G$ on the same set of vertices. Consider the image of $\pi$ in $H$: every edge in $\pi$ is mapped to a path of length at most $C$ in $H$. The resulting walk $\pi'$ may be longer than $\len \pi$ by a factor $C$,
but \cite{nguyen2025asymptotic} show that one can ``straighten'' $\pi'$ into a shortest path with length exactly equal to $\norm{\pi}$ by assigning edge weights appropriately.
Very loosely speaking:%
\footnote{There are technical difficulties because $\pi'$ may be a walk, not a path. In actuality, not all vertices on the walk $\pi'$ remain in the straightened path; 
see \Cref{thm:NSS25-2.2} for a formal statement.}
for every edge in $\pi$, the image of the edge $e$ in $H$ is a walk with length at most $C$; we set one edge in this walk to have weight 1 and the rest to have weight 0. Every edge not in the image of $\pi$ is assigned to have some large weight $C^2$.

After this reweighting, \cite{nguyen2025asymptotic} show that the resulting graph $(H, w)$ has \emph{purely additive} distortion $+ \poly(C)$, rather than multiplicative distortion.
We sketch the proof:
Let $u$ and $v$ be a pair of vertices in~$G$.
First, find the vertex $u'$ on $\pi$ that is closest to $u$ in $G$.
The distance between $u$ and $u'$ is at most $O(C)$ in $G$, and at most $O(C^2)$ in $H$ because $H$ is a $C$-distortion planar emulator, and therefore at most $O(C^4)$ in the reweighted graph $(H,w)$.
Similarly we find the vertex $v'$ on $\pi$ that is closest to $v$ in $G$, so $v$ and $v'$ also has distance at most $O(C^4)$ in $(H,w)$.
Finally, consider the shortest path between $u'$ and $v'$ in $G$.  Because the path is effectively straightened in $(H,w)$, there is no multiplicative distortion accumulated when walking along $\pi'$ in the weighted graph $(H,w)$.
This proves $\dist_{(H,w)}(u,v) \le \dist_G(u,v) + \poly(C)$.

\paragraph{Detours through \boldmath{$C$}-disjoint shortest paths.}
The argument above shows that it is easy to find additive emulators when $G$ is a spined supernode (i.e., an $O(C)$-neighborhood around some shortest path). What about general string graphs? 
One idea is to draw inspiration from \Cref{SS:O(1)-emulator} and compute a partition of $V(G)$ into disjoint spined supernodes.
Suppose we somehow could find a partition such that any shortest path between two vertices $u$ and $v$ intersects only $O(1)$ spined supernodes. 
In this case we could obtain an emulator with purely additive distortion by straightening each spined supernode independently:
walking past each supernode only incurs a $+\poly(C)$ additive distortion, so in total only a constant additive distortion accumulates on the path between $u$ and $v$.
However, requiring that \emph{any} shortest path intersects only $O(1)$ spined supernodes is an incredibly strong requirement. In contrast, the shortcut partition from \Cref{SS:O(1)-emulator} essentially says that there is a partition of $G$ into spined supernodes such that a shortest path \ul{of length $O(C)$} intersects $O(1)$ spined supernodes; but for constructing an additive emulator, we would need such a guarantee to apply to \ul{arbitrarily long} paths.
\cite{nguyen2025asymptotic} show (in different language) that such a partition exists in the special case when $G$ has an $O(1)$-distortion emulator with bounded pathwidth, but in general it is easy to see that an arbitrary string graph $G$ may not admit such a partition, for example when $G$ is a grid graph.%
\footnote{In this case, every column is a shortest path, so it may intersect only $O(1)$ spined supernodes. But every row is also a shortest path --- and if every column intersects only $O(1)$ supernodes, it is not too difficult to see that some row must intersect polynomially many supernodes.}

Fortunately, it is often the case that planar graphs (or close-to-planar graphs) have extremely compact distance-sketching tools in the
$(1+\e)$-multiplicative distortion regime, even when there are lower bounds against exact distances sketches. We could still hope for a partition of $V(G)$ into spined supernodes such that, for any pair of vertices $u$ and $v$, there is a $(1+\e)$-approximate shortest path between $u$ and $v$ that intersects few supernodes.
We draw inspiration from a recent distance-approximating minor construction of Chang and Conroy~\cite{chang2025distance} for planar graphs: we interpret their result as constructing a collection of paths $\cO$ in a planar graph $G$ that $(1+\e)$-approximate all pairwise distances in $G$, such that the paths have few intersections with each other (in some amortized sense). Specifically:
for every pair of vertices $(u,v)$ in $G$, there is a path $P$ that is the concatenation of $\poly \log n$ subpaths $P_j$ in $\cO$, such that the sum of lengths of all $P_j$ is at most 
a $(1+\e)$ multiplicative factor longer than the original distance $\dist_G(u,v)$. %
We aim to generalize the techniques of \cite{chang2025distance} to find our desired partition of a string graph $G$ into spined supernodes. There are two key challenges we need to overcome.
\begin{enumerate}
    \item Can we adapt the techniques of \cite{chang2025distance} to find a suitable collection of paths $\cO$ when $G$ is a \emph{string graph}, not a planar graph?
    \item Can we obtain a collection of paths $\cO$ that are all (pairwise) at distance $\Omega(C)$ from each other, rather than just sparsely intersecting? If the paths intersect, even with few intersections as in~\cite{chang2025distance}, we cannot straighten them independently.
\end{enumerate}
As we describe below, we overcome these two challenges to prove the following key lemma.
\begin{restatable}{lemma}{brokenpaths}
\label{lem:string}
Let $\mathcal{G}$ be a class of unweighted graphs that is closed under induced subgraphs, such that every graph $G \in \mathcal{G}$ has a $C$-distortion planar emulator $H$ with $V(H) = V(G)$. 
    For any $G \in \mathcal{G}$ and any small constant $\e_0 > 0$, there is a set of paths $\cO$ in $G$ such~that:
    \begin{itemize}
    \item \textnormal{[$\Theta(C)$-disjointness.]} For any paths $P_1, P_2 \in \cO$, we have $\dist_G(P_1, P_2) > 2C$.
    \item \textnormal{[Shortest path.]} Every path $P \in \cO$ is a shortest path in the subgraph induced by $\cN_G(P, 2C)$, that is, the $2C$-neighborhood of $P$ in $G$.
    \item \textnormal{[Low-hop.]} For any vertices $u,v \in V(G)$, 
    there is a path $P$ between $u$ and $v$ (not necessarily in $\cO$):
    either $\norm{P} < O(C\e_0^{-4}\cdot \log^{17} n)$ (that is, $P$ is short), 
    or $\norm{P} \le (1+\e_0) \cdot \dist_G(u,v)$, such that 
    $P$ can be written as the concatenation
    \[
    P = Q_1 \circ P_1 \circ Q_2 \circ P_2 \circ \dots \circ Q_\ell
    \]
    where each $P_j$ is a subpath of some path in $\cO$, the sum of lengths of all $Q_j$ is at most $\e_0\cdot \dist_G(u,v)$,
    and $\ell \le O(\e_0^{-3} \cdot \log^{15} n)$. 
    \end{itemize}
\end{restatable}
\noindent Equipped with this lemma, we apply the path-straightening technique of \cite{nguyen2025asymptotic} to prove our \Cref{thm:main2}.

\paragraph{Challenge 1: Shortest-path separators on string graphs.}
The construction of \cite{chang2025distance} relies heavily on shortest-path separators in planar graphs, and we will also need this tool.
However, we have an immediate problem:
there are no known shortest-path separators construction for string graphs, even if we allow separators to be consist of the 1-neighborhood of shortest paths.
Such separators are only known for unit-disk graphs \cite{YXD12,HZ24}.
Our first challenge is to construct a shortest-path separator for arbitrary string graphs.

\begin{restatable}{lemma}{stringSeparator}
\label{lem:string-path-separator}
    Let $G$ be an $n$-vertex graph such that there is a $C$-distortion planar emulator $H$ for $G$ with $V(H) = V(G)$. For any set of non-negative vertex weights on $G$, there is a set $S \subseteq V(H)$ that comprises $O(\log n)$ shortest paths in $G$ and the $2C$-neighborhood of these paths in $G$, such that every connected component of $G \setminus S$ has at most half the weight of $G$.
\end{restatable}

Our proof combines the standard path separator construction on planar graphs by Lipton-Tarjan~\cite{lipton1979separator} with our $O(1)$-distortion planar emulator from \Cref{thm:main}.
First we compute a single-source shortest-path tree $T$ on string graph $G$, and consider the planar emulator $H$ of $G$.  
Map the tree $T$ into $H$ to obtain a subgraph $T^H$ 
(which is not necessarily a tree).
We then make use of a decycling trick to extract a spanning tree from $T^H$.
The following combinatorial lemma, to our best knowledge, appeared first in \cite{hathcock2025steiner} in which they dubbed it the \EMPH{Steiner path aggregation problem}.
Consider the set of $n$ paths $\cP$ from every vertex to root in $T$ and their corresponding images $\cP'$ in $H$.  
Given such paths in $H$  (where the paths are allows to self-intersect) that together covers $T^H$, 
there is a tree $T'$ that is a spanning subgraph of $T^H$, 
such that any leaf-to-root path in $T'$ is a concatenation of $O(\log n)$ paths from the set of paths $\cP'$.
By \cite{lipton1979separator}, some set of $O(1)$ paths in $T'$ form a balanced separator $S$ of $H$. The corresponding  preimage of $S$ is a set of $O(\log n)$ shortest paths in $G$, and one can show that the $C$-neighborhood of these $O(\log n)$ paths is a balanced separator of $G$.

With the separator hierarchy for the string graph $G$, we are able to carry out the first steps of the construction of \cite{chang2025distance} for computing a sparsely-intersecting set of paths $\cO$. 
First, for every piece $R$ in the separator hierarchy and every scale $i \in [\log n]$, we compute \EMPH{scale-$i$ portals} on the internal separator $S$ of $R$, with the property that consecutive scale-$i$ portals along $S$ are at distance $\Theta(\e\cdot 2^i)$.
The portals are chosen so that any shortest path between $(u,v)$ passing through $S$ can be rerouted through some portal on $S$ 
at a cost of a multiplicative $1+\e$ distortion.
Using the portals,
we then define a set of \EMPH{canonical paths}~\cite{abraham2012fully,chang2025distance} between the portals as follows. A pair of portals $p$ and $p'$ are \emph{relevant} to a piece $R$ in the separator hierarchy if $p$ is on the internal separator of $R$ and $p'$ is on the internal separator of some ancestor piece $R'$ of $R$. A \emph{canonical path} at scale $i$ and region $R$
is shortest path between a pair of scale-$i$ portals that are relevant to $R$ and are within distance $O(2^i)$ of each other; see \Cref{SS:canonical-pairs} for a formal definition.
One important property of canonical paths is that, between any pair of vertices $u$ and $v$, there exists a path $P$ with length at most $(1+\e) \cdot \dist_G(u,v)$, which is the concatenation of $\poly \log n$ canonical paths.
Moreover, \cite{chang2025distance} show that one can slightly ``wiggle'' the canonical paths so that they intersect sparsely with each other.

\paragraph{Challenge 2: How to ensure the spined supernodes are \boldmath{$O(C)$}-disjoint.}
With the analog of separator hierarchy for string graphs, we could obtain something similar to the set $\cO$ of sparsely-intersecting paths from \cite{chang2025distance}. However, here we come to another, more difficult, obstacle:  Recall that we have the ability to straighten \emph{one} path with its $O(C)$-neighborhood. 
However if two paths intersect (or even have overlapping $O(C)$-neighborhoods) %
the corresponding images of the two paths in $H$ could intersect, 
and as a result we may have two different values we want to reweight an edge of $H$ to.
This would be problematic.

Our main technical contribution in the second part of the paper
is to devise a novel detour strategy 
to resolve this issue and obtain a set $\cO$ of paths that are at least distance $\Omega(C)$ apart, proving \Cref{lem:string}.
The construction of $\cO$ is a simple \emph{carving algorithm}.
First compute the separator hierarchy and canonical paths, as described above. Order the canonical paths according to a natural ordering  $\preceq$ based on the region and scale of the canonical paths%
\footnote{We order the paths by region containment---that is, if $\pi$ is a canonical path between portals in region $R$, and $\pi'$ is a canonical path between portals in some ancestor region $R'$, we say $\pi \preceq \pi'$. Within a region, we order paths by length---that is, if $\pi$ is a path between two scale-$i$ portals and $\pi'$ is a path between two scale-$i'$ portals with $i \le i'$, we say $\pi \preceq \pi'$. The same ordering appears in \cite{chang2025distance}.}, and iterate through each path in $\cO$ based on such ordering. 
Every time before a new path $\pi$ is inserted in $\cO$, we ``erase'' an $\Theta(C)$-neighborhood around $\pi$ from the existing paths in $\cO$.  
If some path $\pi'$ breaks into multiple parts because of the erasure, remove $\pi'$ from $\cO$ and insert the components of the broken paths back into $\cO$.  
This procedure guarantees that subpaths of different paths $\pi$ and $\pi'$ are at distance $\Theta(C)$ from each other. However, subpaths coming from the same path $\pi$ could be close to each other --- to deal with this, we expand all the ``erased segments'' on path $\pi$ by distance $C$ along $\pi$.
The resulting collection of subpaths in $\cO$ are at distance more than $\Omega(C)$ away from each other.\footnote{In actuality, the procedure is slightly different: the $O(C)$-neighborhood is measured with respect to distances in a subgraph of $G$ that depends on the path $\pi$ chosen. Nevertheless we can still argue that the distance between subpaths in $\cO$ is more than $\Omega(C)$ away \ul{in $G$}. See \Cref{S:broken-paths}.}  See \Cref{fig:erased_paths} for an illustration.

\begin{figure}[th]
    \centering
    \includegraphics[width=0.7\linewidth]{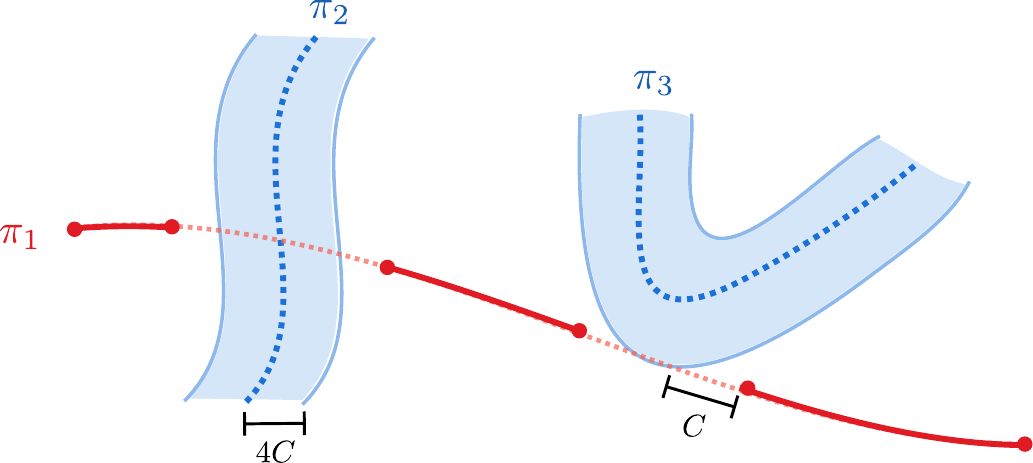}
    \caption{An illustration of the carving algorithm $\textsc{ErasedPath}$ construct a set of paths $\cO$. Paths $\pi_1\preceq \pi_2\preceq \pi_3$ are inserted one by one; the resulting subpaths of $\pi_1$ are shown in solid red lines.}
    \label{fig:erased_paths}
\end{figure}

Moreover, one can show that every $\pi \in \cO$ is a shortest path in the graph induced by the $\Theta(C)$-neighborhood of $\pi$ (which is a technical property we need in order to apply the path-straightening approach of \cite{nguyen2025asymptotic}).
To construct our mixed-distortion emulator, we reweight the emulator $H$ to straighten each path in $\cO$. It remains to prove the distortion bound.

\paragraph{Bound on distortion via detouring.}
The simplicity of the carving algorithm comes with a cost:
the analysis of the distortion bound is the most complicated part of the proof.
We outline some of the technical issues together with a sketch of 
how we overcome them.
We want argue that, between any pair of vertices $u$ and $v$, there is a $(1+\e)$-approximate shortest path in $G$ that walks along only $\tilde O(1)$ paths in $\cO$. (This is because, whenever we walk along paths in $\cO$ we incur no distortion in $H$, and whenever we switch between paths in $\cO$ we incur some additive distortion $+\poly(C)$ to ``jump'' the gap of width $O(C)$ between paths.
Thus it is crucial to bound the number of jumps we make.) By construction of canonical paths, it suffices to consider the case when $u$ and $v$ are the endpoints of some canonical path $\pi$.

\medskip \noindent \emph{Idea 1: Detouring $\pi$.} We would like to claim that $\pi$ gets broken into $\tilde O(1)$ paths in $\cO$, during the carving algorithm; then we could easily find a path between $u$ and $v$ that makes only $\tilde O(1)$ jumps between paths in $\cO$. Unfortunately, this claim is simply not true. Even if $\pi$ is only close (ie, within distance $C$) to a single other canonical path $\pi'$ with $\pi \prec \pi'$, it may be the case that $\pi$ is shattered into arbitrarily many subpaths; see \Cref{fig:detour}. However, in this case we can find a short path $P$ between $u$ and $v$ that makes few jumps with the following \EMPH{detour} procedure: walk along $\pi$ until the first point where $\pi$ and $\pi'$ are within distance $C$, then jump onto $\pi'$ and walk until the last point where $\pi$ and $\pi'$ are within distance $C$, and then jump back onto $\pi$. Assuming the entire path $\pi'$ appears in $\cO$, we have found a path $P$ between $u$ and $v$ that touches only two paths in $\cO$, and $P$ is only longer than $\dist_G(u,v)$ by an additive term $+2C$.

\begin{figure}[h!]
    \centering
    \includegraphics[width=0.6\linewidth]{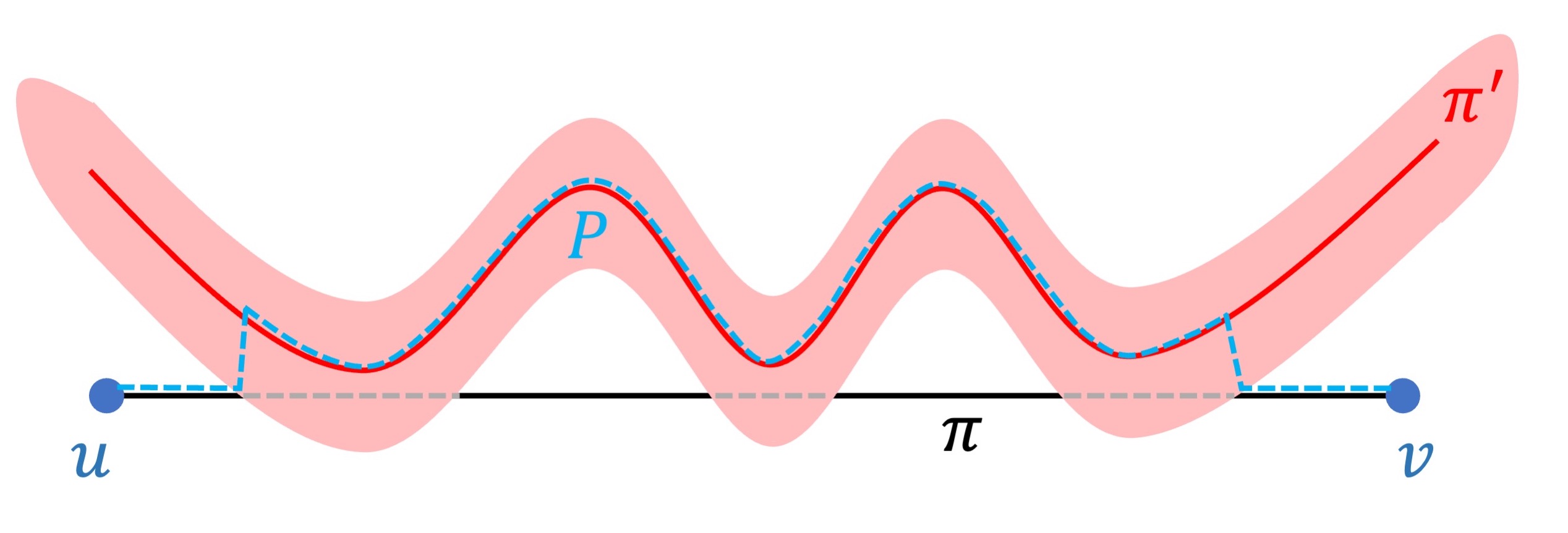}
    \caption{Canonical paths $\pi$ and $\pi'$, with $\pi \prec \pi'$. Erasing the $C$-neighborhood of $\pi'$ shatters $\pi$ into many subpaths. The ``detoured'' path $P$ jumps from $\pi$ to $\pi'$ and back to $\pi'$.}
    \label{fig:detour}
\end{figure}

\medskip \noindent \emph{Idea 2: Controlling the number of detours.} The argument above suggests that we would want to prove that there are only $\tilde O(1)$ paths $\pi'$ with $\pi \prec \pi'$ that are close to $\pi$. We say that $\pi'$ \EMPH{threatens} $\pi$ if $\pi \prec \pi'$ and $\pi'$ is within distance $C$ of $\pi$. If we could prove that $\pi$ is threatened by $\tilde O(1)$ canonical paths, we could hope to detour $\pi$ along each such threatening $\pi'$, ending up with a short path $P$ between $u$ and $v$ that makes $\tilde O(1)$ jumps. There are two issues here. 
First, it is just not true that $\pi$ is threatened by $\tilde O(1)$ canonical paths.
Second, even if we had a bound on the number of threateners of $\pi$, there is the issue that $\pi'$ itself could be threatened by some other canonical pair: that is, even if we detour along $\pi'$, the path $\pi'$ may not belong to $\cO$.  We now describe how we circumvent these two issues.

\begin{itemize}
    \item \emph{Issue: There are many threateners of $\pi$.}  On one hand, it is relatively straightforward to prove that there are only $\tilde O(1)$ paths $\pi'$ that threaten $\pi$ and satisfy $\norm{\pi'} > \e \cdot \norm{\pi}$; that is, there aren't many \emph{long} paths $\pi'$ that threaten $\pi$. (This will follow from the definition of canonical pair and the fact that scale-$i$ portals are spaced out from each other.) On the other hand, however, there could be many \emph{very short} paths $\pi'$ that threaten $\pi$.

    We overcome this issue as follows. To begin with, we compute a path $P$ between the endpoints of $\pi$ by performing detours along sufficiently long threatening paths $\pi'$, where $\len {\pi'} > \e \cdot \len \pi$. This means that we perform only $\tilde O(1)$ detours. Of course, at this point there may still be short paths $\pi'$ that threaten $\pi$, meaning $P$ could still make arbitrarily many jumps between paths in $\cO$.
However, in this case we argue that $P$ is within distance $\e \cdot \len \pi$ of a separator path of some piece $R'$ in the separator hierarchy, where $R'$ is a proper ancestor of the piece $R$ that contains $\pi$. 
Whenever our path $P$ is within distance $O(\e \cdot \len \pi)$ of a separator path $\pi_{R'}$ of some ancestor $R'$, we immediately ``jump'' to $\pi_{R'}$ and make a detour along $\pi_{R'}$; one can show that $\pi_{R'}$ is itself a canonical path that is not threatened by any short path $\pi'$ with an endpoint on $\pi_{R'}$. 
Thus, detouring along separator paths effectively allows us to ignore the short paths that threaten $\pi$, while only increasing the length of $P$ by a small amount.

\item \emph{Issue: Recursive threateners of $\pi'$.} 
Our next issue is that, even after we detour our path $P$ along a threatening path $\pi'$, we have no guarantee that $\pi'$ itself appears in $\cO$. We may have to recursively detour $\pi'$ along some other path $\pi''$ that threatens $\pi'$.
To overcome this issue, we would like to claim that if $\pi''$  threatens $\pi'$, then $\pi''$ also threatens $\pi$, under some relaxed notation of ``threatening''. 
We say a canonical path $\pi'$ \EMPH{$\beta$-threatens} a path $\pi$ if, roughly speaking\footnote{see \Cref{def:threatener} for the precise definition}, some endpoint of $\pi'$ is within distance $O(\len {\pi'}) + \beta$ of some endpoint of $\pi$. Clearly, if $\pi$ is threatened by $\pi'$, then $\pi$ is $0$-threatened by $\pi'$. As discussed in the previous section, there are only $\tilde O(1)$ long paths $\pi'$ that $0$-threaten $\pi$, and in fact a similar argument shows that there are only $\tilde O(1)$ long paths $\pi'$ that $\beta$-threaten $\pi$ for any constant $\beta = O(1)$. Moreover, the $\beta$-threatening relationship is almost transitive: if $\pi$ is $\beta$-threatened by $\pi'$, and $\pi'$ is threatened by another path $\pi''$, then $\pi$ is $(\beta+C)$-threatened by $\pi''$. That is, if we recursively detour $\pi$ along $\theta$ different paths $\pi'$, all these paths $\pi'$ are $(\theta \cdot C)$-threateners of $\pi$.

To bound the number of recursive detours we make, we want to bound the number of $\beta$-threateners of $\pi$. But, for what value of $\beta$ should we bound the number of $\beta$-threateners?
If we find that there are $f(\beta)$ many $\beta$-threateners of $\pi$, then we could make as many $f(\beta)$ detours, meaning that subsequent paths we detour against are $f(\beta) \cdot C$-threateners of $\pi$.
This ``self-referential'' problem between the number of detours and the radius we can claim $\pi$ threatens seems circular and  may escalate into an uncontrollable number of segments in the final detour path $P$.
Our one final trick is to establish the existence of a magic $\theta$ in \Cref{lem:self-referential-threat}, whose value is well-defined and non-circular, such that the number of (long) canonical paths that $\theta$-threaten $\pi$ is at most $\theta$.  The value of $\theta$ ends ups to be around $O(\e^{-3}\log^3 n)$, which is sufficient for an $(1+\e, +\poly\log n)$ approximation.
\end{itemize}
We remark that our detour procedure is somewhat reminiscent of the detours performed by \cite{chang2025distance}. However, there are substantial differences. The detours in \cite{chang2025distance} were directly performed to construct $\cO$; the detours were only performed against separator paths; and a main challenge was to detour in a way that deals with $O(\log n)$ scales. In contrast, our detours are performed only in the analysis (not in the construction of $\cO$); we detour along other canonical paths; and a key challenge is to cope with recursively detouring along canonical paths via the existence of the seemingly self-referential $\theta$.

\section{Preliminaries}
In this paper, we consider unweighted graph $G$. 
We let $E(G)$ denote the edges of $G$ and $V(G)$ denote the vertices of $G$. For a set of vertices $S\subseteq V(G)$, we will frequently consider the induced subgraph $G[S]$. The notation $\dist_G(u,v)$ denotes the shortest-path distance between two vertices $u$ and $v$ in $G$. We use $\diam(G)$ to denote the diameter of $G$.
For any set of vertices $X$ and any real number $r > 0$, the set $\cN_G(X,r)$ denotes $\set{v \in V(G) : \dist_G(X,v) \le r}$, called the $r$-neighborhood of $X$ in $G$.
In addition, we use the following notation for paths in graph $G$. For any two vertices $a$ and $b$ in a path $P$, we will use $P[a : b]$ to denote the subpath of $P$ between vertices $a$ and $b$. 
We will sometimes concatenate two paths $P_1 = [p_0, \dots, p_k]$ and $P_2 = [q_0, \dots, q_\ell]$ when either $p_k = q_0$ or $p_k$ is adjacent to $q_0$,
by joining the sequence of vertices in $P_1$ with those in $P_2$, suppressing the duplicate vertex when $p_k=q_0$,
and denote the resulting concatenated walk by $P_1 \circ P_2$. Note that $P_1 \circ P_2$ may not be a path even when $P_1$ and $P_2$ are.

The logarithms in this paper are all of base $2$.
For any positive integer $z$, we write $[z]$ to mean $\set{x \in \mathbb{Z}: 1 \le x \le z}$.

\section{O(1)-Distortion Planar Emulator}
\label{sec: O(1)}

\subsection{Proof of $O(1)$-Distortion Emulator}
\label{SS:const-emu}

Here we prove \Cref{thm:main}, assuming \Cref{lem:good-clustering} as stated above. We first recall the notion of region contact graph.
\begin{definition}
    Let $G$ be a graph, and let $\cR$ be a set of regions on $G$.
    The \EMPH{region contact graph} of $\cR$ on $G$, denoted \EMPH{$G_\cR$}, is the graph with vertex set $\cR$, where there is an edge between regions $R_1$ and $R_2$ if and only if there are vertices $v_1 \in R_1$ and $v_2 \in R_2$ such that either $v_1 = v_2$ or \ul{$v_1$ is adjacent to $v_2$ in $G$}.
\end{definition}
\begin{claim}
\label{clm:string-to-contact}
    Every $n$-vertex string graph $G^{\rm string}$ can be represented as a region contact graph $G_{\cR}$ on some planar graph $G$. Moreover, if there is a drawing of $G^{\rm string}$ with $N$ crossing points between strings, then $G$ has $O(N + n)$ vertices and can be computed from the drawing in $\poly(N + n)$ time.
\end{claim}
\begin{proof}
    Let $G^{\rm string}$ be a string graph, and fix some drawing of $G^{\rm string}$ with $N$ crossing points. Assume WLOG that every string intersects at least one other string, so $N \ge n$.
    As argued in \cite{lee2017separators}, the graph $G^{\rm string}$ is a region intersection graph of an $N$-vertex planar graph. Indeed,
    let $G_0$ be the planar graph with a vertex at every point where two strings intersect, and with an edge between two vertices that are adjacent on a string. Define a set of connected regions $\cR_0$ in one-to-one correspondence with the strings (specifically, the region associated with string $S$ contains the vertices of $G_0$ that correspond to intersection points involving $S$). The intersection graph of $\cR_0$ is precisely the string graph.
    
    We now obtain a region contact graph from the region intersection graph. Let $G$ be the $O(N)$-vertex graph obtained from $G_0$ by \emph{subdividing} every edge. We massage the regions $\cR_0$ into a set of regions $\cR$ on $G$ in the natural way: a subdivided vertex is in a region if and only if both end points were in that region. Clearly the contact graph $G_\cR$ is the same as the intersection graph of $\cR_0$.
\end{proof}

We remark that the reverse direction of \Cref{clm:string-to-contact} also holds (though we do not use this fact): one can convert a region contact graph on $G$ into an equivalent region intersection graph by subdividing every edge of $G$ and then expanding every region outward by one vertex; by \cite[Lemma 1.4]{lee2017separators}, this region intersection graph is a string graph.

\medskip
We now prove \Cref{clm:constant-emulator}: the planar emulator $H$ constructed in \Cref{SSS:constant-emulator-construction} has $O(1)$-distortion.

\constantEmulator*
\begin{proof}
Recall that the graph $H$ is defined by contracting every connected cluster $C\in \cC$ from \Cref{lem:good-clustering}, connecting adjacent clusters together, and adding edges from vertices representing regions to arbitrary representative clusters. The edges in $H$ have weight $\alpha_{\rm diam} + 1$. See \Cref{fig:partition} for an illustration.

First we prove an upper bound on distances in $H$. 
It suffices to show that for every edge $(R_a,R_b)$ in the region contact graph $G_\cR$, we have $\dist_H(R_a, R_b) \le (2\alpha_{\rm hop}+1) \cdot (\alpha_{\rm diam}+1)$.
To this end, fix one such edge $(R_a,R_b)$ in $G_\cR$. 
There are vertices $v_a \in R_a$ and $v_b \in R_b$ such that either $v_a = v_b$ or $v_a$ and $v_b$ are adjacent in $G$. 
Let $C_a \subseteq V(G)$ (resp.\ $C_b \subseteq V(G)$) be the representative cluster of $R_a$ (resp.\ $R_b$), as defined in the construction of $H$ in \Cref{SSS:constant-emulator-construction}. 
Let $x_a\in C_a\cap R_a$ and $x_b\in C_b\cap R_b$; these vertices exist because $C_a,C_b$ are the representative clusters of $R_a,R_b$.  Since $x_a,v_a\in R_a$, The [scattering] property of \Cref{lem:good-clustering} gives a path $Q_a$ in $G$ from $x_a$ to $v_a$ intersecting at most $\alpha_{\rm hop}$ clusters.  Similarly, since $v_b,x_b\in R_b$, there is a path $Q_b$ from $v_b$ to $x_b$ intersecting at most $\alpha_{\rm hop}$ clusters.  If $v_a=v_b$, concatenate $Q_a$ and $Q_b$; otherwise concatenate $Q_a$, the edge $v_av_b$, and $Q_b$.  The resulting walk has vertices in at most $2\alpha_{\rm hop}$ clusters, so it induces a walk from $C_a$ to $C_b$ in the cluster graph with at most $2\alpha_{\rm hop}-1$ cluster-cluster edges.  Adding the two terminal-cluster edges gives at most $2\alpha_{\rm hop}+1$ edges of $H$.
As each edge has weight $\alpha_{\rm diam}+1$, we have that $\dist_H(R_a, R_b) \le (2\alpha_{\rm hop}+1) \cdot (\alpha_{\rm diam}+1)$ as desired. 

Now we prove a lower bound on distances in $H$: for any $R_a,R_b \in \cR$, we claim $\dist_{G_\cR}(R_a,R_b) \le \dist_{H}(R_a,R_b)$. Let us consider two fixed regions $R_a,R_b \in \cR$ with $P_H \coloneqq [R_a, C_1, C_2, \ldots, C_\ell, R_b]$ as a shortest path in $H$ between $R_a$ and $R_b$; note that every vertex other than the endpoints represents some cluster $C_i \in \cC$, and cluster $C_1$ (resp.\ $C_\ell$) is the representative cluster of the region $R_a$ (resp.\ $R_b$). 
Observe that $P_H$ consists of $\ell+1$ edges of weight $\alpha_{\rm diam}+1$, so $P_H$ is a path of length $(\ell+1) \cdot (\alpha_{\rm diam}+1)$.
We aim to show that $\dist_{G_\cR}(R_a,R_b) \le (\ell+1) \cdot (\alpha_{\rm diam}+1)$, which we do by constructing a walk between $R_a$ and $R_b$ in $G_\cR$ with the desired length.
To this end, first observe that for every $i \in [\ell-1]$, the existence of the edge $(C_i, C_{i+1}) \in E(H)$ implies that there exist vertices $v_i' \in C_i$ and $v_{i+1} \in C_{i+1}$ such that $v_i'$ and $v_{i+1}$ are adjacent in $G$. We define region $\EMPH{$R_i'$} \in \cR$ to be some region that contains $v_i'$, and likewise define $\EMPH{$R_{i+1}$} \in \cR$ to be some region contains $v_{i+1}$; such regions exist because we assumed that $\cR$ collectively covers all vertices of $G$. The regions $R_i'$ and $R_{i+1}$ are adjacent in the contact graph $G_\cR$.
To simplify terminology, we additionally define $R_1 \coloneqq R_a$ and $R_\ell' \coloneqq R_b$, so that $R_i$ and $R_i'$ are defined for all $i \in [\ell]$. Observe that for every $i \in [\ell]$, the regions $R_i$ and $R_i'$ intersect the cluster $C_i$, and so the [diameter] property of \Cref{lem:good-clustering} implies that $\dist_{G_\cR}(R_i, R_i') \le \alpha_{\diam}$.
Consider the walk \EMPH{$P_{\mathrm {contact}}$} in $G_\cR$ between $R_a$ and $R_b$, formed by taking the concatenation of the shortest paths from $R_i$ to $R_i'$ for all $i \in [\ell]$, together with the connecting edges $(R_i', R_{i+1})$ for all $i \in [\ell-1]$.
The walk $P_{\rm contact}$ consists of $\ell$ shortest paths between regions intersecting the same cluster, plus $\ell-1$ edges between clusters; thus $P_{\rm contact}$ has length at most $\ell\cdot \alpha_{\rm diam} + (\ell - 1)$. 
As the walk $P_{\rm contact}$ is at least as long as the shortest path between $R_a$ and $R_b$ in $G_\cR$, we have that
\[  
\dist_{G_\cR}(R_a, R_b) 
\le \ell \cdot \alpha_{\rm diam} + (\ell-1)
< (\ell+1) \cdot (\alpha_{\rm diam} + 1) \le \dist_{H}(R_a, R_b), 
\]
which proves the lower bound on distances as desired.
\end{proof}

Assuming \Cref{lem:good-clustering}, our \Cref{thm:main} follows by combining our reduction to region contact graphs (\Cref{clm:string-to-contact}) with our planar emulator for region contact graphs (\Cref{SSS:constant-emulator-construction} and \Cref{clm:constant-emulator}).
Note that the emulator $H$ could have many more vertices than $G_\cR$. To fix this issue, we apply Steiner point removal \cite{CCLMST24}.
\begin{corollary}
\label{cor:small-emulator}
    There is an planar graph $H'$ such that $H'$ is a $O(1)$-distortion emulator for the graph $G_\cR$, and $V(H') = V(G_\cR)$. The emulator $H'$ has non-negative integer edge weights.
\end{corollary}
\begin{proof}
    Let $H$ be (edge-weighted) $O(1)$-distortion planar emulator of the contact graph $G_\cR$, provided by \Cref{thm:main}. It is immediate from the construction that every edge of $H$ has some non-negative integer edge weight. We define the \EMPH{terminals} to be the vertices $V(G_\cR) = \cR$, and define the \EMPH{Steiner vertices} to be $V(H) \setminus V(G_\cR) = \cC$; recall that the terminals represent regions in $\cR$, and the Steiner vertices represent clusters from the clustering $\cC$ of \Cref{lem:good-clustering}.
    Note that $H$ could contain many Steiner vertices, so we now apply the Steiner point removal procedure of \cite{CCLMST24} --- their theorem says that for any planar graph $H$ and set of terminals $T$, there is a minor $H'$ of $H$ with $V(H') = T$ and non-negative edge weights, such that for any two terminals $t_1, t_2$ we have $\dist_{H}(t_1, t_2) \le \dist_{H'}(t_1, t_2) \le O(1) \cdot \dist_H(t_1, t_2)$.
    Applying Steiner point removal on $H$ with terminal set $\cR$ produces a $O(1)$-distortion planar emulator for $G_\cR$. The construction of \cite{CCLMST24} guarantees that every \emph{edge} $(t_1,t_2)$ in $H'$ has weight $\dist_H(t_1,t_2)$, so the weight of every edge in $H'$ is a non-negative integer.
\end{proof}

\subsection{Towards Proving \Cref{lem:good-clustering}}
The remainder of \Cref{sec: O(1)} is dedicated to the proof of \Cref{lem:good-clustering}. We fix a plane graph \EMPH{$G$}.  All regions are subsets of $G$.

\paragraph{Shattering.} 
We will frequently require the operation of ``shattering'' a region; i.e.\ breaking it apart into smaller components. If $R$ is a region in $G$ and $S\subseteq V(G)$ is a set of vertices of $G$, then the process of \EMPH{shattering $R$ by the boundary of $S$}---denoted \EMPH{$\shatter R S$}%
---returns the set of connected components of $G[R \setminus S]$ together with the connected components of $G[R \cap S]$; see \Cref{fig:shatter}. 
In other words, shattering $R$ by $S$ returns a set of maximal connected regions which do not ``cross''  the boundary of $S$ and whose union is $R$. 
Suppose $R\in \cR$ is shattered by $S$ into connected components \EMPH{$R_1, R_2, \dots, R_\kappa$} $\in \shatter{R}{S}$, we say that each $R_i$ is a \EMPH{subregion} of $R$. The subregion property is transitive: if $R_i'$ is a subregion of $R_i$, we also say that $R_i'$ is a subregion of $R$.
If $\cR$ is a collection of regions, then \EMPH{$\shatter \cR S$} denotes $\bigcup\Set{\big. \shatter R S: R \in \cR}$. 
We say that a set of regions $\cR'$ is a \EMPH{shattering} of $\cR$ if it is a subset of some regions produced by an iterative sequence of \textsc{Shatter} operations run on $\cR$ (possibly with different $S$).

The ``shattering'' operation is the primary reason we work with the region contact graph $G_\cR$ instead of the region intersection graph. Suppose we shatter a region $R$ into multiple pieces, say $R = R_1 \sqcup R_2$. In the contact graph, $R_1$ and $R_2$ are adjacent, whereas in the intersection graph they are at infinite distance. In other words, distances in the contact graph behave more nicely under the shatter operation.

\begin{figure}[h!]
	\centering
	\subfigure[A set of regions $\cR$ on graph $G$, and the corresponding contact graph $G_\cR$.]
	{\scalebox{0.13}{\includegraphics{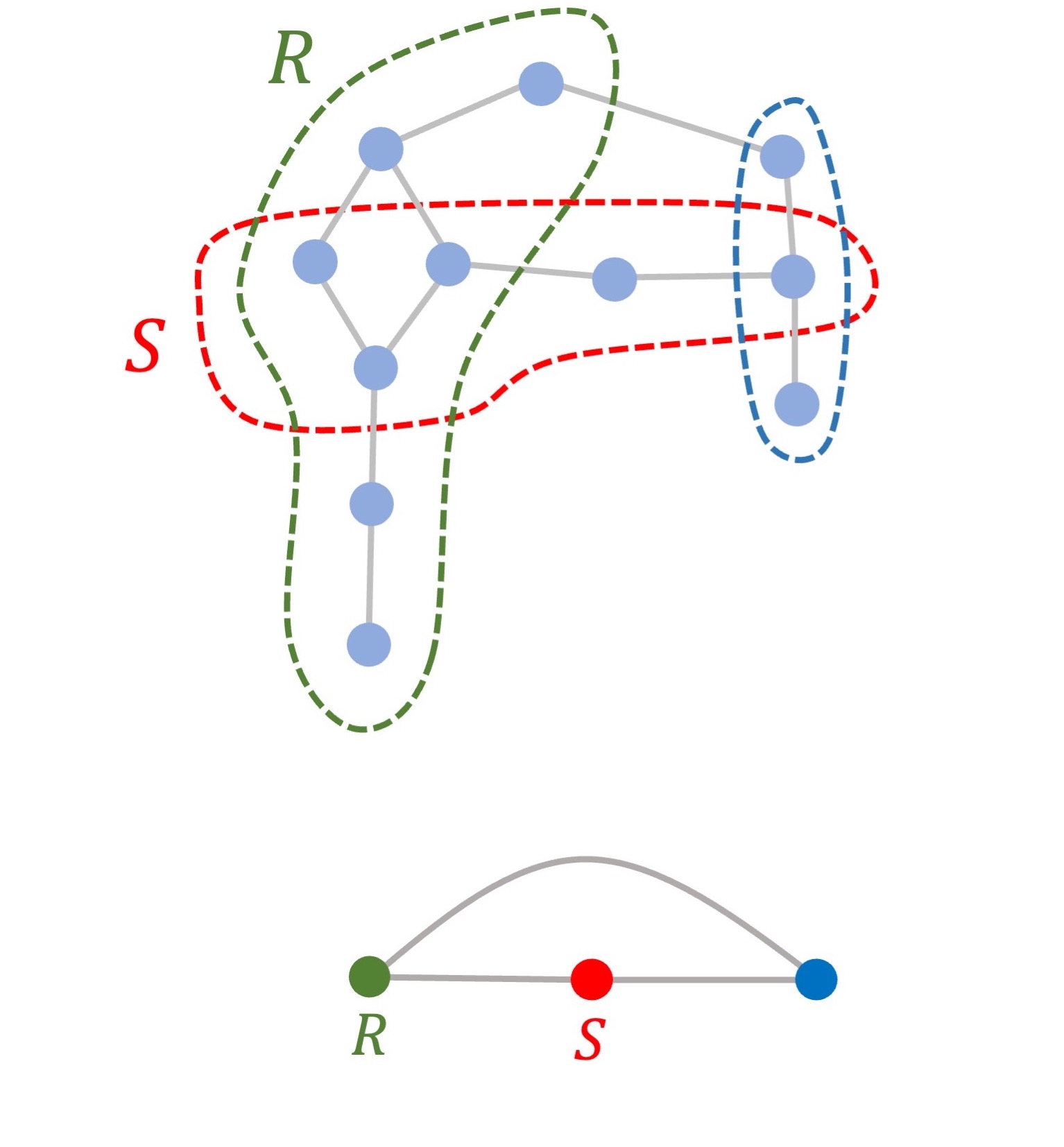}}}
	\hspace{1.0cm}
	\subfigure[The regions $\cR'$ produced by $\shatter{\cR}{S}$, and the resulting contact graph $G_{\cR'}$. The subregions of region $R$ are $R_1, R_2, R_3$.]
	{\scalebox{0.13}{\includegraphics{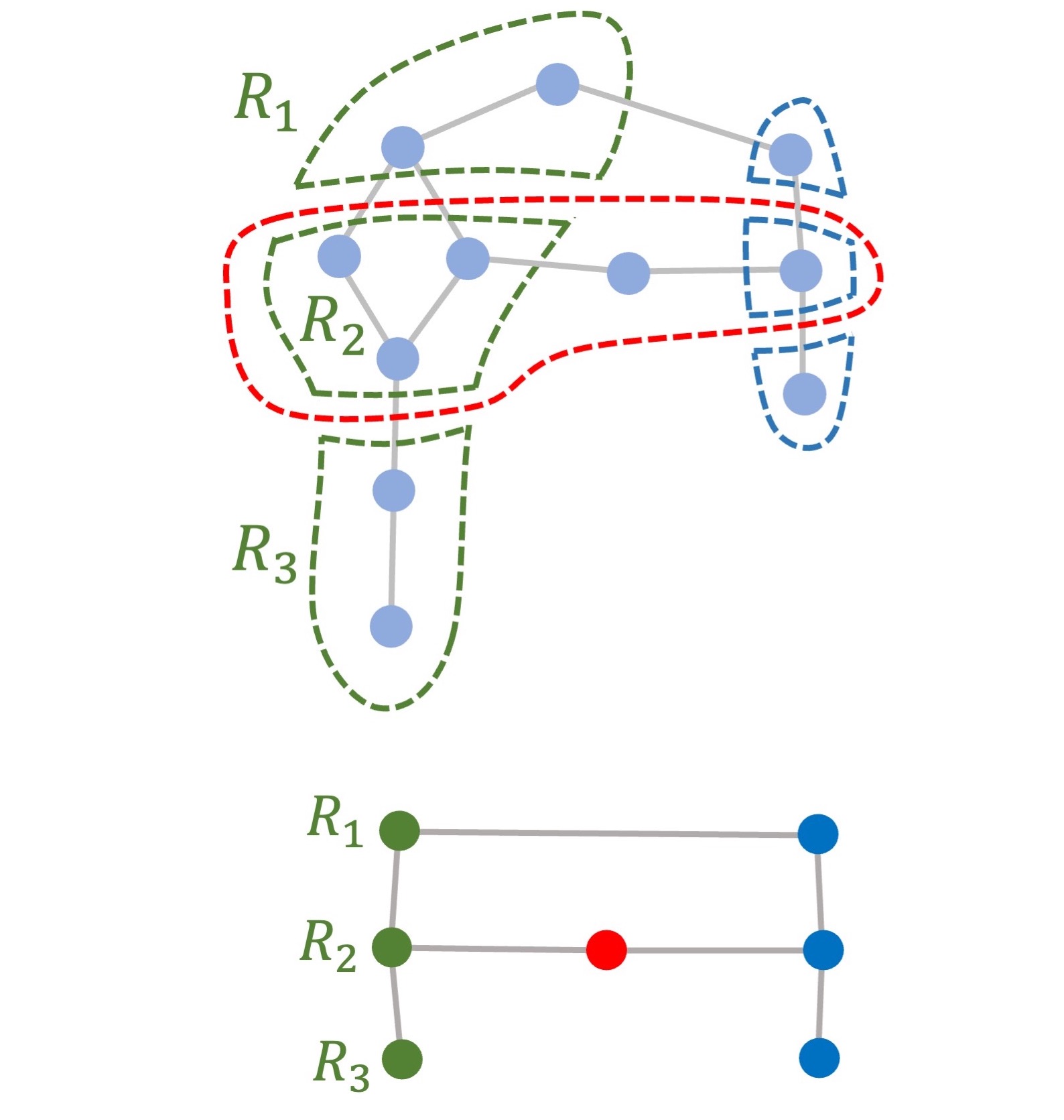}}}
	\caption{An illustration of regions and shattered regions.\label{fig:shatter}}
\end{figure}

\paragraph{Support and distances in the contact graph.} The \EMPH{support} of a set of regions $\cR$, denoted \EMPH{$V(\cR)$}, is the set of vertices in $G$ that are in some region in $\cR$.
As mentioned in \Cref{SSS:constant-emulator-construction}, given a contact graph $G_{\cR}$, we assume without loss of generality that $V(\cR) = V(G)$, since deleting vertices of $G$ that are not part of any region does not change the contact graph.

Fix any set of regions $\cR$. 
When the base graph $G$ is clear from context, we abuse notation and write \EMPH{$\dist_\cR(\cdot, \cdot)$} to be the distance function of the contact graph $G_{\cR}$; that is, for every pair of regions $R_1,R_2 \in \cR$, we define $\dist_\cR(R_1, R_2)$ to be the length of the shortest path between $R_1$ and $R_2$ in $G_{\cR}$.
In this section we frequently consider an induced subgraph $H$ of $G$, and a set of regions $\cR'$ with $V(\cR') = V(H)$, i.e. that the support of the regions of $\cR'$ is in $V(H)$.
Note that there is no difference between the graph $G_{\cR'}$ and $H_{\cR'}$, because $H$ is an induced subgraph of $G$. 
As such, the notation $\dist_{\cR}(\cdot, \cdot)$ is unambiguous, even if we do not explicitly specify whether $\cR$ should be understood to be a set of regions on $H$ or on $G$.

\paragraph{Clustering the outer regions.} 
Given a plane graph $H$ (a planar graph along with an embedding of $H$) and a set of regions $\cR$, we say a region $R \in \cR$ is an \EMPH{outer region} if $R$ contains at least one vertex on the outer face of $H$. 
We find the clustering of \Cref{lem:good-clustering} in stages: we first assign (some of) the vertices of $G$ to a (partial) set of clusters \EMPH{$\cC = \{C_1, ..., C_k\}$} which includes all vertices that are in outer regions, and then recurse on the subgraph induced by the unassigned vertices.
(A similar strategy was used by \cite{CCLMST23} to find a shortcut partition for planar graphs; they also create a partial set of clusters that cover the outer face and then recurse on the unassigned vertices.)
Here we state the lemma with the properties we need for our partial clustering. For notational convenience, we will denote the union of the vertices that are within our partial cluster as $\EMPH{$V(\cC)$} \coloneqq \bigcup_i C_i$.

\begin{lemma}
\label{lem:gridtree}
There are absolute constants \EMPH{$\alpha$} $\le 126$, \EMPH{$\beta$} $\le 57$, and \EMPH{$\gamma$} $\le 12$ such that the following holds. Let $G$ be a plane graph. There is a procedure \EMPH{$\textsc{ClusterOuter}(H, \cR)$} that takes as input a connected plane graph $H$ that is an induced subgraph of $G$ (and inherits the planar embedding 
of $G$) and a set of regions $\cR$ whose support is $V(H)$, and returns a partial partition of $V(H)$ into a set of clusters $\cC$ such that: 
\begin{enumerate}
    \item \textnormal{[Outer-clustered.]} For every outer region $R \in \cR$, every vertex in $R$ is assigned to some cluster of $\cC$. Moreover, the subgraph $H[V(\cC)]$, which contains all vertices assigned to clusters, is connected.
    \item \textnormal{[Diameter.]} Every cluster $C\in\cC$ has $\cR$-diameter at most $\alpha$.
    \item \textnormal{[Scattering.]} 
    For every region $R \in \cR$ and any two vertices $v_a, v_b \in R$,
    there is a path $\Pi$ in $G$ between $v_a$ and $v_b$ that has one of the following two forms: either $\Pi$ is the concatenation of at most $\beta$ paths
    \[\Pi = \Pi_1 \circ \Pi_2 \circ \ldots \circ \Pi_\beta\]
    where each subpath $\Pi_i$ is contained entirely in some cluster $C \in \cC$; or $\Pi$ is the concatenation of at most $2\gamma + 1$ paths
    \[
    \Pi = \Gamma_1 \circ \Pi_1 \circ \Gamma_2 \circ \Pi_2 \circ \ldots \circ \Pi_{\gamma} \circ \Gamma_{\gamma+1}
    \]
    where each $\Pi_i$ is contained entirely in some cluster $C\in \cC$,
    and each $\Gamma_i$ consists of unassigned vertices in $R$.
    Moreover, if $R$ is an outer region, then there is always a path $\Pi$ of the first type (ie, comprising $\beta$ subpaths contained in clusters).
\end{enumerate}
\end{lemma}

In the next section, \Cref{SS:hierarchy}, we explain how to apply \textsc{ClusterOuter} recursively to find a clustering of the entire graph $G$, proving \Cref{lem:good-clustering}.
The description of the procedure \textsc{ClusterOuter} and proof of \Cref{lem:gridtree} occupies \Cref{SS:select-paths,SS:expand}.
Namely, \Cref{SS:select-paths} describes a helper procedure \textsc{SelectPaths} that selects a set of \emph{supernodes} by sweeping across the outer face of $G$; \Cref{SS:expand} describes the procedure \textsc{ClusterOuter} which starts with the supernodes of \textsc{SelectPaths}, then (simultaneously) expands the supernodes and partitions them into clusters.

\subsection{\textsc{Cluster}: Recursively applying \textsc{ClusterOuter}}
\label{SS:hierarchy}

We describe the procedure \textsc{Cluster}.

\begin{tcolorbox}
\ul{$\textsc{Cluster}(H, \cR)$}:

\textbf{Input:} plane induced subgraph $H$, set of regions $\cR$ whose support is $V(H)$

\textbf{Output:} partition of $V(H)$ into clusters\\

    Let $\EMPH{$\cC^*$} \gets \textsc{ClusterOuter}(H, \cR)$. Let $\EMPH{$U$} \subseteq V(H)$ denote the set of vertices that are unassigned to any cluster in $\cC^*$.
    Let $\EMPH{$\cR^*$} \gets \shatter{\cR}{U}$.
    
We now recurse on $U$: for each connected component $H_1, H_2, \dots H_\ell$ of $H[U]$, let \EMPH{$\cR_i$} denote the subset of regions in $\cR^*$ that intersect $H_i$, and let $\EMPH{$\cC_i$} \gets \textsc{Cluster}(H_i, \cR_i)$.
    
Return $\cC \gets \cC^* \cup \cC_1 \cup \cC_2 \cup \cdots \cup \cC_\ell$, the union of all computed clusters.
\end{tcolorbox}

To prove \Cref{lem:good-clustering}, we run $\textsc{Cluster}(G, \cR)$ and claim that the output set of clusters $\cC$ satisfies the [diameter] and [scattering] properties.
We begin with a simple observation that helps us prove the [diameter] property.
\begin{observation}
\label{obs:shattered-diam}
    Let $\cR$ be a set of regions on the planar graph $G$, let $H$ be a subgraph of $G$, and let $\cR'$ be a shattering of $\cR$ whose support contains $V(H)$. Let $C$ be a cluster in $H$. If $C$ has $\cR'$-diameter $\le \alpha$, then $C$ has $\cR$-diameter $\le \alpha$.
\end{observation}
\begin{proof}
Let $R_a, R_b \in \cR$ be two arbitrary regions that intersect $C$; we want to show $\dist_{\cR}(R_a, R_b) \le \alpha$.
Let $r_a \in R_a$ and $r_b \in R_b$ be two vertices in $C$.
Let $R'_a \in \cR'$ and $R'_b \in \cR'$ be two regions that contain $r_a$ and $r_b$, respectively; these regions exist because the support of $\cR'$ contains $V(H)$, and $C \subseteq V(H)$.
As $C$ has $\cR'$-diameter at most $\alpha$, there is a path $\Gamma = [R_0', R_1', \ldots, R_{\ell}']$ between $R_a'$ and $R_b'$ in the contact graph of $\cR'$of length $\ell \le \alpha$.
Now, for each $R_i'$ in $\Gamma$ (for $1 \le i \le \ell - 1$), let $R_i \in \cR$ be an arbitrary region in $\cR$ that contains $R_i'$; such a region exists because $\cR'$ is a shattering of $\cR$.  
Then, $[R_a = R_0, R_1, \ldots, R_{\ell} = R_b]$ is a walk in the contact graph of $\cR$ of length $\ell$. Thus $\dist_{\cR}(R_a, R_b) \le \ell \le \alpha$.
\end{proof}

Consider the output $\cC$ of $\textsc{Cluster}(G, \cR)$.
Every cluster $C \in \cC$ is created, at some point in the recursive execution, by some call to $\textsc{ClusterOuter}(H_i, \cR_i)$: in this case $\cR_i$ is a shattering of $\cR$ that covers $H_i$, and $C \subseteq V(H)$. Thus, \Cref{lem:gridtree}(2) and \Cref{obs:shattered-diam} implies:
\begin{observation}
\label{obs:hierarchy-diam}
    Let $\cR$ be a set of regions covering the planar graph $G$.
    Every cluster output by the procedure $\textsc{Cluster}(G, \cR)$ has $\cR$-diameter at most $\alpha$.
\end{observation}

It remains to prove the [scattering] property of \Cref{lem:good-clustering}. The key insight is that, if a region $R$ gets some of its vertices assigned by some call to $\textsc{Cluster}$, all of its remaining vertices will be assigned during the next round of recursive calls to $\textsc{Cluster}$.

\begin{figure}[th]
    \centering
    \includegraphics[width=0.7\linewidth]{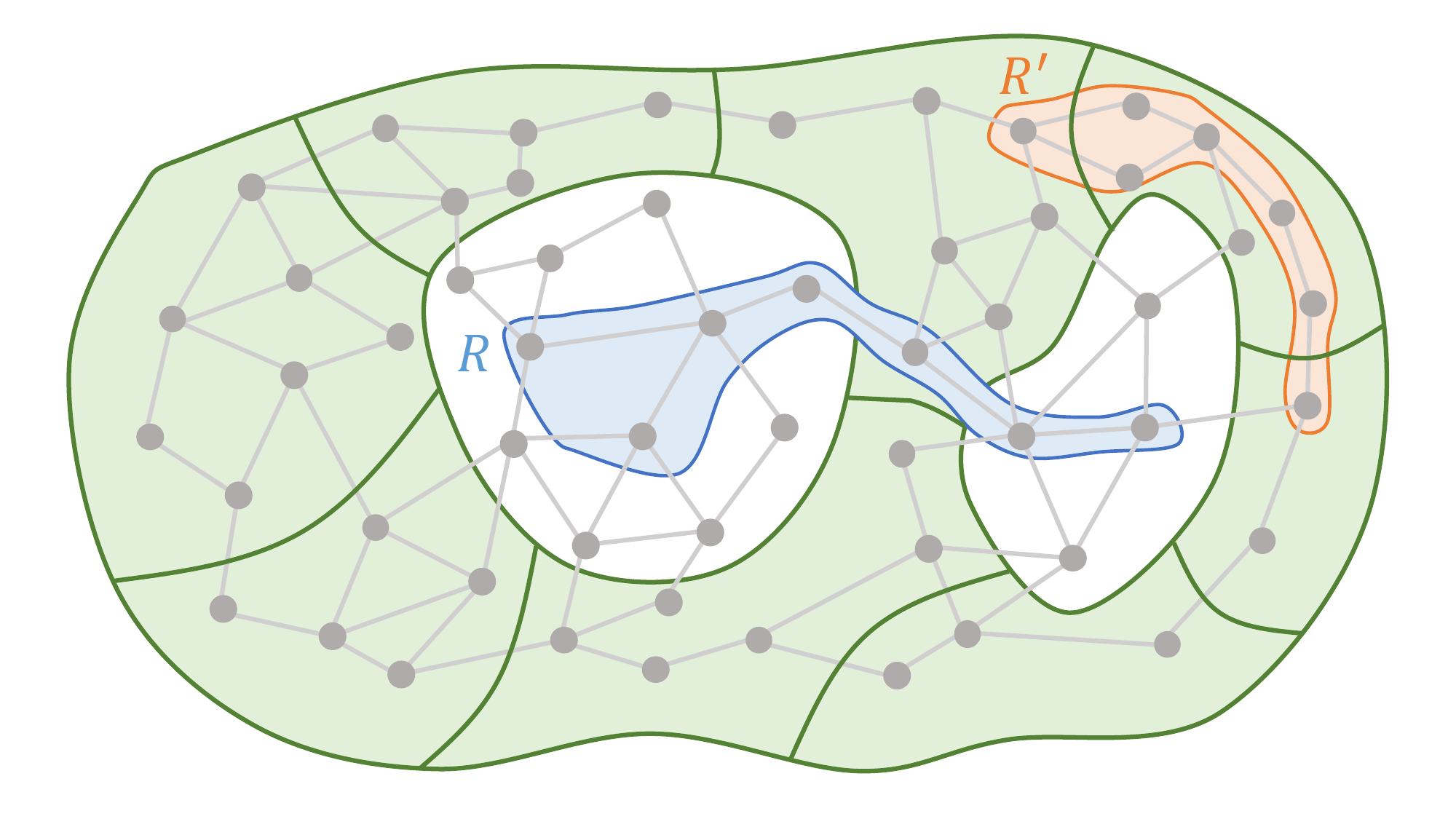}
    \caption{A depiction of the proof of \Cref{lem:hierarchy-scattering}. There is a plane graph $G$, and a set of clusters $\cC$ (from \Cref{lem:gridtree}) drawn in green. The orange region $R'$ is an outer region in $G$ and is contained entirely in $V(\cC)$. The blue region $R$ is \emph{not} an outer region and is \emph{not} contained entirely in $V(\cC)$; observe that after shattering $\cR^* \gets \shatter{R}{(V(\cC))}$, every subregion of $R$ in $\cR^*$ \emph{is} an outer region in the subgraph induced by the vertices not in $V(\cC)$.}
    \label{fig:clusterOuter}
\end{figure}
\vspace{-9pt}

\begin{lemma}
\label{lem:hierarchy-scattering}
Let $\cR$ be a set of regions on the planar graph $G$. Let $\cC$ be the set of clusters output by $\textsc{Cluster}(G, \cR)$.
Let $\beta$ and $\gamma$ be the constants from \Cref{lem:gridtree}.
For any region $R \in \cR$ and any $v_s, v_t \in R$, there is a path in $G$ between $v_s$ and $v_t$ that intersects at most $\gamma+(\gamma+1)\beta$ different clusters of $\cC$.
\end{lemma}
\begin{proof}
    Vertices in $G$ are gradually assigned to clusters in recursive calls to $\textsc{ClusterOuter}$.
    To prove the lemma, we strengthen the statement and proceed by induction on the depth of the recursion tree.
    \begin{quote}
        \textbf{Claim.} Let $\textsc{Cluster}(\EMPH{$H$}, \EMPH{$\cP$})$ be a recursive call made during the execution of $\cC \gets \textsc{Cluster}(G, \cR)$.
        For every region $\EMPH{$R$} \in \cP$ that is contained in $V(H)$\footnote{in other words, all vertices of $R$ are unassigned at the start of the call to $\textsc{Cluster}(H, \cP)$} and for every $v_s, v_t \in R$, the following statement holds:
        If $R$ is an \emph{outer} region of $H$, then there is a path between $v_s$ and $v_t$ in $H$ that intersects at most $\beta$ different clusters of $\cC$.
        Otherwise, there is a path between $v_s$ and $v_t$ in $H$ that intersects at most $\gamma+(\gamma+1)\beta$ different clusters of $\cC$.
    \end{quote}
    Proving this claim is sufficient to prove the lemma, as it must hold for $H = G$ and $\cP = \cR$. We prove the claim by induction. Assume the claim is true for all recursive calls to $\textsc{Cluster}$ made by $\textsc{Cluster}(H, \cP)$.
    Recall that $\textsc{Cluster}(H, \cP)$ consists of computing $\EMPH{$\cC^*$} \gets \textsc{ClusterOuter}(H, \cP)$ and then recursing on the unassigned vertices. If no vertex of $R$ is assigned in $\cC^*$, then $R$ appears in some recursive call, and the claim holds by induction. If some vertex of $R$ is assigned in $\cC^*$, then we divide into two cases.

    First suppose that $R$ is an outer region of $H$. 
    By \Cref{lem:gridtree}(3), there is a path $\Pi$ between $v_s$ and $v_t$ in $H$
    that is contained in the union of at most $\beta$ clusters in $\cC^*$; thus $\Pi$ passes through at most $\beta$ clusters in the final clustering $\cC$ as desired.

    On the other hand, suppose that $R$ is not an outer region of $H$. By \Cref{lem:gridtree}(3), there is a path $\Pi$ between $v_s$ and $v_t$ in $H$ with one of two forms. In the first case, $\Pi$ is contained in the union of at most $\beta$ clusters in $\cC^*$, and we are done. In the second case, we can write
    $\Pi = \Gamma_1 \circ \Pi_1 \circ \Gamma_2 \circ \Pi_2 \circ \ldots \circ \Pi_{\gamma} \circ \Gamma_{\gamma+1}$
    where each subpath $\Pi_i$ is contained entirely in some cluster $C\in \cC$,
    and each $\Gamma_i$ consists of unassigned vertices in $R$.
    Denote the endpoints of the subpath $\Gamma_i$ as \EMPH{$s_i$} and \EMPH{$t_i$}.
    We will show that $\Gamma_i$ intersects at most $\beta$ clusters in $\cC$; we can therefore conclude that $\Pi$ intersects at most $\gamma + (\gamma+1)\beta$ clusters in $\cC$.
    
    $\Gamma_i$ consists of unassigned vertices in $R$. Let \EMPH{$U$} denote the set of unassigned vertices, and let $\cP^* \gets \shatter{\cP}{U}$.
    Because $\Gamma_i \subseteq U \cap R$, there is some (connected) region $\EMPH{$R_i$} \in \cP^*$, such that $R_i$ is a subregion of $R$ and $R_i$ contains both $s_i$ and $t_i$. 
    The region $R_i$ is contained in the set of regions $\cP_i \subseteq \cP^*$ passed to some recursive call $\cC_i \gets \textsc{Cluster}(H_i, \cP_i)$ on the unassigned regions. Moreover, we claim $R_i$ is an outer region of the graph $H_i$ (when $H_i$ inherits its embedding from $H$). Indeed, because $R_i$ is a subregion produced by shattering $R$ with the boundary of $U$, and at least one vertex of $R$ is in $V(\cC^*)$ (that is, not in $U$), there is some vertex $r \in R_i$ that is adjacent to a vertex in $V(\cC^*)$; now, \Cref{lem:gridtree}(1) implies that $H[V(\cC^*)]$ is connected and contains an outer vertex in $H$, and so we conclude that $r$ is an outer vertex of $H_i$.
    We apply induction hypothesis to conclude that there is a path between $s_i, t_i \in R_i$ in $H_i$ that is contained in $\beta$ clusters of $\cC_i$; this same path is a path in $H$ that passes through $\beta$ clusters of the final clustering $\cC$, and we are done.
\end{proof}
\Cref{lem:good-clustering} follows from \Cref{obs:hierarchy-diam} (the diameter bound) and \Cref{lem:hierarchy-scattering} (the scattering bound).

\subsection{\textsc{SelectPaths}: Walking along the outer face}
\label{SS:select-paths}
The next two sections are dedicated to the description of the procedure \textsc{ClusterOuter} and the proof of \Cref{lem:gridtree}. This section is dedicated to a helper procedure, \textsc{SelectPaths}. We sweep along the outer face of $G$ and select a set of \emph{supernodes}.
The sweeping procedure is roughly analogous to $\textsc{SelectPaths}$ procedure of \cite{CCLMST23} and the \textsc{Depth-Cover} procedure of \cite{BLT14}.

A key lemma for the sparse cover construction of \cite{BLT14} and the gridtree of \cite{CCLMST23} is that one can select the ``next path'' by walking along the outer face, which separates the remaining graph from the previous graph. We will use a similar lemma. We introduce some (new) terminology, and stating a rephrasing of their result in our terminology.

Let \EMPH{$H_0$} be a plane graph, fixed throughout this section. Let $H$ be a subgraph of $H_0$, and let $W \subseteq V(H_0)$ be a connected set of vertices, that is, $H_0[W]$ is connected. We say that $H$ is \EMPH{outer-bounded} by $W$ in $H_0$ if $H$ is induced by a connected component of $H_0 \setminus W$, and $H$ and $W$ both contain at least one outer vertex of $H_0$. 
(See \Cref{fig:outer-bounded} for an example.)
Observe that the definition of outer-bounded implies that all edges in $H_0$ with one end point in $V(H)$ has the other end point in either $V(H)$ or $W$. 
The \EMPH{critical vertices} of $(H,W)$ are the vertices of $H$ that are connected to a vertex of $W$ by an edge on the outer face of $H_0$. 
If $W = \varnothing$ and $H = H_0$, we abuse terminology and also say that $H$ is outer-bounded by $W$; in this case, we define the critical vertex of $(H,W)$ to be an arbitrary vertex on the outer face of $H_0$.

\begin{figure}[t]
	\centering
	\subfigure[Subgraph $H$ outer-bounded by $W$ in $H_0$, with 2 critical vertices shown in red (as described in \Cref{lem:blt}(1)).]
	{\scalebox{0.064}{\includegraphics{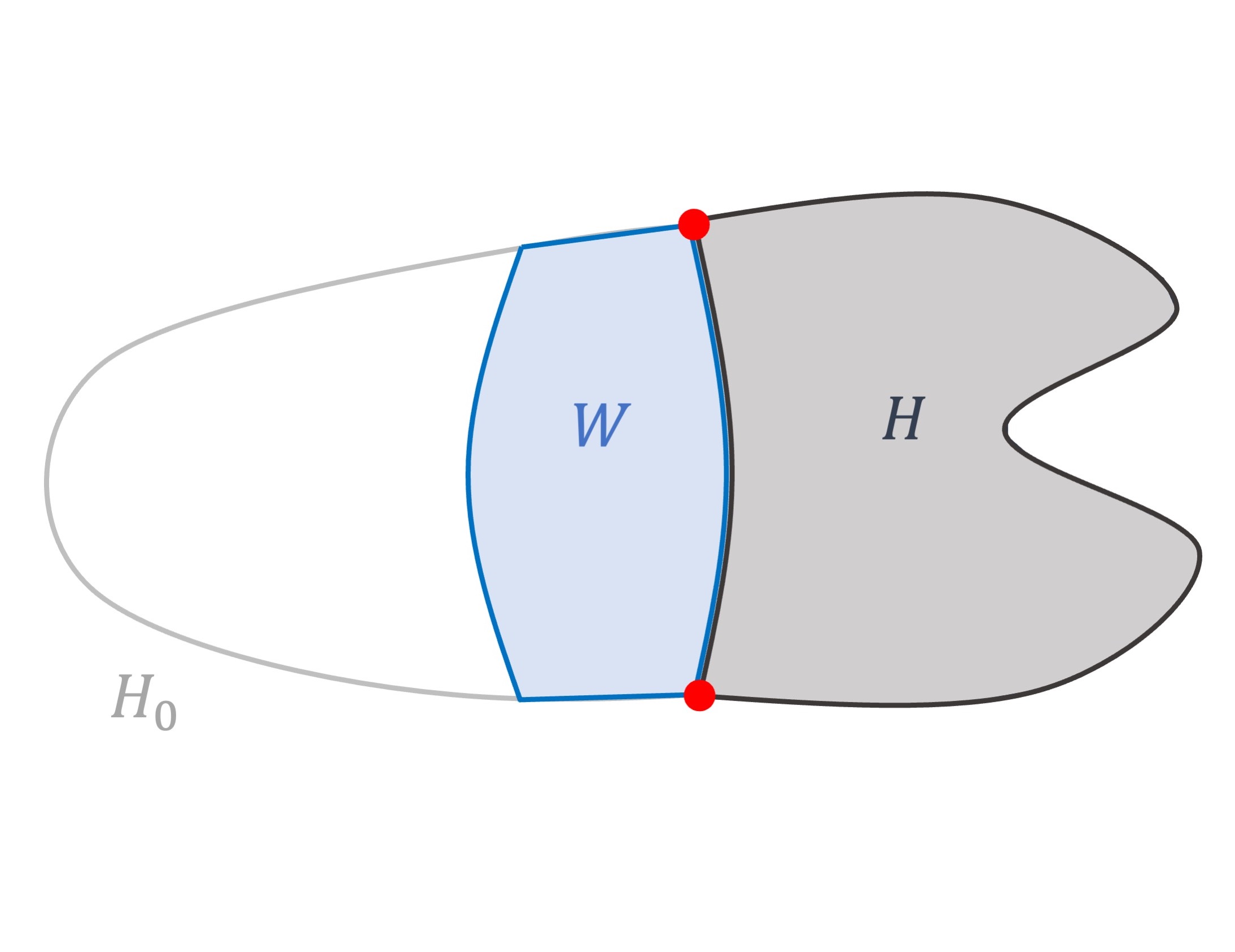}}}
	\hspace{0.2cm}
	\subfigure[The path $\pi$ separates $W$ from the outer vertices of $H_0$ in $H$ (as described in \Cref{lem:blt}(2)).]
	{\scalebox{0.064}{\includegraphics{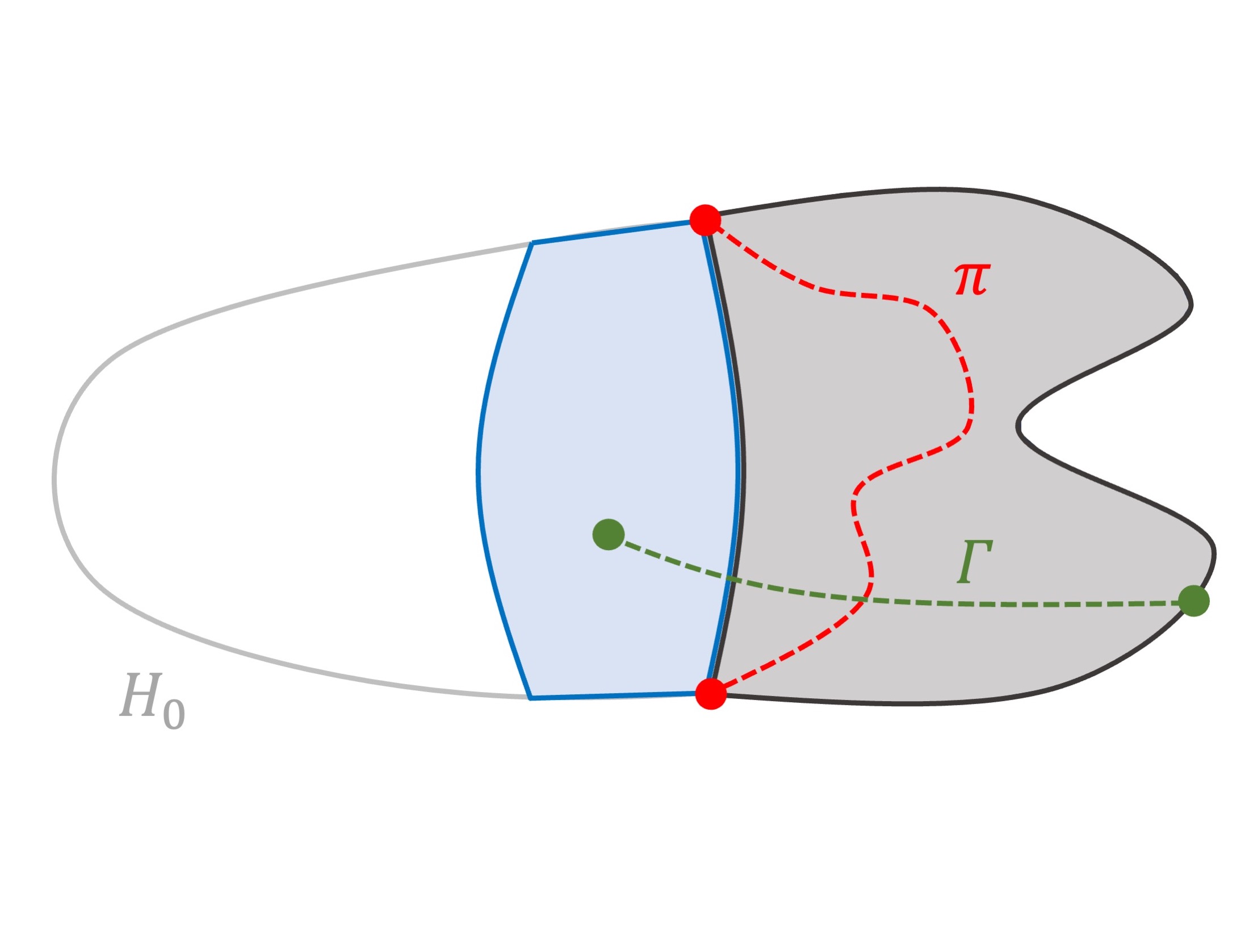}}}
	\hspace{0.2cm}
	\subfigure[A connected component $H'$ of $H \setminus W'$ is outer-bounded by $W'$ in $H_0$ (as described in \Cref{lem:blt}(3)).]
	{\scalebox{0.064}{\includegraphics{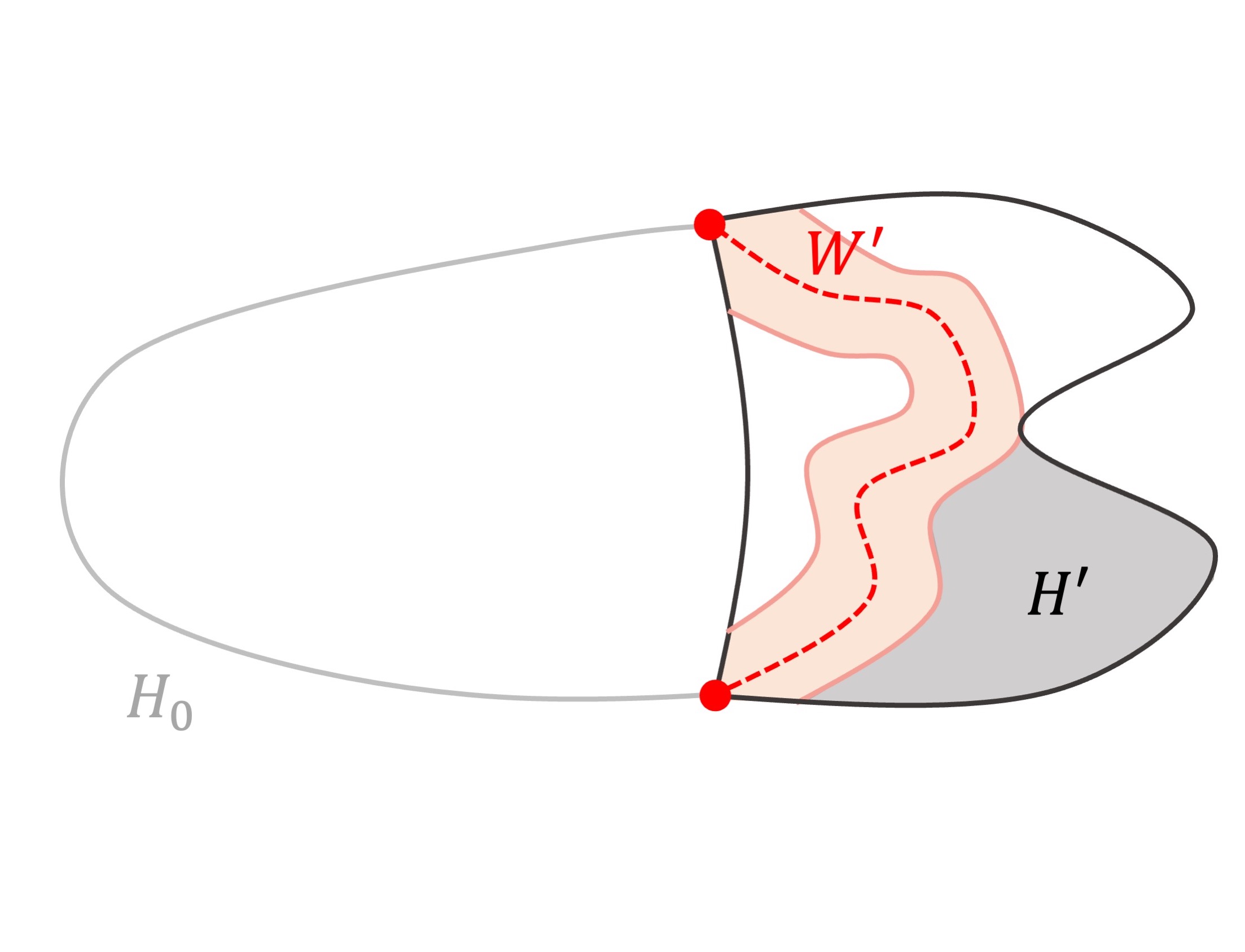}}}
	\caption{An illustration of a subgraph outer-bounded by a connected set of vertices.\label{fig:outer-bounded}}
\end{figure}

The following lemma is implicit in \cite[Sec.~5.3]{BLT14}. 
 For completeness, we give a proof in \Cref{ap:bltproof}.
We emphasize that this lemma refers to the planar (base) graph, not the region contact graph.

\begin{restatable}{lemma}{bltlemma}
\label{lem:blt}
    Let \EMPH{$H$} be a subgraph of the plane graph $H_0$, outer-bounded by some set \EMPH{$W$} in $H_0$. 
    \begin{enumerate}
        \item[(1)] $(H,W)$ has either 1 or 2 critical vertices.
        \item[(2)] Let \EMPH{$\pi$} be an arbitrary path between the critical vertices of $(H,W)$. If \EMPH{$\Gamma$} is a path in $H_0$ between a vertex in $W$ and a vertex in $V(H)$ that is an outer vertex of $H_0$, then $\Gamma$ intersects $\pi$. 
        \item[(3)] Let $\EMPH{$W'$} \subseteq V(H)$ be a connected subset of vertices that includes $\pi$. Let \EMPH{$H'$} be a subgraph induced by a connected component of $H \setminus W'$. If $H'$ contains at least one outer vertex of $H_0$, then $H'$ is outer-bounded by $W'$ in $H_0$.
    \end{enumerate}
\end{restatable}

Equipped with this lemma, we now describe the \textsc{SelectPaths} procedure for selecting a set of supernodes, a key subroutine for the \textsc{ClusterOuter} procedure in the next section.

\begin{tcolorbox}[breakable]

\ul{$\textsc{SelectPaths}_{H_0}(H, W_{\rm old}, \cR)$}:

\textbf{Input:} an implicit plane graph $H_0$, a subgraph $H$ outer-bounded by $W_{\rm old}$ in $H_0$, and a set of regions $\cR$ whose support is $V(H)$

\textbf{Output:} a shattering $\cR^*$ of $\cR$, and a partial partition of that shattering into \emph{supernodes} $\cS$

\begin{enumerate}
    \item \emph{Select spine $\pi$.}

    Let $Y \subseteq H$ be the set of critical vertices of $(H,W_{\rm old})$
    (by \Cref{lem:blt}(1), $1 \le |Y| \le 2$). For each $y \in Y$, let $R_y \in \cR$ be an arbitrary region that contains $y$. Define the \EMPH{spine} \EMPH{$\pi$} to be a shortest path in the contact graph of $\cR$ between $\set{R_y}_{y \in Y}$ (if $|Y| = 1$, then $\pi$ is the trivial path $[R_y\!]$).
    
    \item \emph{Initialize supernode $S$, and shatter the neighborhood of $\pi$ to obtain $\cR'$.}
    
    Initialize $\EMPH{$S^{\rm init}$} \subseteq \cR$ to be $\Set{\big. R \in \cR : \dist_{\cR}(R, \pi) \le 1}$.
    Let $\EMPH{$\cR'$} \gets \shatter{\cR}{(V(S^{\rm init}))}$.
    Define the \EMPH{supernode} $\EMPH{$S$} \subseteq \cR'$ to contain all regions in $\cR'$ that intersect $V(S^{\rm init})$. (Note that $V(S) = V(S^{\rm init})$. 
    The supernode $S$ contains all the regions in $S^{\rm init}$
    plus some regions created by the shattering.)
    
    \item \emph{Recurse on each connected component.}

    Iterate over the $\kappa$ connected components $H_1, H_2, \dots, H_\kappa$ of $H \setminus V(S)$.
    Consider the connected component \EMPH{$H_i$}. 
    Let $\EMPH{$\cR_i$} \subseteq \cR'$ be the set of regions that intersect $H_i$.
    If $H_i$ does not contain any outer vertex of $H_0$,
    then set $\cR^\out_i\coloneqq \cR_i$ and
    $\cS_i\coloneqq \varnothing$.
    Otherwise, recursively compute $(\cR^\out_i, \cS_i) \gets \textsc{SelectPaths}_{H_0}(H_i, V(S), \cR_i)$.

    \item \emph{Return.}

    Define $\EMPH{$\cR^\out$} \coloneqq S \cup \bigcup_{i \in [\kappa]} \cR^\out_i$ to be a shattering of $\cR$, and define $\EMPH{$\cS$} \coloneqq \set{S} \cup \bigcup_{i \in [\kappa]} \cS_i$ to be a partial partition of $\cR^\out$. Return $(\cR^\out, \cS)$.
\end{enumerate}

\end{tcolorbox}

To begin, we initialize $H \gets H_0$ and $W_{\rm old} \gets \varnothing$. For the rest of the section, we let $H_0$ and $\cR_0$ be a fixed graph and set of regions, and let $(\EMPH{$\cR^*$}, \EMPH{$\cS$})$ be the output of $\textsc{SelectPaths}(H_0, \varnothing, \cR_0)$.
The set $\cS$ returned by $\textsc{SelectPaths}$ is called the set of \EMPH{supernodes}---each supernode $S$ is a set of (possibly shattered) regions, and is associated with a spine $\pi$ which is a path in the contact graph of $S$.

\subsubsection{Supernode structure: partial order and domains}

We begin by defining some terminology for the structure of the supernodes $\cS$.
We will prove an assorted collection of basic properties about the supernodes created by $\textsc{SelectPaths}$.
The motivation for these lemmas will be clear when we prove the [diameter] and [scattering] properties of \Cref{lem:gridtree}.
These lemmas will all follow from the definition of \textsc{SelectPaths} and \Cref{lem:blt}(3).

There is a natural partial order to the supernodes $\cS$, based on the order they were constructed in the recursive call to $\textsc{SelectPaths}(H_0, \varnothing, \cR_0)$. Indeed, consider the recursion tree of $\textsc{SelectPaths}$. We say supernode $S_1$ is an \EMPH{ancestor} of $S_2$ if $S_1$ was initialized during some call to $\textsc{SelectPaths}$ that is an ancestor (in the recursion tree) of the call to $\textsc{SelectPaths}$ that created $S_2$. We define \EMPH{descendants} and the \EMPH{parent} of a supernode similarly.
For any supernode $S \in \cS$ created by some recursive call $\textsc{SelectPaths}(H, W_{\rm old}, \cR)$, 
define the \EMPH{domain} of $S$, denoted \EMPH{$\dom(S)$}, to be the set of regions $\cR^*$ in the output that intersect $H$. 
\begin{observation}
    For any two supernodes $S, S_a \in \cS$, if $S_a$ is an ancestor of $S$, then $\dom(S) \subseteq \dom(S_a)$.
\end{observation}

\begin{claim}
\label{clm:spine-and-radius}
    Every supernode $S \in \cS$ 
    is a subset of $\dom(S)$. Moreover, $S$ consists of a spine $\pi$ which is a shortest path in the contact graph of $\dom(S)$, together with regions of $\cR^*$ in the final output that are within distance 2 of $\pi$ in the contact graph of $S$ (and are thus within distance 2 in the contact graph of $\dom(S)$).
\end{claim}
\begin{proof}
Consider some supernode $S$ initialized by a call $\textsc{SelectPaths}(H, W_{\rm old}, \cR)$. Let $\cR'$ be the set of shattered regions defined during Step 2 of this call. Clearly $S \subseteq \cR'$, and $\cR'$ is a subset of the regions returned by $\textsc{SelectPaths}(H, W_{\rm old}, \cR)$. By an easy inductive argument, $S \subseteq \cR^*$. Further, every region in $S$ intersects $H$, so we conclude $S \subseteq \dom(S)$.

Let $\pi$ be the spine initialized in Step 1 of \textsc{SelectPaths}. Path $\pi$ is a shortest path in the contact graph of $\cR$. Moreover, by construction, every region of $\pi$ appears in the supernode $S$, and thus in $\dom(S)$. As $\dom(S)$ is a shattering of $\cR$, we conclude that $\pi$ is also a shortest path in the contact graph of $\dom(S)$. (Indeed, any path in $\dom(S)$ between vertices in $\pi$ can be pulled back to a path in $\cR$ of no greater length, as in the proof of \Cref{obs:shattered-diam}.)

Now, consider an arbitrary region $R \in S$; we show here that $\dist_{S}(R, \pi) \le 2$. If $R \in S^{\rm init}$ in Step 2 of \textsc{SelectPaths}, by definition of $S^{\rm init}$ and $\cR'$, $R$ was not shattered, so $R$ is in the supernode $S$ and we have $\dist_{S}(R, \pi) = \dist_{\cR}(R, \pi) \le 1$. Otherwise, $R$ contains some vertex in $V(S^{\rm init})$, and thus is adjacent to some other region $R' \in S^{\rm init}$ in the contact graph of $S$, 
so $\dist_{S}(R, \pi) \le 2$. Observe that $\dist_{\dom(S)}(R, \pi) \le \dist_S(R, \pi) \le 2$, so we also have $\dist_{\dom(S)}(R, \pi) \le 2$.
\end{proof}

\begin{observation}
\label{obs:intersect-parent}
    Let $S \in \cS$ be a supernode, let $S_{\rm a} \in \cS$ be the parent of $S$, and let $S_{\rm up} \in \cS$ be an arbitrary ancestor of $S$. Any path in $H_0$ between a vertex in $V(S)$ and a vertex in $V(S_{\rm up})$ intersects $V(S_a)$.
\end{observation}
\begin{proof}
    Consider the call $\textsc{SelectPaths}(H, W_{\rm old}, \cR)$ that initialized $S$. By construction, $W_{\rm old} = V(S_a)$. Further, $V(S_{\rm up})$ is disjoint from $V(H)$. As $H$ is a connected component in $H_0 \setminus W_{\rm old}$, any path in $H_0$ from a vertex in $H$ to a vertex not in $H$ intersects $W_{\rm old}$ (and thus intersects $V(S_a)$).
\end{proof}

\begin{observation}
\label{obs:adj-to-two}
    Let $R \in \cR^*$ be a region that is not in any supernode, ie.\ $R \not \in \cS$.
    \begin{enumerate}
        \item[(a)] $R$ is vertex-disjoint from regions in $\cS$, ie.\ $V(R) \cap V(\bigcup \cS) = \varnothing$.
        \item[(b)] There are two supernodes $S, S_a \in \cS$ such that any path in $H_0$ between a vertex in $V(R)$ and a vertex in $V(\bigcup \cS)$ intersects $V(S \cup S_a)$.\footnote{In particular, this means that the only supernodes that may be adjacent to $R$ (in $H_{\cR^*})$ are $S$ and $S_a$.}
        Moreover, $S_a$ is the parent of $S$, and $R \in \dom(S)$. 
    \end{enumerate}
\end{observation}
\begin{proof}
Because $R$ is not in any supernode, there is some call $\textsc{SelectPaths}(H, W_{\rm old}, \cR)$ that creates a supernode $S$, such that $V(R)$ is contained in some connected component $H_i$ of $H \setminus V(S)$ and $H_i$ does not contain any outer vertex of $H_0$.
Note that $H_i$ is a connected component in $H_0 \setminus (V(S) \cup W_{\rm old})$. There are no supernodes in $H_i$, so $R$ is vertex-disjoint from the regions in $\cS$. Moreover, any path in $H_0$ between $V(R)$ and a vertex in $V(\bigcup \cS)$ intersects intersect $V(S)$ or $W_{\rm old}$. These are precisely the regions in $S$ and its the parent supernode $S_a$. By definition of domain, $R \in \dom(S)$.
\end{proof}

We now prove two lemmas that will be helpful for proving \Cref{lem:gridtree}(1).  Observe that the supernodes created by \textsc{SelectPaths} may not cover all outer regions of $H_0$ if an outer region of $\cR_0$ is shattered at some point. Nevertheless, we show in the following lemma that all vertices in any outer regions that are not covered by supernodes are adjacent to some region in some supernode.
\begin{lemma}
\label{lem:external-neighbor}
Let $\hat \cS \gets \bigcup \cS$ denote the set of all regions in $\cR^*$ that belong to some supernode.
Let $R \in \cR_0$ be an outer region of $H_0$, and let $r \in R$. There is some region $R^* \in \cR^*$ that contains $r$, such that $\dist_{\cR^*}(R^*, \hat \cS) \le 1$.
\end{lemma}
\begin{proof}
We prove the following stronger claim.
\begin{quote}
\textbf{Claim.} Let $\textsc{SelectPaths}(H, W_{\rm old}, \cR)$ be some recursive call made during the computation of $\cR^*$ and $\cS$. Let $R$ be any region (including a shattered subregion of a region of $\cR_0$) that either contains an outer vertex of $H_0$, or has a vertex in $R$ that is adjacent in $H_0$ to some vertex in some region of $\hat \cS$.
    Let $r \in R$. There is some region $R^* \in \cR^*$ that contains $r$, such that $\dist_{\cR^*}(R^*, \hat \cS) \le 1$.
\end{quote}
The proof is by induction, on the depth of the recursive call.
Let \EMPH{$S$} be the supernode initialized during the call $\textsc{SelectPaths}(H, W_{\rm old}, \cR)$, and let \EMPH{$\cR^\out$} be the set of regions formed by shattering $\cR$ with $S$. We divide into several cases.

Suppose that $R$ is \emph{not} shattered, that is $R \in \cR^\out$.
If $R \in S$, then $R \in \hat \cS$, so we are done.
Otherwise, there is some connected component $H_i$ of $H \setminus V(S)$ that intersects $R$, and a corresponding set of regions $\cR_i$ such that $R \in \cR_i$.
If $H_i$ contains some outer vertex of $H_0$ (for example, if $R$ contains an outer vertex of $H_0$) then the \textsc{SelectPaths} procedure makes a recursive call to $\textsc{SelectPaths}(H_i, V(S), \cR_i)$, and the claim holds by induction.
Otherwise, if $H_i$ does not contain some outer vertex, there is no recursive call, and we conclude $R \in \cR^*$. In this case, recall that $R$ is adjacent to some region in $\hat \cS$ (by assumption, as $R$ does not contain an outer vertex of $H_0$), so we are done. 

Now suppose that $R$ is shattered during the creation of $\cR^\out$. Let $R' \in \cR^\out$ be a subregion of $R$ containing $r$. If $R' \in S$, we are done. Otherwise, $R'$ lies inside some connected component $H_i$ of $H \setminus V(S)$.
By the definition of shattering procedure, $R'$ is adjacent to some region in $S$ (and thus adjacent to some region in $\hat \cS$).
If $H_i$ contains no outer vertex of $H_0$, then (as above) $R' \in \cR^*$ and we are done.
Otherwise, the claim follows by induction (as above).
\end{proof}

\begin{observation}
\label{obs:outer-connected}
The induced subgraph $H_0[V(\cS)]$ is connected and every supernode contains at least one outer vertex of~$H_0$.
\end{observation}
\begin{proof}
    By construction, every supernode $S$ contains a region of at least one outer vertex of $S$, and $V(S)$ induces a connected subgraph of $H_0$. Moreover, if $S_a$ is the parent of $S$, some vertex in $V(S)$ is connected to some vertex of $V(S_a)$ by an edge of $H_0$ (in fact, an outer edge); thus $H_0[V(\cS)]$ is connected.
\end{proof}

\subsubsection{Towards a scattering clustering}
To conclude this section, we prove a crucial lemma (\Cref{lem:select-paths-scattering}) which says that the shattering of \textsc{SelectPaths} does not affect distances in the contact graph too much. We will use this lemma to prove the [scattering] property of the final clustering. We first need a technical lemma, where we use \Cref{lem:blt}(2).

\begin{lemma}
\label{lem:crossing-pi}
Suppose that at some point $\textsc{SelectPaths}(H, W_{\rm old}, \cR)$ is called and creates a supernode $S$, and makes recursive call to $\textsc{SelectPaths}(H_i, V(S), \cR_i)$ in Step 3 for connected components $H_1, \dots, H_\kappa$ and shattered regions $\cR_1, \dots, \cR_\kappa$. 
If $R \in \cR$ is a region that contains a vertex adjacent to $W_{\rm old}$, then no subregion of $R$ is in $\cR_i$ for any $i$. 
\end{lemma}
\begin{proof}
    As defined in Step 1 of \textsc{SelectPaths}, let $Y$ be the critical vertices of $(H,W_{\rm old})$, let $\set{R_y}_{y \in Y}$ be regions that each contain the critical vertices $Y$, and let $\pi$ be the shortest path between $\set{R_y}_{y\in Y}$ in the contact graph of $\cR$. By definition of the contact graph, there is a path $\hat \pi$ in $H$ between the vertices of $Y$ which contains only vertices in $V(\pi)$.

    Suppose for contradiction that there is some subregion $R'$ of $R$ such that $R' \subseteq V(H_i)$ for some connected component $H_i$ of $H \setminus V(S)$. Let $r' \in R'$, and let $r \in R$ be a vertex adjacent to some vertex $w \in W_{\rm old}$. As $R$ is a connected region, there is a path $\Gamma_1$ in $R$ between $r$ and $r'$ in $R$. As $H_i$ is a connected component of $H \setminus V(S)$, and (by the algorithm in Step 3) $H_i$ contains an outer vertex of $H_0$, there is a path $\Gamma_2$ in $H_0 \setminus \bigcup S$ between $r'$ and an outer vertex of $H_0$. Let 
    \[\Gamma \coloneqq \set{w} \circ \Gamma_1 \circ \Gamma_2.\]
    By \Cref{lem:blt}(2), $\Gamma$ intersects $\hat \pi$. But $w$ and $\Gamma_2$ are disjoint from $V(S) \supseteq \hat \pi$, so $\Gamma_1 \subseteq R$ intersects $\hat \pi$. Thus, $R$ is adjacent to some region of $\pi$ in the contact graph of $\cR$. The definition of $S$ in Step 2 implies that $R \in S$. But this implies that $R' \subseteq R \subseteq V(S)$, contradicting our assumption that $R' \subseteq V(H) \setminus V(S)$.
\end{proof}

One simple consequence of \Cref{lem:crossing-pi} is that after a region is shattered for the first time, it can only be shattered one more time by its direct descendents.
\begin{observation}
\label{obs:shatter-two-level}
\sloppy Suppose a region $R\in \cR_0$ is shattered in the call to $\textsc{SelectPaths}(H, W_\textrm{old}, \cR)$ which created supernode $S$. 
Then subregions of\, $R$ may be further shattered in a recursive call
$\textsc{SelectPaths}(H_i, V(S), \cR_i)$ that is made immediately,
but not in any further recursive calls. 
In other words, the subregions of $R$ produced in $\textsc{SelectPaths}(H_i, V(S), \cR_i)$ appear in the final shattering $\cR^*$.
\end{observation}
\begin{proof}
   Consider the first instance of $\textsc{SelectPaths}(H, W_\textrm{old}, \cR)$ that created the supernode $S$ which shattered a region $R \in \cR_0$. A subregion $R'$ of $R$ may appear in $R_i$ in recursive calls to $\textsc{SelectPaths}(H_i, V(S), \cR_i)$ in Step 3 of $\textsc{SelectPaths}$ that created a supernode $S_i$. In these recursive calls, $R'$ has a vertex adjacent to $V(S)$ as it was shattered by $S$. Thus by \Cref{lem:crossing-pi}, subregions $R'$ will not appear in subsequent recursive calls, but may still have been shattered in $\textsc{SelectPaths}(H_i, V(S), \cR_i)$.
\end{proof}

Now we can examine the distance between two subregions of $R$ immediately after it is shattered.

\begin{lemma}
\label{lem:dist-after-shatter}
Consider some call $\textsc{SelectPaths}(H, W_\textrm{old}, \cR)$, and let $R \in \cR$ be some region shattered by the call.
Let $\cR'$ be the regions created in Step 2, and 
$R_a, R_b \in \cR'$ be two subregions of $R$.  Then,
there is a walk $Y$ in $H_{\cR'}$ between $R_a$ and $R_b$ such that $\len{Y} \le 8$ and every region in $Y$ other than $R_a$ and $R_b$ is in the supernode $S$ created by the call $\textsc{SelectPaths}(H, W_\textrm{old}, \cR)$.
\end{lemma}

\begin{proof}
Since $R_a$ is a subregion of $R$ which was shattered by $S^{\textrm{init}}$ (as defined in Step 2 of the pseudocode for \textsc{SelectPaths}), then $R_a$ must be adjacent to a region in $S^{\textrm{init}}$ and so $\dist_{\cR'}(R_a, \pi) \le 2$.
The same holds for $R_b$.
Let $P_a, P_b\in \pi$ be the closest regions to $R_a$ and $R_b$ respectively in the contact graph of $\cR'$. 
Since $R_a$ and $R_b$ are subregions of $R$, then $\dist_\cR(R, P_a)\le 2$ and $\dist_\cR(R,P_b)\le 2$. As $\pi$ is a shortest path of $\cR$ that lies completely in $S^{\rm init}$, we have that $\dist_{\cR'}(P_a, P_b) = \dist_\cR(P_a, P_b) \le 4$. 
Now we can let $Y$ be the walk from $R_a$ to $P_a$ to $P_b$ to $R_b$, and observe that the following holds:
\[ \len{Y} \le \dist_{\cR'}(R_a, P_a) + \dist_{\cR'}(P_a, P_b) + \dist_{\cR'}(P_b, R_b) \le 8.\]
It is clear that every region in $Y$ other than $R_a$ and $R_b$ is in $S$.
\end{proof}

We can now show that all the shattering we do (across all recursive calls of $\textsc{SelectPaths}$) keeps shattered regions close, that is, two subregions of a single shattered region are still within $O(1)$ distance in the resulting contact graph.%
\begin{lemma}
\label{lem:select-paths-scattering}
    For any region $R$ in $\cR_0$, for every pair of subregions $R_a$, $R_b$ of $R$ in $\cR^*$, there is a walk $Y$ between $R_a$ and $R_b$ in the contact graph $H_{\cR^*}$ with 
    \[\len Y \le 24.\] Moreover,
    $Y$ has a particular structure: it can be written as the concatenation 
    \[Y = [R_a] \circ Y_1 \circ [\tilde R_a] \circ Y_2 \circ [\tilde R_b] \circ Y_3 \circ [R_b]\]
    where $Y_1,Y_2,Y_3$ are each walks of length at most $6$ and are contained in a single supernode, and $\tilde R_a$ and $\tilde R_b$ are two subregions of $R$ that are adjacent to a supernode.
\end{lemma}

\begin{proof}
If $R$ was never shattered, then $R_a = R_b = R$ and the claim is trivial.
If $R$ was shattered, let $\textsc{SelectPaths}(H, W_\textrm{old}, \cR)$ be the earliest recursive call that shattered $R$ and created supernode $S$, 
let $\cR'$ be the set of regions from Step~2,
and let $R_a', R_b'\in \cR'$ be subregions of $R$ that contain $R_a$ and $R_b$ respectively. By \Cref{lem:dist-after-shatter}, there is a walk 
\[Y' = [R'_a = R_0, ..., R_\ell=R'_b]\]
in $H_{\cR'}$ where $\ell = \len{Y'} \le 8$ and all $R_i$ are in $S\subseteq \cR^*$ for $1\le i \le \ell-1$. We define $\EMPH{$Y_2$} = [R_1, \ldots, R_{\ell-1}]$ to be the interior of $Y'$. There must be a region $\EMPH{$\tilde{R}_a$}\in \cR^*$ (resp.\ $\EMPH{$\tilde{R}_b$}\in \cR^*$) that is a subregion of $R'_a$ (resp.\ $R'_b$) such that there is a walk 
\[\EMPH{$\tilde{Y}$} = [\tilde{R}_a, R_1, ..., R_{\ell-1}, \tilde{R}_b] = [\tilde{R}_a] \circ Y_2 \circ [\tilde{R}_b]\]
in $H_{\cR^*}$.

We now construct a walk \EMPH{$Y_a$} in $H_{\cR^*}$ between $R_a$ and $\tilde R_a$; we will define \EMPH{$Y_1$} to be the interior of $Y_a$.
By \Cref{obs:shatter-two-level}, either $R_a$ and $\tilde R_a$ are both subregions of $R$ created because of $R'_a$ being shattered in some recursive call $\textsc{SelectPaths}(H_a, V(S), \cR_a)$, or $R'_a$ was not subsequently shattered and $R_a = \tilde{R}_a = R'_a$. 
In the latter case, there is a trivial walk $Y_a = [R_a]$ between $R_a$ and $\tilde{R}_a$ in $H_{\cR^*}$. 
In the former case, \Cref{lem:dist-after-shatter}  implies that there is a walk $Y_a$ from $R_a$ to $\tilde{R}_a$ in the contact graph $H_{\cR_a}$ of length at most $8$. Moreover, $Y_a$ is also a walk  in $H_{\cR^*}$, the contact graph of the final shattering $\cR^*$: this is because the interior of $Y_a$ is in a supernode and the endpoints of $Y_a$ are $\set{R_a, \tilde R_a}$, so every region in $Y_a$ appears in $\cR^*$.

Similarly we can define \EMPH{$Y_b$} to be a walk in $H_{\cR^*}$ from $\tilde{R}_b$ to $R_b$ that has length at most $8$, and define \EMPH{$Y_3$} to be the interior of $Y_b$.
Define the walk $Y$ as 
\[Y = Y_a \circ \tilde Y \circ Y_b = [R_a] \circ Y_1 \circ [\tilde R_a] \circ Y_2 \circ [\tilde R_b] \circ Y_3 \circ [R_b].\]
For each walk $Y_1$, $Y_2$, $Y_3$,
it follows from \Cref{lem:dist-after-shatter} that 
the walk belongs to a single supernode and has length at most 6. The regions $\tilde R_a$ and $\tilde R_b$ are subregions of $R$ by construction. Overall, $Y$ has length at most $8+8+8 = 24$.
\end{proof}

\subsection{\textsc{ClusterOuter}: The expansion after \textsc{SelectPaths}}
\label{SS:expand}

The goal of this section is twofold. First, we want to \emph{expand} all the supernodes produced by \textsc{SelectPaths} by one hop; this will guarantee (by \Cref{lem:external-neighbor}) that all subregions of external regions are assigned to some expanded supernode).
Second, we want to \emph{cluster} the vertices by breaking apart every supernode into clusters of bounded $\cR$-diameter. For technical reasons, 
we will do both at the same time.

A \EMPH{cluster} is a connected subset of $V(H)$. We will build clusters by iteratively assigning vertices (of $H$) to clusters. 
Every vertex is initially \EMPH{unassigned}.  Repeatedly we will mark the vertices \EMPH{assigned} using the following procedure \EMPH{$\textsc{Assign}(v \rightsquigarrow C, X)$}.
Given some vertex $v$ and a cluster $C$, and a set of vertices $X \subseteq V(H)$ such that $H[X]$ is connected and $C \cap X \neq \varnothing$:
\begin{tcolorbox}[breakable]
    \ul{$\textsc{Assign}(v \rightsquigarrow C, X)$:}
    
    \textbf{Input:} vertex $v$ in $X$, cluster $C$, vertex set $X \subseteq V(H)$ such that $H[X]$ is connected and $C \cap X \neq \varnothing$\\
    
    Let $P$ be an arbitrary path in $H[X]$ between $v$ and an arbitrary vertex in $X \cap C$. Let $v'$ be the first assigned vertex that $P$ intersects (when viewing $P$ as starting at $v$ and ending at $C$), and let $P[v:v']$ be the set of vertices that $P$ intersects before $v'$. Let \EMPH{$C'$} be the cluster that $v'$ is assigned to. 
    Assign all the vertices in $P[v:v']$ to the cluster $C'$.
\end{tcolorbox}

With these definitions in place, we now describe the expand-and-cluster procedure \textsc{ClusterOuter}. At a high level, after we compute a set of supernodes $\cS$ from \textsc{SelectPaths}, we will first partition all regions $R$ that are within 1-hop of a supernode. These regions will be (arbitrarily) assigned to a set $S^+$ corresponding to some supernode $S\in \cS$ that neighbors $R$.
We wish to next find a clustering for all vertices in each of these regions.
Conveniently, for every pair of supernodes $S,S' \in \cS$, the set $V(S)$ is vertex disjoint from $V(S')$.
However, for a supernode $S\in \cS$ with parent supernode $S_a\in \cS$,  the vertex sets of their 1-hop neighborhood $V(S_a^+)$ and $V(S^+)$ may not be disjoint; that is, there could be a vertex $v\in V(S_a^+)\cap V(S^+)$, which can happen when a region $R_a\in S_a^+$ is adjacent to a region $R\in S^+$ (or we could even have $R_a=R)$. In this case, it is difficult to decide how to assign such a vertex to in a way that ensures that each cluster $C$ is both connected (in $H$) and also has a good diameter and scattering bound.
In order to resolve this issue, we carefully choose the order in which we call \textsc{Assign} on the vertices.

\begin{tcolorbox}[breakable]
\textul{$\textsc{ClusterOuter}(H, \cR)$}:

\textbf{Input:} plane graph $H$, set of regions $\cR$ whose support is $V(H)$. 

\textbf{Output:} a (partial) set of clusters $\cC$ of $H$

\begin{enumerate}
    \item \emph{Select paths via \textsc{SelectPaths}.}
    
    Compute $(\EMPH{$\cR^*$}, \EMPH{$\cS$}) \gets \textsc{SelectPaths}(H, \varnothing, \cR)$.

    \item \emph{Define the expansion $S^+$, and define balls in preparation for assignment.}

    For every supernode $S \in \cS$, initialize $\EMPH{$S^+$} \gets \cR^*[S]$ as the set of regions in $S$. Now iterate over the supernodes $S \in \cS$ in an arbitrary order: add to $S^+$ the set of all regions $R^* \in \cR^*$ such that $R^*$ is \emph{adjacent} to some region in $S$ (in the contact graph of $\cR^*$) and $R^*$ does not already belong to another set $S_{\rm old}^+$ for some other supernode $S_{\rm old} \in \cS$.

    Define a set of \EMPH{net points} $\EMPH{$\cP_S$} \subseteq S$ by walking along the spine $\pi$ of $S$, starting with the first region of $\pi$ and adding every 24th region along $\pi$ to $\cP_S$. (That is, two consecutive net points are separated by distance exactly 24.)

    For each net point $P \in \cP_S$, define the ball $\EMPH{$B_P$} \subseteq S$ to be the set of all regions $R$ in $S$ with $\dist_{S}(P, R)\le 11$. Define the ball $\EMPH{$B^{+}_P$} \subseteq S^+$ to be the set of all regions $R \in S^+$ with $\dist_{S^+}(P, R)\le 6$.
    
     \item \emph{Assign vertices in $S \in \cS$ to clusters.}

    Iterate over the supernodes $S \in \cS$ in an arbitrary order.
    \begin{enumerate}
        \item For every net point in $P \in \cP_S$, create a new cluster \EMPH{$C_P$}. Assign every vertex in $V(B_P^{})$ to the cluster $C_P$.
        \item While there exists some unassigned vertex $v \in V(S)$: let $P \in \cP_S$ be an arbitrary net point
        and call $\textsc{Assign}(v \rightsquigarrow C_P, V(S))$.
    \end{enumerate}

    \item \emph{Assign vertices in $S^+$ to clusters.}
    
    Iterate over the supernodes $S \in \cS$ respecting the natural partial order. For each $S$:
    \begin{enumerate}
        \item For every net point $P \in \cP_S$,
        while there exists an unassigned vertex $v \in V(B_P^{+})$, call $\textsc{Assign}(v \rightsquigarrow C_P, V(B_P^{+}))$.
        \item While there exists some unassigned vertex $v \in V(S^+)$: let $P \in \cP_S$ be an arbitrary net point
        and call $\textsc{Assign}(v \rightsquigarrow C_P, V(S^+))$.
    \end{enumerate}
    
    \item \emph{Return.}

    Return the set of clusters $\cC \gets \set{C_P}$.
\end{enumerate}
\end{tcolorbox}

\begin{figure}[th!]
    \centering
    \includegraphics[width=0.9\linewidth,page=1]{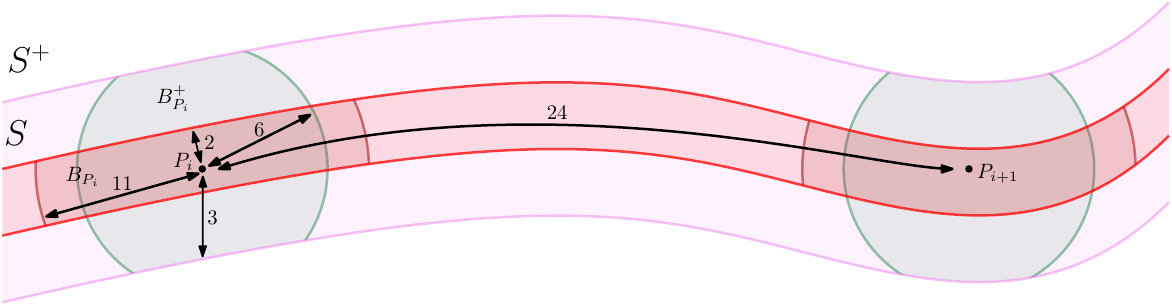}

    \caption{Schematic illustration of distances for a net point $P_i$ of a supernode $S$, its expanded region $S^+$ and the balls $B_{P_i}$ and $B_{P_i}^+$. Distances in the region intersection graph from $P_i$ is illustrated.}
    \label{fig:clustering-1}
\end{figure}

\begin{figure}[th!]
    \centering
    \includegraphics[width=0.9\linewidth,page=2]{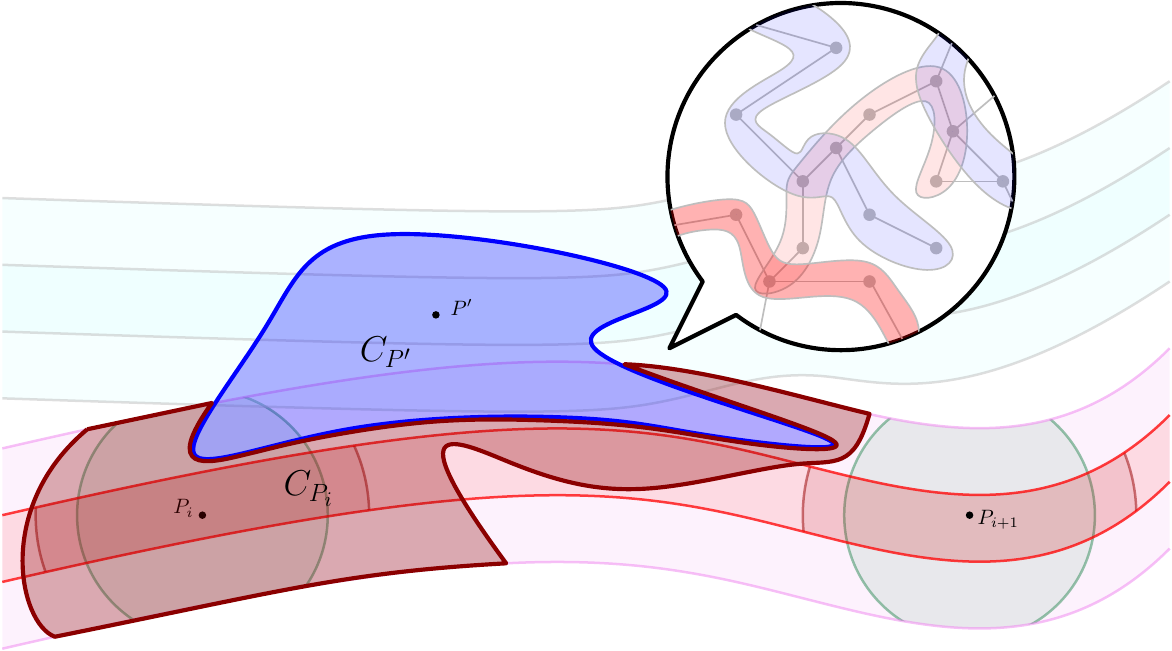}
    \caption{Illustration of a possible clustering of $C_{P_i}$ around net point $P_i$, and the interaction with another net point $P'$ from a parent cluster $S_a$. Note that $C_{P'}$ does not contain any of $S$. $C_{P_i}$ contains $B_{P_i}$, but avoids $B_{P_{i+1}}$ and $B_{P_{i+1}}^+$.
    Regions of $S^+$ and $S_a^+$ can share vertices as can be seen in the zoom in the top right where the dark red region is a part of $S$, the light red region is a part of $S^+$, and the light blue regions are a part of $S^+_a$.
    It is unavoidable that we assign vertices of $V(S^+)$ to the net point $P'$ as vertices of $V(S^+)$ being disconnected by $C_{P'}$.
    }
    \label{fig:clustering-2}
\end{figure}

\bigskip
Our goal is to show that the clusters we have created satisfy the properties of \Cref{lem:gridtree}.
We begin with a sequence of simple observations about \textsc{ClusterOuter}.

\begin{observation}
\label{obs:clusterouter}
The following is true:
\begin{enumerate}
    \item[(a)] 
        Let $S \in \cS$, let $S_a \in \cS$ be the parent of $S$, and let $R \in S^+$ be an arbitrary region. Then $R \in \dom(S)$ (and thus $R\in \dom(S_a)$).
    \item[(b)]  
    For any two distinct net points $P$ and $P'$, the sets $V(B_{P}^{})$ and $V(B_{P'}^{})$ are disjoint. If $P$ and $P'$ are net points from the same supernode, then $V(B_{P}^{+})$ and $V(B_{P'}^{+})$ are disjoint.

    \item[(c)] For any net point $P$ on $S$, we have $V(S)\cap V(B^+_P) \subseteq V(B_P)$.
    
    \item[(d)] 
    After the execution of \textsc{ClusterOuter} every vertex in $V(S^+)$ for every supernode $S\in \cS$ is assigned to exactly one cluster.%
    \item[(e)]
    At all times during \textsc{ClusterOuter}, every cluster induces a connected subgraph of $H$.
\end{enumerate}
\end{observation}
\begin{proof}
\textit{Proof of (a).} If $R \in S$, we are done. Otherwise, $R$ is adjacent to $S$ and $R$ is not in any supernode of $\cS$, and the observation follows from \Cref{obs:adj-to-two}(b).

\textit{Proof of (b).}
We first prove $V(B_{P}^{}) \cap V(B_{P'}^{}) = \varnothing$. 
If $P$ and $P'$ belong to two different supernodes $S, S' \in \cS$, then the statement holds trivially because $V(S)$ and $V(S')$ are disjoint. Suppose that $P$ and $P'$ are both net points of the same supernode $S$.
Suppose for contradiction that there is some vertex $v \in R \cap R'$ where $R \in B_P^{}$ and $R' \in B_{P'}^{}$.
By triangle inequality and definition of $B^{}$, we have $\dist_S(P, P') \le \dist_S(P, R) + \dist_S(R, R') + \dist_{S}(R', P') \le 11 + 11 + 1 = 23$.
Now, let $\pi$ be the spine of $S$.
It follows from the definition of net points (and the fact that $\pi$ is a shortest path in the contact graph of $S$) that $\dist_S(P, P') \ge 24$, a contradiction.

Similarly, assume $P$ and $P'$ are net points of the same supernode $S$, we have $V(B_{P}^{+}) \cap V(B_{P'}^{+}) = \varnothing$. Triangle inequality and definition of $B^{+}$ implies $\dist_{S^+}(P, P') \le 6 + 6 + 1 = 13$. But $\pi$ is a shortest path in the contact graph of $S^+$ (because $S^+ \subseteq \dom(S)$ by (a), and $\pi$ is a shortest path in $H_{\dom(S)}$ by \Cref{clm:spine-and-radius}), so $\dist_{S^+}(P, P') \ge 24$, a contradiction.

\textit{Proof of (c).}
Let $v\in V(S) \cap V(B_P^{+})$. Then there must be some region $R\in S$ where $v\in R$. 
Let $\pi(R)$ denote an arbitrary region of $\pi$ of minimum distance (in $H_{S}$) to $R$. We have $\dist_S(R,\pi(R)) \le 2$ by definition of $S$.

Because $v \in V(B_P^+)$, there is some region $R^+\in B_P^+$ where $v\in R^+$ (where it is possible that $R=R^+$) and $\dist_{S^+}(P, R^{+}) \le 6$. 
Observe that this means that by the triangle inequality
\[
\dist_{S^+}(P, \pi(R)) \le \dist_{S^+}(P, R^+) + \dist_{S^+}(R^+, R) + \dist_{S^+}(R, \pi(R)) \le 6 + 1 + 2 = 9.
\]
We use the fact that $S\subseteq S^+$, so distances can only be shorter in $S^+$.
Note that $P$ and $\pi(R)$ are both on a shortest path $\pi$ in $\dom(S)$, and $S^+\subseteq\dom(S)$ by (a), so we must also have $\dist_{S}(P, \pi(R))\le 9$. Thus,
\[
\dist_{S}(P, R) \le \dist_{S}(P, \pi(R)) + \dist_{S}(\pi(R), R) \le 9 + 2 = 11
\]
so we conclude that $R\in B_P$ and $v\in V(B_{P})$.

\textit{Proof of (d).}
The claim holds after all clusters are initialized in step 3(a), because the sets $V(B_{P}^{})$ are disjoint, by (b). 
During the remaining steps, vertices are only assigned to clusters by the \textsc{Assign} procedure, which assigns all currently-unassigned vertices to a cluster. 

\textit{Proof of (e).}
Every cluster is initialized as a connected subset of $H$, because (by definition of contact graph) $H[V(B_{P})]$ is connected. The invariant is maintained later because vertices are only assigned via the \textsc{Assign} procedure: this procedure guarantees that whenever vertices are assigned to a cluster $C$, there is a path in $H$ between every newly-assigned vertex and a previously-assigned vertex in $C$. 
\end{proof}

We now prove the three properties needed in \Cref{lem:gridtree}.

\subsubsection{Outer regions are clustered}
\begin{lemma}
\label{lem:gridtree-outer-clustered}
    \textnormal{[Outer-clustered.]}
    For every outer region $R \in \cR$, every vertex of $R$ is assigned to exactly one cluster of $\cC$. Moreover, the subgraph $H[V(\cC)]$ is connected.
\end{lemma}
\begin{proof}
Let $v \in R$. Let $\cR^*$ be the set of shattered regions and let $\cS$ be the set of supernodes computed by $\textsc{SelectPaths}(H, \varnothing, \cR)$. Let $\hat \cS \coloneqq \bigcup \cS$ denote the set of regions in any supernode of $\cS$. By \Cref{lem:external-neighbor}, there is a region $R^* \in \cR^*$ that contains $v$, such that
$\dist_{\cR^*}(R^*, \hat \cS) \le 1$. We conclude that, for some supernode $S \in \cS$, region $R^*$ belongs to $S^+$. 
\Cref{obs:clusterouter}(d) shows that
all vertices in $V(S^+)$ for each $S\in \cS$ are clustered.
As every cluster in $\cC$ is connected in $H$ by \Cref{obs:clusterouter}(e), every cluster contains some vertex in $V(\hat \cS)$, and $H[V(\hat \cS)]$ is connected (by \Cref{obs:outer-connected}), we conclude that $H[V(\cC)]$ is connected.
\end{proof}

\subsubsection{Where a vertex can be assigned}
Before we analyze the diameter bound, it is helpful to examine which clusters a vertex may be assigned to.
We will focus on the spine $\pi$ of a supernode $S\in \cS$.
Recall that in Step 2 of \textsc{ClusterOuter}, we walked along $\pi$ to select net points $P_1, P_2, \dots P_k$. The walk along $\pi$ defines a total order \EMPH{$\preceq$} of all the regions in $\pi$, where $P_i \preceq P_{i+1}$ for any $i$.
For a region $R \in \cR^*$, we let $\EMPH{$\pi(R)$}\in S^+$ denote the closest region on $\pi$ to $R$ (in $H_{S^+}$), breaking ties arbitrarily.
For a vertex $v$, let $R_v$ be a region in $\cR^*$ containing $v$ that minimizes the distance (in $H_{S^+}$) from $R_v$ to $\pi(R_v)$; we define
$\EMPH{$\pi(v)$}\coloneqq \pi(R_v)$.
Observe the following:
\begin{observation}
\label{obs:dist-to-proj}
    Let $S$ be a supernode. For every $v \in V(S)$, there exists a region $R \in S$ with $\dist_S(R, \pi(v)) \le 2$. Likewise, for every $v \in V(S^+)$, there exists a region $R \in S^+$ with $\dist_{S^+}(R, \pi(v)) \le 3$.
\end{observation}
\begin{proof}
Suppose $v \in V(S)$. No region outside $S$ contains a vertex in $V(S)$, so the region $R_v$ in the definition of $\pi(v)$ must lie in $S$. By \Cref{clm:spine-and-radius}, $R_v$ is within distance $2$ of $\pi$ in $H_S$. Finally, observe that (by description of the \textsc{SelectPaths} algorithm), the spine $\pi$ is not adjacent to any region outside of $S$; thus, any path between $R_v$ and $S$ that wanders outside $S$ has length more than 2. We conclude that the shortest path between $R_v$ and $\pi$ lies within $S$ and has length at most 2, as desired.

Now suppose $v \in V(S^+)$. The observation follows from the fact that every region in $S^+$ is adjacent to a region in $S$. Thus $v$ is within distance 3 (in $H_{S^+}$) of some region in $\pi$.
\end{proof}

For notational convenience, we additionally use the notation that \EMPH{$X$}$\coloneqq V(\bigcup_{P\in \cP_S} B^{}_P)$, and let \EMPH{$X^+$}$\coloneqq V(\bigcup_{P\in \cP_S} B^{+}_P)$.
We begin with a simple lemma that shows that connected components of $V(S)\setminus X$ are sandwiched between two balls $B_P$.
\begin{lemma} \label{lem:sandwich}
Consider a supernode $S\in \cS$. Suppose two vertices $v_a$ and $v_b$ are in the same component of $H[V(S)\setminus X]$. For any $P\in \cP_S$, if $P \preceq \pi(v_a)$ then $P \preceq \pi(v_b)$. Similarly, if $\pi(v_a) \preceq P$ then $\pi(v_b) \preceq P$. The same is true of two vertices $v_a^+$ and $v_b^+$ that are in the same component of $H[V(S^+)\setminus X^+]$.
\end{lemma}
\begin{proof}
We first show that if $P \preceq \pi(v_a)$, then $P \preceq \pi(v_b)$. The claim that $\pi(v_a) \preceq P$ implies $\pi(v_b) \preceq P$ can be proven similarly.

Let
$\EMPH{$\Gamma$} = (v_a = v_0, v_1, \ldots, v_\ell = v_b)$
be a path between $v_a$ and $v_b$ in
$H[V(S) \setminus X]$.
We claim that, for all vertices $v_i \in \Gamma$, we have $P \preceq \pi(v_i)$. Suppose otherwise, and let $i \in [\ell]$ be the smallest index so that $\pi(v_i) \prec P$. As $P \preceq \pi(v_a) = \pi(v_0)$ by definition, we have $i > 0$, so in particular the vertex $v_{i-1}$ is well-defined and satisfies $P \preceq \pi(v_{i-1})$. Let $\alpha=2$ and we first claim that
\begin{equation}
\label{eq:proj-close}
    \dist_{S}(\pi(v_{i-1}), \pi(v_{i})) \le 2\alpha + 1.
\end{equation}
To this end, let $R_{v_i} \in S$ (resp. $R_{v_{i-1}} \in S$) be a region in $S$ containing $v_i$ that minimizes the distance between $R_{v_i}$ and $\pi(R_{v_i})$ (resp. between $R_{v_{i-1}}$ and $\pi(R_{v_{i-1}})$) in the contact graph $H_S$. By definition of contact graph, there is an edge between the regions $R_{v_{i-1}}$ and $R_{v_i}$. Finally, \Cref{obs:dist-to-proj} implies that
$\dist_{S}(\pi(v_{i-1}), R_{v_{i-1}}) \le \alpha$ and $\dist_{S}(\pi(v_i), R_{v_i}) \le \alpha$. \Cref{eq:proj-close} follows by triangle inequality.

Now, because $\pi(v_i) \preceq P \preceq \pi(v_{i-1})$ and the spine $\pi$ is a shortest path in the contact graph of $S$, it must be the case that $P$ lies on a shortest path between $\pi(v_j)$ and $\pi(v_{j-1})$. 
Thus, by \Cref{eq:proj-close}, we have $\dist_S(\pi(v_i), P) + \dist_S(P, \pi(v_{i-1})) = \dist_S(\pi(v_i), \pi(v_{i-1})) \le 2 \alpha + 1$. 
Suppose that $P$ is no farther from $\pi(v_i)$ than $\pi(v_{i-1})$ (the other case is symmetric if we swap $v_i$ and $v_{i-1}$), so that
\begin{equation}
\label{eq:P-close}
    \dist_{S}(P, \pi(v_i)) \le \alpha.
\end{equation}
But this leads to a contradiction. As $\dist_S(\pi(v_i), R_{v_i}) \le \alpha$, \Cref{eq:P-close} implies that $\dist_S(R_{v_i}, P) \le 2 \alpha$. As $\alpha = 2$, this distance is at most 4: thus  $R_{v_i} \in B_{P}$ and so $v_i \in X$, a contradiction.

For the last sentence of the lemma, we repeat the proof but with two vertices in $S^+$ instead of in $S$. The proof is similar, except that we take $\alpha = 3$; \Cref{obs:dist-to-proj} implies that \Cref{eq:proj-close} still holds. Carrying out the proof, \Cref{eq:P-close} implies that $\dist_{S^+}(R_{v_i}^*, P) \le 2 \alpha = 6$: thus $R_{v_i}^* \in B_{P}^{+}$ and so $v_i \in X^+$, a contradiction.
\end{proof}

Using this lemma, it is easy to show that vertices assigned in step 3 must be assigned to one of the nearby net points.
\begin{claim} \label{clm:assigned-locations-3}
Let $v\in V(S)$. Then $v$ is assigned to some net point of $S$ in step 3. In particular, if $i$ denotes the index where $P_i \preceq \pi(v)$ and $\pi(v) \preceq P_{i+1}$ (if $P_{i+1}$ exists), then $v$ is assigned to the cluster corresponding to either $P_i$ or $P_{i+1}$.
\end{claim}
\begin{proof}
It is immediate from the description of the algorithm that $v$ is assigned during step 3. If $v$ is assigned in step 3(a), then $v\in V(B_P)$ for some $P$. Clearly $P$ must be the closer of $P_i$ and $P_{i+1}$, since all other net points are farther.

If $v$ is assigned in step 3(b), then by \Cref{lem:sandwich} we have $v\in S\setminus X$, and all other vertices in the same component have $P_i \preceq \pi(v) \preceq P_{i+1}$. Thus as step 3(b) calls \textsc{Assign} restricted to $V(S)$, the path found by \textsc{Assign} must eventually hit a vertex of $B_{P_i}^{}$ or $B_{P_{i+1}}^{}$, which must have been assigned to a cluster corresponding to $P_i$ or $P_{i+1}$.
\end{proof}
On the other hand, vertices assigned in step 4 may be assigned to a parent supernode.

\begin{lemma} \label{lem:parent-only}
Let $S\in \cS$ be a supernode, and let $S_a$ be the parent supernode.
In step 4, each $v\in V(S^+)$ will either be assigned to a cluster with a net point of $S$ or be assigned to a cluster with a net point of $S_a$.
\end{lemma}
\begin{proof}
    Step 4 ensures each vertex $v \in V(S^+)$ is assigned. 
    First we argue that $v$ can only be assigned to a net point on $S$ or an ancestor supernode $S_{\rm up}$ of $S$. At a high level, this follows from the fact that step 4 processes supernodes in an order that respects the partial order defined by ancestor/descendent relationships. Formally, the proof is by induction --- suppose by induction that vertices in $V(S_{a}^+)$ are assigned to ancestors of $S$ (recall that $S_a$ is the parent supernode of $S$). We need a technical observation:
    \begin{equation}
    \label{eq:overlap}
        \text{For any $S_{\rm other} \in \cS$, if $v \in V(S^+) \cap V(S_{\rm other}^+)$, then $S_{\rm other}$ is either the parent or child of $S$.}
    \end{equation}
    To prove \eqref{eq:overlap}, first observe that $v \in V(S^+) \cap V(S_{\rm other}^+)$ implies $v$ is not in any supernode, so there exists a region $R \in \cR^*$ with $R \not \in \cS$ and $v \in V(R)$. Now observe that $v \in V(S^+)$ (resp. $v \in V(S_{\rm other}^+)$) implies that there is a path in $H$ between $v$ and $V(S)$ (resp. $v \in V(S_{\rm other})$) that only uses vertices in $V(S^+)$ (resp. $v \in V(S_{\rm other}^+)$). Thus \Cref{obs:adj-to-two}(b) implies that $S$ and $S_{\rm other}$ are in a parent-child relationship.
    
    We are now ready to show that $v$ is assigned to $S$ or an ancestor supernode $S_{\rm up}$. Consider the moment before step 4 processes supernode $S$. Let us consider what assignments vertices in $V(S^+)$ can have at this moment; we call these the \EMPH{old assignments}. At this moment, some vertices of $V(S^+)$ (namely, those in $V(S)$) are assigned to $S$ (because of step 3), and some vertices of $V(S^+)$ (namely, those in $V(S_a^+) \cap V(S^+)$ are already assigned to ancestors of $S$ (by induction hypothesis). No other vertices of $V(S^+)$ are assigned, due to \eqref{eq:overlap} and the fact that step 4 processes $S$ before the children of $S$.  From the description of the \textsc{Assign} procedure, when a vertex $v \in V(S^+)$ is assigned when step 4 processes $S$, then $v$ is assigned to some cluster with an old assignment. We conclude that every vertex in $V(S^+)$ is assigned to a net point of $S$ or an ancestor of $S$.

    To complete the proof, we show that $v$ cannot be assigned to an ancestor supernode of $S$ other than the parent $S_a$. Let $S_{\rm up} \in \cS$ with $S_{\rm up} \neq S_a$, and suppose for contradiction that $v$ was assigned to a cluster $C$ with a net point in $S_{\rm up}$. By \Cref{obs:clusterouter}(e), there is a path $\Gamma$ in $H$ from $v$ to a vertex of $V(S_{\rm up})$ that is contained in cluster $C$. By \Cref{obs:intersect-parent}, $\Gamma$ intersects $V(S_a)$. But every vertex in $V(S_a)$ was already assigned to a cluster with net point in $S_a$ during step 3. As clusters are disjoint, this is a contradiction.
\end{proof}

To summarize, a vertex $v \in V(S^+)$ is assigned to either $S$ or to the parent $S_a$. If $v$ was assigned in step 3, then \Cref{clm:assigned-locations-3} tells us specifically that $v$ is assigned to $P_i$ or $P_{i+1}$. In the next lemma, we develop a similar characterization for vertices assigned in step 4. If $v$ is assigned to $S$ in step 4, then we can also prove that $v$ is assigned to either $P_i$ or $P_{i+1}$. If $v$ is assigned to $S_a$, we don't have as nice a characterization of the net points that $v$ can be assigned to, but we nevertheless give a technical characterization.
\begin{claim} \label{clm:assigned-locations-4}
Let $v^+\in V(S^+)$ be a vertex assigned when step 4 processes $S$, with some index $i$ where $P_i \preceq \pi(v^+)$ and $\pi(v^+) \preceq P_{i+1}$ (if $P_{i+1}$ exists).
\begin{enumerate}
    \item[(a)] If $v^+$ is assigned to a net point of $S$, then it is assigned to either $P_i$ or $P_{i+1}$.
    \item[(b)] Otherwise, $v^+$ is assigned to some net point $P_a$ of $S_a$, where $S_a$ is the parent of $S$. In this case, there is a path $\Gamma$ in $H$, consisting of vertices assigned to $P_a$, from $v^+$ to some vertex $v_a^+ \in V(S^+_a)$. Moreover, if $v^+$ is assigned during step 4(a) when some net point $P$ is processed, then $v_a^+ \in B_P^+$. If $v^+$ is assigned during step 4(b), then either $v_a^+ \in B_P^+$ for some $P \in \set{P_i, P_{i+1}}$, or $v_a^+$ is in the same connected component of $V(S^+) \setminus X^+$ as $v^+$.
\end{enumerate}
\end{claim}
\begin{proof}
If $v^+ \in S^+$ is assigned in step 4(a), then $v^+\in V(B_P^{+})\setminus V(B_P)$ for some net point $P \in \{P_i, P_{i+1}\}$.
As in \Cref{lem:parent-only}, let us consider what assignments can exist in $B_P^+\setminus B_P$ before step 4(a) was run on supernode $S$. We call this the old assignments. We claim that every vertex in $B_P^+$ with an old assignment is either assigned to $P$ or to some net point in $S_a$, the parent of $S$. Indeed,
by \Cref{obs:clusterouter}(c), $B_P^+ \cap V(S) \subseteq B_P$, so $P$ is the only net point in $S$ that have old assignments in $B_P^+$.
Other than this, vertices in $B_P^+$ could have been assigned to net points of an ancestor supernode $S_a$, which \Cref{lem:parent-only} shows it must be the parent supernode of $S$. Observe that if a vertex has an old assignment to a net point of $S_a$, then that vertex must be in $V(S^+_a)$.
Thus, when \textsc{Assign} is called while constrained within $V(B_P^+)$, either $v^+$ is assigned to $P$, or there is a path in $B^+_P$ from $v^+$ to some vertex $v_a$ of $V(S^+_a) \cap V(B^+_P)$ with an old assignment to some net point $P_a$ of $S_a$ (and this path consists of vertices assigned to $P_a$). This proves the claim in the case that $v^+$ is assigned during step 4(a).

Now consider the case that $v^+ \in V(S^+)$ is assigned in step 4(b). In this case, $v^+\in V(S^+)\setminus X^+$. 
We now consider the old assignments in $V(S^+)\setminus X^+$ before step 4(b). These old assignments are to either a net point of $S$ (for vertices of $V(S)$), or to a net point of $S_a$ (for vertices of $V(S^+_a)$).
If $v^+$ is not assigned to a net point of $S_a$, then there must be a path in $H[V(S^+)]$ from $v^+$ to a net point $P$ on $S$ (where every vertex on the path is assigned to the same net point). However, by \Cref{lem:sandwich}, such a path must run into a vertex in $B^+_{P'}$ for $P'\in \{P_i, P_{i+1}\}$. Such a vertex, if it is assigned to a net point of $S$, must be assigned to $P'$ as we have just argued. So in this case, $v^+$ is assigned to some $P \in \set{P_i, P_{i+1}}$.
On the other hand, if $v^+$ is assigned to a net point of $S_a$, then it must have a path within $H[V(S^+)]$ to a vertex of $V(S_a^+)$ either entirely in $V(S^+)\setminus X^+$, or through a vertex of $B^+_{P'}$ assigned in step 4(a) that has a path to a vertex of $V(S^+_a) \cap B^+_{P'}$. In either case, there is a path $\Gamma$ in $H$ from $v^+$ to a vertex of $V(S^+_a)$ that is either in $B^+_{P'}$ or is in the same connected component of $V(S^+) \setminus X^+$ as $v^+$.
\end{proof}

\subsubsection{Diameter bound within a cluster}

We introduce some new terminology, to simplify the statements of our lemmas. Let $v \in V(H)$ be a vertex assigned to some cluster $C_P$, where $P$ is a net point in supernode $S$. We say that $v$ is \EMPH{$t$-close} to $P$ if every region $R \in \dom(S)$ that contains $v$ satisfies $\dist_{\dom(S)}(R, P)\le t$, and there is at least one region $R \in \dom(S)$ that contains $v$.

\begin{claim}
\label{clm:diam-3a}
At the end of step 3(a), every vertex $v\in V(S)$ assigned to a cluster $C_P$ is $12$-close to $P$, where $P$ is a net point of $S$.
\end{claim}
\begin{proof}
If $v$ is assigned to $C_P$ in step 3(a), then $v \in R$ for some region $R \in B_P^{} \subseteq \dom(S)$. By definition of $B_P$, we have $\dist_{\dom(S)}(R, P) = \dist_{S}(R,P) \le 11$. 
There may be other regions containing $v$, but all such regions are adjacent to $R$, so we conclude that $v$ is $12$-close to $P$.
\end{proof}

\begin{claim}\label{clm:same-supernode}
At the end of step 3(b), every vertex $v \in V(S)$ assigned to a cluster $C_P$ with $P$ a net point of $S$, and $v$ is $27$-close to $P$.
Moreover, if at any time a vertex $v^+\in V(S^+)$ is assigned to a cluster $C_P$ with $P$ a net point of $S$, then $v^+$ is $28$-close to $P$.
\end{claim}
\begin{proof}
First we consider a vertex $v\in V(S)$. This vertex is assigned during step 3.
Let $i$ be an index such that $P_i \preceq \pi(v)$ and $\pi(v) \preceq P_{i+1}$. (We assume for simplicity that $P_{i+1}$ exists. The case where $P_{i+1}$ does not exist is similar.)
By \Cref{clm:assigned-locations-3}, $v$ can only be assigned to either $P_i$ or $P_{i+1}$.
We will bound how close $v$ is to $P_i$; the same argument will also work for $P_{i+1}$. 

As $\pi$ is a shortest path in $\dom(S)$, we have that by definition of the net points
\[\dist_{\dom(S)}(\pi(v), P_i) \le \dist_{\dom(S)}(P_{i+1}, P_i) = 24.\]
Let $R^*_v$ be any region in $S$ such that $\dist_{\dom(S)}(R^*_v, \pi(v))\le \alpha$ where $\alpha=2$; such a region exists by \Cref{obs:dist-to-proj}. Any region $R_v$ containing $v$ is adjacent to the region $R^*_v\in S$, so $\dist_{\dom(S)}(R_v, \pi(v)) \le \dist_{\dom(S)}(R^*_v, \pi(v))+1 \le \alpha+1$. Thus by the triangle inequality we conclude that
\[ \dist_{\dom(S)}(R_v, P_i) \le \dist_{\dom(S)}(R_v, \pi(v)) + \dist_{\dom(S)}(\pi(v), P_i) \le 25+\alpha = 27.\]
Thus we conclude that $v$ is $27$-close to $P_i$ (and similarly $P_{i+1}$), so $v$ is $27$-close to whichever net point it is assigned to.

Next we consider a vertex $v^+\in V(S^+)$ assigned to a net point $P$ of $S$ during step 4.
By \Cref{lem:parent-only} and the fact that we process ancestor supernodes before their descendants, vertex $v^+$ must be assigned to $S$ when step 4 processes supernode $S$.
Suppose, as before, that there is some index $i$ where $P_i \preceq \pi(v^+) \preceq P_{i+1}$. (We assume for simplicity that $P_{i+1}$ exists.)
\Cref{clm:assigned-locations-4}(a) shows that $P\in \{P_i, P_{i+1}\}$.
By the same proof with $\alpha=3$, we can show that if $v^+\in V(S^+)$ is assigned to a net point $P$ of $S$, then $v$ is $28$-close to $P$.
\end{proof}

If a vertex $v^+\in V(S^+)$ is not assigned to a net point in $S$, then by \Cref{lem:parent-only}, it must be assigned to a net point $P$ of $S_a$ in step 4. We now consider the diameter bound for such vertices.

\begin{claim} \label{clm:dif-supernode}
Let $v^+\in V(S^+)$ be a vertex that is assigned to a net point $P_a$ of supernode $S_a$, the parent supernode $S$. Then $v^+$ is $63$-close to $P_a$.
\end{claim}

\begin{proof}
If $v^+ \in V(S_a^+)$, then the claim follows immediately from \Cref{clm:same-supernode}. So we consider the case where $v^+ \in V(S^+) \setminus V(S_a^+)$, meaning that $v^+$ is assigned when Step 4 processes supernode $S$.
By \Cref{clm:assigned-locations-4}(b), there is a path 
in $H$ from $v^+$ to a vertex $v^a\in V(S^+)\cap V(S^+_a)$, where every vertex in the path is assigned to $P_a$. We consider two cases based on whether $v^+$ is assigned to $P_a$ in step 4(a) or in step 4(b).

\textit{Case 1:} If $v^+$ is assigned in step 4(a), then \Cref{clm:assigned-locations-4}(b) implies there exists a net point $P$ on $S$ such that both $v^+$ and $v^a$ are in $B^+_P$.
In this case, both vertices are $7$-close to $P$ (as there exists regions containing $v$ and $v^a$ with distance $6$ to $P$). Moreover, $v^a$ is $28$-close to $P_a$, by \Cref{clm:same-supernode}. Thus we conclude that for any region $R^+\ni v^+$ and any region $R^a\ni v^a$, 
\[ \dist_{\dom(S_a)}(R^+,P_a) \le \dist_{\dom(S_a)}(R^+,P) + \dist_{\dom(S_a)}(P,R^a)  + \dist_{\dom(S_a)}(R^a,P_a) \le 7+7+28 = 42.  \]

\textit{Case 2:} If $v^+$ is assigned in step 4(b), then $v^+ \in V(S^+) \setminus X^+$ (recall that $X^+ = V(\bigcup_{P\in \cP_S} B^+_P)$).
Now consider the unique $i$ where $P_i\preceq \pi(v^+)$ and $\pi(v^+) \preceq P_{i+1}$ (if $P_{i+1}$ exists).
Consider any region $R^+\ni v^+$ and any region $R^a\ni v^a$.
By \Cref{clm:assigned-locations-4}(b), either there exists $P'\in \{P_i, P_{i+1}\}$ such that $v^a$ is in $V(B^+_{P'})$, or $v^a$ is in the same component of $V(S^+) \setminus X^+$ as $v^+$.
In the former case, observe that $v^a$ is $7$-close to $P'$ (because there is some region containing $v^a$ with distance 6 to $P$), and also $\dist_{\dom(S_a)}(R^+, \pi(v^+) \le 4$ (because there is some region in $S^+ \subseteq \dom(S_a)$ that contains $v^+$ and is within distance 3 of $\pi(v^+)$, by \Cref{obs:dist-to-proj}). Thus we have:
\[ \dist_{\dom(S_a)}(R^+,P_a) \le \dist_{S^+}(R^+,\pi(v^+)) + \dist_{\pi}(\pi(v^+), P') + \dist_{B^+_{P'}}(P',R^a)  + \dist_{\dom(S_a)}(R^a,P_a) \le 4+24+7+28 = 63.  \]
In the latter case, since $v^+$ and $v^a$ are in the same component of $V(S^+)\setminus X^+$, so $\dist_\pi(\pi(v^+), \pi(v^a))\le 24$. This means that 
\[\dist_{S^+}(R^+, R^a) \le \dist_{S^+}(R^+,\pi(v^+)) + \dist_{\pi}(\pi(v^+), \pi(v^a)) + \dist_{S^+}(\pi(v^a),R^a) \le 4 + 24 + 3 = 31.
\]
Thus we have that
\[ \dist_{\dom(S_a)}(R^+,P_a) \le  \dist_{S^+}(R^+, R^a) + \dist_{\dom(S_a)}(R^a,P_a) \le 31+28 = 59.  \]
In any case we conclude that $v^+$ is $63$-close to $P_a$.
\end{proof}

Claims~\ref{clm:same-supernode} and \ref{clm:dif-supernode} together yield the following corollary.
\begin{corollary}
\label{cor:close-bound}
    For every vertex $v \in V(H)$ that is assigned to a cluster $C_P \in \cC$ during \textsc{ClusterOuter}, $v$ is $63$-close to $P$.
\end{corollary}

Now we can bound the $\cR$-diameter of every cluster.

\begin{lemma}
\label{lem:gridtree-diam}
    \textnormal{[Diameter.]}
    For every $C_P \in \cC$, the cluster $C_P$ has $\cR$-diameter at most $126$.
\end{lemma}
\begin{proof}
We prove that $C_P$ has $\cR^*$-diameter at most $126$; by \Cref{obs:shattered-diam}, this implies that $C_P$ has $\cR$-diameter at most $126$.
Let $R_a, R_b \in \cR^*$ be two regions that intersect $C_P$.
Let $r_a \in R_a \cap C_P$ and $r_b \in R_b \cap C_P$. By \Cref{cor:close-bound}, vertex $r_a$ is $63$-close to $P$: in particular, $\dist_{\cR^*}(R_a, P) \le 63$. Similarly, $\dist_{\cR^*}(R_b, P) \le 63$.
Thus $\dist_{\cR^*}(R_a, R_b) \le 63+ 63 = 126$, as desired.
\end{proof}

\subsubsection{Scattering bound between clusters}
The goal of this section is to prove the scattering bound, which we restate here.
\begin{lemma}
\label{lem:gridtree-scattering}
    \textnormal{[Scattering.]} Let $\beta = 57$ and $\gamma = 12$.
    For every region $R \in \cR$ and any two vertices $v_a, v_b \in R$,
    there is a path $\Pi$ in $G$ between $v_a$ and $v_b$ that has one of the following two forms: either $\Pi$ is the concatenation of at most $\beta$ paths
    \[\Pi = \Pi_1 \circ \Pi_2 \circ \ldots \circ \Pi_\beta\]
    where each subpath $\Pi_i$ is contained entirely in some cluster $C \in \cC$; or $\Pi$ is the concatenation of at most $2\gamma + 1$ paths
    \[
    \Pi = \Gamma_1 \circ \Pi_1 \circ \Gamma_2 \circ \Pi_2 \circ \ldots \circ \Pi_{\gamma} \circ \Gamma_{\gamma + 1}
    \]
    where each $\Pi_i$ is contained entirely in some cluster $C\in \cC$,
    and each $\Gamma_i$ consists of unassigned vertices in $R$.
    Moreover, if $R$ is an outer region, then there is always a path $\Pi$ of the first type (ie, comprising $\beta$ subpaths contained in clusters).
\end{lemma}

To begin with, we consider regions $R^* \in \cR^*$ (instead of $R \in \cR$) and prove an analogous statement for such regions.

\begin{lemma}
\label{lem:constant-candidates}
    Let $R^* \in \cR^*$. There is a set of $12$ clusters such that every vertex in $R^*$ either is unassigned or is assigned to one of these $12$ clusters.
\end{lemma}

\begin{proof}

First suppose that $R^*$ is in some supernode, that is $R^* \in S$ for some $S \in \cS$.
In this case, every vertex of $R^*$ is assigned during step 3 to some net point of $S$, by \Cref{clm:assigned-locations-3}.
Suppose two vertices $u,v \in R^*$ are assigned to net points $P_u, P_v$ of $S$, respectively.
By \Cref{clm:same-supernode}, $v$ is $27$-close to $P_v$ and $u$ is $27$-close to $P_v$. By the triangle inequality, this means that $\dist_{S}(P_u, P_v) \le 27 + 1 + 27 = 55$. However the distance between $4$ consecutive net points is $3\cdot 24 = 72 > 55$, so vertices of $R^*$ can be assigned to at most $3$ different (consecutive) net points of $S$.

Now suppose $R^* \not \in \bigcup \cS$. 
Let $S, S_a \in \cS$ be the two supernodes given by \Cref{obs:adj-to-two}(b) that (intuitively) ``enclose'' the region $R^*$. We claim that each vertex in $R^*$ is either unassigned or is assigned to a net point of $S$ or $S_a$. Indeed, suppose for contradiction that some $v \in R^*$ is assigned to a net point $P$ in a supernode $S' \not \in \set{S, S_a}$. By \Cref{obs:clusterouter}(e), there is a path in $H$ between $v$ and a vertex in $P$ (thus, between $v$ and a vertex in $V(S')$) that consists of vertices assigned to $P$. By \Cref{obs:adj-to-two}(b), any such path intersects $V(S \cup S_a)$. But \Cref{clm:assigned-locations-3} says that every vertex of $V(S \cup S_a)$ is assigned to a net point of $S$ or $S_a$, not to a net point of $S'$, a contradiction.

This proves that every vertex of $R^*$ is assigned to a net point of $S$ or $S_a$. Moreover, one can show that all vertices of $R^*$ are assigned to a set of 6 consecutive net points on $S$ and 6 consecutive net points on $S_a$. The reasoning is similar to earlier: if vertices $u$ and $v$ in $R^*$ are assigned to net points $P_u$ and $P_v$ on supernode $S$ (or, symmetrically, on $S_a$), then \Cref{cor:close-bound} implies that $\dist_S(P_u,P_v) \le 63 + 1 + 63 \le 127$, meaning $P_u$ and $P_v$ are drawn from a set of 6 consecutive net points in $S$ (or, symmetrically, on $S_a$). In total, vertices of $R^*$ can be assigned to at most different $12$ clusters.
\end{proof}

For any region $R$, we say a path $\Gamma$ in $H$ is \EMPH{$R$-safe} if either there exists some cluster $C \in \cC$ such that $\Gamma$ contain only vertices in $C$, or $\Gamma$ contains only unassigned vertices in $R$.

\begin{lemma}
\label{lem:single-scattering}
    Let $R^* \in \cR^*$. 
    For any two vertices $v_s, v_t \in R^*$, there is a path $\Gamma$ in $H$ between $v_s$ and $v_t$ that can be written as the concatenation of at most $25$ subpaths which are all $R^*$-safe.
    Moreover, if $R^*$ is in $\bigcup \cS$, or if $R^*$ is adjacent in $H_\cR^*$ to a region in $\bigcup \cS$, then all vertices of $R^*$ are assigned, and $\Gamma$ can be written as the concatenation of at most $12$ subpaths, each of which is contained in a cluster.
\end{lemma}

\begin{proof}
Let $\Gamma'$ be an arbitrary path in $H[R^*]$ from $v_s$ to $v_t$. By \Cref{lem:constant-candidates}, vertices on $\Gamma'$ belong to at most $12$ clusters. However, we can't just set $\Gamma = \Gamma'$ because it may be the case that $\Gamma'$ wiggles back-and-forth between two clusters many times. To fix this issue, we construct $\Gamma$ as the concatenation of many subpaths $\Gamma_i$ by the following \EMPH{greedy chopping argument}.
\begin{quote}
    Initialize $v_1 \gets v_s$. For any $i \ge 1$, we maintain the invariant that $v_i$ is some vertex on $\Gamma'$. Define \EMPH{$\Gamma_i$} as follows.
    If $v_i$ is assigned to some cluster $C_{P_i}$, then let $v_i'$ be the \emph{last} vertex along $\Gamma'$ that is assigned to $C_{P_i}$, and define $\Gamma_i$ to be a path between $v_i$ and $v_i'$ in $H[C_{P_i}]$; such a path exists by \Cref{obs:clusterouter}(e). On the other hand, if $v_i$ is unassigned, then we define \EMPH{$\Gamma_i$} be the maximal subpath of $\Gamma'$ that starts at $v_i$ and contains only unassigned vertices.
    Define \EMPH{$v_{i+1}$} to be the vertex that follows the last vertex $\Gamma_i$ on $\Gamma'$; if $\Gamma_i$ ends at $v_t$, terminate the process.
\end{quote}
Clearly, each subpath $\Gamma_i$ either consists entirely of unassigned vertices in $R^*$ or consists entirely of vertices in a single cluster. The process terminates after $2 \cdot 12 +1 = 25$ subpaths are found (because the two subpaths at the ends might be both unassigned).

Moreover, if $R^*$ is adjacent to a supernode (or is in a supernode),
then Step 2 of \textsc{ClusterOuter} places $R^*$ in some expanded supernode $S^+$.  By \Cref{obs:clusterouter}(d), every vertex of $V(S^+)$ is assigned to a cluster, hence every vertex of $R^*$ is assigned.  
Thus, an arbitrary path $\Gamma'$ in $H[R^*]$ from $v_s$ to $v_t$ contains no unassigned vertices, and so the greedy chopping argument gives at most 12 subpaths.
\end{proof}

\Cref{lem:single-scattering} is essentially a scattering bound for regions $R^* \in \cR^*$. To generalize to regions $R \in \cR$, we need to cope with the fact that $R$ could get split into many subregions in $\cR^*$. We exploit the fact (\Cref{lem:select-paths-scattering}) says that any region $R \in \cR$ gets shattered into subregions of $\cR^*$ that are close to each other in $H_{\cR^*}$.

\begin{lemma}
\label{lem:supernode-scattering}
    Let $Y$ be a walk in $H_{\cR^*}$ with length at most 6, where every region in $Y$ is contained in supernode $S$. For any two vertices $v_s$ and $v_t$ in two regions of $Y$, there is a path in $H$ between $v_s$ and $v_t$ that is the concatenation of at most $3$ subpaths, each of which is contained in a cluster of $S$.
\end{lemma}

\begin{proof}

We first prove that there is a set of 3 clusters such that all vertices in regions of $Y$ are assigned to one of those three clusters. To this end,
let $v_1$ and $v_2$ be two arbitrary vertices in two regions of $Y$, say $R_1$ and $R_2$.
By \Cref{clm:same-supernode}, $v_1$ (resp. $v_2$) is assigned to some net point $P$ (resp. $P'$) of $S$, where $\dist_{\dom(S)}(R_1, P) \le 27$ (resp. $\dist_{\dom(S)}(R_2, P') \le 27$).
As $R_1$ and $R_2$ are both in $Y$, we have $\dist_S(R_1,R_2) \le 6$, and thus \[
    \dist_{\dom(S)}(P,P') \le 27+6+27=60.
\]
In other words, if any two vertices in regions of $Y$ are assigned to net points $P$ and $P'$, then $\dist_{\dom(S)}(P, P') \le 60$.
Since net points are spaced by $24$, a set of four net points would have endpoints at distance at least $3 \cdot 24 = 72$, impossible.  Hence vertices in regions of $Y$ are assigned to at most different three clusters. the greedy chopping argument from the proof of \Cref{lem:single-scattering} yields a path that comprises at most three subpaths, each of which is contained in one cluster.
\end{proof}

We are now ready to prove \Cref{lem:gridtree-scattering}.
\begin{proof}[of \Cref{lem:gridtree-scattering}]

Let $R \in \cR$ be an arbitrary region and let $v_a, v_b \in R$ be two vertices. We aim to construct a path $\Pi$ in $H$ between $v_a$ and $v_b$.
Let $R_a, R_b \in \cR^*$ be subregions of $R$ containing $v_a$ and $v_b$, respectively. 
We deal with two cases: either $R_a = R_b$, or $R_a \neq R_b$. These two cases will produce the two different forms of $\Pi$ in the statement of the lemma.

In the case when $R_a=R_b$,
by \Cref{lem:single-scattering} and the chopping argument within, 
there is a path $\Pi$ in $G$ between $v_a$ and $v_b$ that can be written as the concatenation of at most $25$ paths
\[
\Pi = \Gamma_1 \circ \Pi_1 \circ \Gamma_2 \circ \Pi_2 \circ \ldots \circ \Pi_{12} \circ \Gamma_{13}
\]
where each subpath $\Pi_i$ is contained entirely in some cluster $C\in \cC$,
and each $\Gamma_i$ consists of unassigned vertices in $R$.

In the case when $R_a \ne R_b$,
by \Cref{lem:select-paths-scattering}, 
there is a walk $Y$ between $R_a$ and $R_b$ in the contact graph of $\cR^*$ such that 
$Y = [R_a] \circ Y_1 \circ [\tilde R_a] \circ Y_2 \circ [\tilde R_b] \circ Y_3 \circ [R_b]$, where each subwalk $Y_1,Y_2,Y_3$ has length at most $6$ and belongs to a single supernode.
For each subwalk $Y_i \in \set{Y_1, Y_2, Y_3}$, \Cref{lem:supernode-scattering} implies that any two vertices in regions of $Y_i$ can be connected by a path in $H$ that uses at most 3 subpaths, each contained in a single cluster.
Further, each of the endpoint regions $R_i \in \set{R_a, \tilde R_a, \tilde R_b, R_b}$ is adjacent to a supernode, so \Cref{lem:single-scattering} implies that any two vertices in $R_i$ can be connected by a path in $H$ that requires at most 12 subpaths.  
By definition of contact graph, for every consecutive pair $[R_i] \circ Y_i$ in the walk $Y$, there exists some vertex in $R_i$ that is adjacent (in $H$) to some vertex in a region of $Y_i$.
Overall, there is a path in $H$ between $v_a$ and $v_b$ that uses at most
$4\cdot 12+3\cdot 3=57$ subpaths.

Finally, we consider the special case where $R \in \cR$ is an outer region: in this case, we need to prove there is a path $\Pi$ between $v_a$ and $v_b$ that comprises at most 57 subpaths, each of which is contained in a cluster. If $R_a \neq R_b$, then we are done. Otherwise, if $R_a = R_b$, we observe that \Cref{lem:external-neighbor} implies that $R_a$ is adjacent in $H_{\cR^*}$ to some region in a supernode. Thus, \Cref{lem:single-scattering} yields a path that can be written as the concatenation of at most $12 < 57$ subpaths, each of which is contained in a cluster. \qed

\end{proof}

\Cref{lem:gridtree} follows from \Cref{lem:gridtree-outer-clustered} (the [outer-clustered] property), \Cref{lem:gridtree-diam} (the [diameter] property), and \Cref{lem:gridtree-scattering} (the [scattering] property).

\section{\texorpdfstring{\boldmath{$(1+\e, +\poly\log n)$}-Distortion Planar Emulator}{(1+ε, +poly(log n))-Distortion Planar Emulator}}
\label{S:multToAdd}

In this section we prove the second main theorem, \Cref{thm:main2}.
In fact we prove a more general theorem, which applies to any graph with a $O(1)$-distortion planar emulator. Our \Cref{thm:main2} follows from combining \Cref{thm:multToAdd} with the planar emulator of \Cref{cor:small-emulator}.

\begin{theorem}
\label{thm:multToAdd}
Let $\mathcal{G}$ be a class of unweighted graphs that is closed under induced subgraphs, such that every graph $G \in \mathcal{G}$ has an $O(1)$-distortion planar emulator $H_{\rm init}$ that has non-negative integer edge weights and $V(H_{\rm init})=V(G)$.
For any $n$-vertex graph $G \in \mathcal{G}$ and any small constant $\e_0 > 0$,
    there is a weighted planar graph $H$ with $O(n)$ vertices and $V(G) \subseteq V(H)$, such that for every pair of vertices $(u,v)$ in $G$,
    \[
    \dist_G(u,v) \le \dist_H(u,v) \le (1+\e_0) \cdot \dist_G(u,v) + O(\e_0^{-4}\cdot \log^{17} n).
    \]
    Moreover, if one is given the $O(1)$-distortion planar emulator, then the graph $H$ can be computed in polynomial time.
\end{theorem}
We remark that $H$ is essentially a \emph{reweighting} of the $O(1)$-distortion planar emulator for $G$ (though we do add a few Steiner points to guarantee the non-contracting property.)

To prove our second main theorem, we present the following lemma that can be combined with the path-straightening technique of \cite{nguyen2025asymptotic} to achieve a $(1+\e_0, +\poly\log n)$-distortion embedding from string graphs to planar graphs.
\brokenpaths*

We prove \Cref{lem:string} in \Cref{SS:C-disjoint-paths}, and use \Cref{lem:string} to establish the existence of $(1+\e, +\poly\log n)$-distortion planar emulator and prove \Cref{thm:main2} in Section~\ref{SS:approx-reduction}.

\subsection{Shortest-path separators, portals, and canonical pairs} 
\label{SS:canonical-pairs}
\paragraph{Shortest path separators.} A powerful tool in planar graph algorithms is the existence of \EMPH{shortest path separator} \cite{lipton1979separator,thorup2004compact}: in any $n$-vertex planar graph $G$, there is a collection of 3 shortest paths whose removal breaks $G$ into connected pieces with at most $n/2$ vertices. Applying shortest path separator recursively yields a hierarchy of pieces, with height at most $O(\log n)$. It is clear that arbitrary string graphs do not have separators consisting of few shortest path: indeed, the complete graph $K_n$ is a string graph. Nevertheless, one could hope that string graphs have a separator which consists of an $O(1)$-neighborhood around a small number of shortest paths. Such a separator theorem has been obtained for unit-disk graphs \cite{YXD12,HZ24}. We use our $O(1)$-distortion planar emulator to generalize the result to arbitrary string graphs.

\begin{definition}
\label{def:separator}
Let $G$ be a graph. A \EMPH{$D$-neighborhood path separator hierarchy} of $G$ is a tree $\cT$ such that:
\begin{itemize}
    \item The nodes of $\cT$ are induced subgraphs \EMPH{$R$} of $G$ called \EMPH{pieces}%
    \footnote{Commonly, these subgraphs are called ``regions''. Unfortunately this choice of terminology conflicts with the regions discussed in the previous section, so we use the term ``piece'' instead. We commonly use the variable $R$ to denote a piece.}.
    Each piece is associated with 
    a set of vertices \EMPH{$S$}
    called an \EMPH{internal separator}.
    The set of internal separators of nodes that are proper ancestors of $R$ are called the \EMPH{external separators} of $R$. 
    \item The internal separator $S$ lies in the union of $R$ and all external separators of $R$, denoted as \EMPH{$R^+$}.
    \item The internal separator $S$ comprises a single path \EMPH{$\pi_S$} and \ul{the $D$-neighborhood} of $\pi_S$ in $R^+$. The path $\pi_S$ is a shortest path \ul{in $R^+$}.%
    \item The pieces of the children of $R$ 
    form a vertex-partition of $R \setminus S$.
    In particular, if $R \setminus S = \varnothing$, then $R$ is a leaf of $\cT$.
    \item Any path in $G$ that contains a vertex in $R$ and a vertex not in $R$ must also contain a vertex in an external separator of $R$.
\end{itemize}
The \EMPH{height} of the separator hierarchy is the height of $\cT$.
\end{definition}
We emphasize a few unusual properties of the definition, compared to standard planar separator hierarchy. First, each internal separator can contain a $D$-neighborhood of $\pi_S$, rather than just $\pi_S$.  Second, $\pi_S$ is a shortest path in $R^+$ rather than $R$, and in fact $\pi_S$ and its $D$-neighborhood might not even lie in $R$. Third, the pieces of the children of $R$ are not necessarily connected components of $R \setminus S$ (they can be disjoint union of components); this difference is nonessential but slightly simplifies the construction. 
Also note that an external separator might not contain any vertices incident to $R$.
Fortunately, these peculiarities do not affect the distortion analysis much;
we can handle the slight inconvenience in \Cref{obs:nearby-portal} and the proof of \Cref{lem:canonical-pairs}.

\begin{figure}[th]
    \centering
    \includegraphics[width=0.6\linewidth]{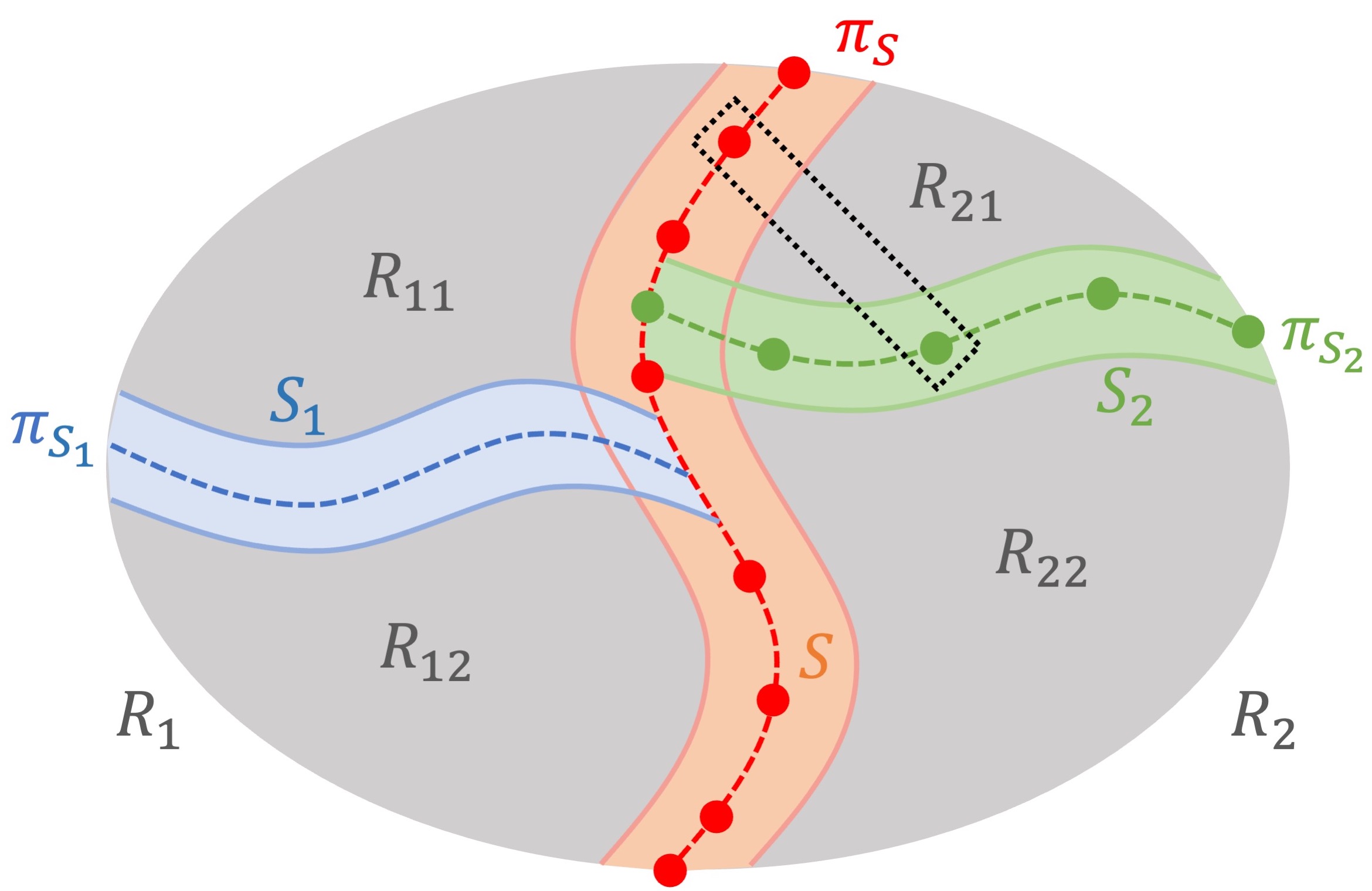}
    \caption{An illustration of a separator hierarchy: graph $G$ is separated by its internal separator $S$ into pieces $R_1$ (left) and $R_2$ (right); $R_1$ is further separated by $S_1$ into pieces $R_{11}$ and $R_{12}$; and $R_2$ is further separated by $S_2$ into pieces $R_{21}$ and $R_{22}$. A canonical pair of piece $R_{21}$ is shown in the dashed black box.}
    \label{fig: separator_hierarchy}
\end{figure}

We claim that any string graph $G$ has an $O(1)$-neighborhood path separator hierarchy; a proof can be found in Section~\ref{SSS:separator-hierarchy}.
The proof crucially relies on using \emph{path aggregation} of \cite{hathcock2025steiner} to pull back separators in the planar emulator to separators in $G$.

\begin{restatable}{lemma}{hierarchy}
\label{lem:separator-hierarchy}
Let $\mathcal{G}$ be a class of unweighted graphs that is closed under induced subgraphs, such that every graph $G \in \mathcal{G}$ has a $C$-distortion planar emulator $H$ with $V(H)=V(G)$.
    For any $n$-vertex graph $G \in \mathcal{G}$, there is a $2C$-neighborhood separator hierarchy of $G$ with height $O(\log^2 n)$.
\end{restatable}

\paragraph{Portals.}
Given a $2C$-neighborhood separator hierarchy, we define \emph{portals} along each internal separator path $\pi_S$, similar to \cite{thorup2004compact} in planar graphs.
We fix some $\EMPH{$\e$} \in (0,1)$, specified later.
Let $R$ be a piece, and let $S$ be the internal separator of $R$, and let $\pi_S$ be the path associated with $S$.  Recall that $\pi_S$ lies in $R^+$ but might not be in $R$.
For every $i \ge 0$, we now define \EMPH{scale-$i$ portals in $R$} --- for short, the \EMPH{$(i, R)$-portals}. For $i = 0$, we define the $(0, R)$-portals to be the set of all vertices on $\pi_S$. For $i > 0$, we construct the $(i, R)$-portals to be a set \EMPH{$\cP_i$} which is a subset of the $(i-1, R)$-portals $\cP_{i-1}$, using a greedy approach:
initialize $\cP_i \gets \varnothing$, then walk along the path $\pi_S$ and add each vertex $v \in \cP_{i-1}$ to $\cP_i$ if 
$v$ is at distance greater than $\frac{\e}{2} \cdot 2^i$ from the current $\cP_i$.
Observe that one endpoint of $\pi_S$ (the ``starting'' endpoint) is in $\cP_i$; for ease of notation later, add the other endpoint of $\pi_S$ (the ``final'' endpoint) to $\cP_i$.
The portals  $\cP_i$ satisfy:
\begin{enumerate}
    \item[(1)] every vertex $s$ in $\cP_{i-1}$ has $\dist_{\pi_S}(s, \cP_i) \le \frac{\e}{2} \cdot 2^i$, and
    \item[(2)] for every pair of  distinct vertices $p_1, p_2 \in \cP_i$ (excluding the case where $p_1$ or $p_2$ is the final endpoint of $\pi_S$), we have $\dist_{\pi_S}(p_1, p_2) > \frac{\e}{2} \cdot 2^i$.
\end{enumerate}
The \EMPH{scale-$i$ portals} refer to the set of all $(i,R)$-portals over all pieces $R$ in $\cT$.
We remark that portals at scale $i$ will be responsible for preserving distances around $2^i$, up to a $(1+\e)$ factor. As such, they are separated by distance roughly $\e \cdot 2^i$ (not by distance $2^i$).

\paragraph{Canonical pairs and paths.}
The notion of \emph{canonical pair} was introduced by \cite{chang2025distance} to provide a $O(n\poly\log n)$-size set of shortest paths to approximate distances in a planar graph; somewhat similar ideas also appeared in previous work \cite{abraham2012fully}. A \EMPH{canonical pair at scale $i$ in piece $R$}, also called an \EMPH{$(i, R)$-pair} for short, is a pair of vertices denoted \EMPH{$\canon a b$}
where (1) $a$ is a scale-$i$ portal on the internal separator path of $R$, and $b$ is a scale-$i$ portal on some external or internal separator path $\pi_S$ of $R$; and (2) $\dist_{R^+}(a,b) \le 2^i$.
We define $(i,R)$-pairs for every scale $i$ in the set $\set{0, 1, \ldots, \lceil \log n \rceil}$\footnote{recall that we focus on unweighted graphs, so the minimum distance in $G$ is $1 = 2^0$ and the maximum distance is at most $n-1 \le 2^{\lceil \log n \rceil}$} and every piece $R$ in the separator hierarchy.

We define a set of paths \EMPH{$\Pi$}, called the \EMPH{canonical paths}, as follows.
For every $(i,R)$-pair $\canon a b$, let \EMPH{$\pi$} be a shortest path in $R^+$ between $a$ and $b$, and add $\pi$ to $\Pi$.
Note that there may be multiple shortest paths between $a$ and $b$, and we may arbitrarily choose one such path as $\pi$, except in the following case: if $a$ and $b$ both lie on the internal separator path $\pi_S$ of $R$, then we must choose $\pi$ to be the subpath $\pi_S[a:b]$ (by definition of separator hierarchy, $\pi_S$ is indeed a shortest path between $a$ and $b$).\footnote{We mention a small subtlety. Our definition of $\pi$ depends on the piece $R$ of the canonical pair $\canon a b$. A priori, it may be the case that the pair of vertices $(a,b)$ is also an $(i',R')$-pair for some other scale $i' \neq i$ and possibly different $R' \neq R$. We treat the $(i,R)$-pair $\canon a b$ as different from the $(i', R')$-pair $\canon a b$; that is, we add one path $\pi$ to $\Pi$ when treating $(a,b)$ as an $(i,R)$-pair and another path $\pi'$ to $\Pi$ when treating $(a,b)$ as an $(i', R')$-pair.}
We call $\pi$ an \EMPH{$(i,R)$-path}.
We let \EMPH{$\Pi$} be the set of paths $\cpath a b$ across all canonical pairs $\canon a b$.

\begin{lemma}[Canonical Pairs Lemma]
\label{lem:canonical-pairs}
   Let $G$ be a graph with an $2C$-neighborhood separator hierarchy. For any two vertices $u,v \in V(G)$,
   there is a path $P$ in $G$ between $u$ and $v$ with $\len{P} \le (1+ O(\log n) \cdot \e) \cdot \dist_G(u,v) + O(C \cdot \log n)$%
   , such that $P$ can be written as the concatenation of $O(\log n)$ subpaths: 
   two of these subpaths have length at most $O(\e \cdot \dist_G(u,v) + C)$, and each of the other subpaths $P_j$ is a shortest path in $R_j^+$ between some $(i,R_j)$-pair with $2^i = \Theta(\dist_G(u,v))$.
\end{lemma}
The proof is deferred to \Cref{S:canonical-pairs}.
It is essentially a simplified version of \cite[Claim 3.2]{chang2025distance}, but with $2C$-neighborhood separator hierarchy instead of normal separator hierarchy.

\subsection{Constructing $\Theta(C)$-disjoint broken paths}
\label{S:broken-paths}

This section begins the proof of \Cref{lem:string}. Let $G$ be a graph with a $C$-distortion planar emulator, and let \EMPH{$\e_0$} be a parameter in $(0,1)$. We construct a set of paths \EMPH{$\cO$} in graph $G$, none of which are within distance $2C$ of each other.
In the next section, we prove that these paths $\cO$ satisfy the [low-hop] property of \Cref{lem:string}.

We begin by setting $\EMPH{$\e$} \coloneqq \frac{\e_0}{\Theta(\log^3 n)}$. 
By \Cref{lem:separator-hierarchy}, there is a $2C$-neighborhood separator hierarchy for $G$. Define the portals (with respect to $\e = \frac{\e_0}{\Theta(\log^3 n)}$), canonical pairs, and canonical paths, as in the previous section. Recall that \EMPH{$\Pi$} denotes the set of all canonical paths.

\begin{definition}
    Define a \EMPH{partial order $\preceq_*$} on paths in $\Pi$ as follows. 
    If $\pi \in \Pi$ is an $(i, R)$-path, and $\pi'\in\Pi$ is a $(i',R')$-path, we say $\pi \preceq_* \pi'$ if either (1) $R'$ is a proper ancestor of $R$ in the separator hierarchy, or (2) $R = R'$ and $i \le i'$.  
    In short, the partial order $\preceq_*$ is a lexicographical order based on piece containment followed by scale.
    
    Fix an arbitrary \EMPH{total order $\preceq$} on paths in $\Pi$ respecting the partial order $\preceq_*$.  We use $\prec$ to denote the strict total order.
\end{definition}

\paragraph{Construction of \boldmath{$\cO$}.} 
The construction of the path set $\cO$ is simple. Refer to \Cref{fig:erased_paths} for an illustration.

\begin{tcolorbox}%
\textul{$\textsc{ErasedPath}$}: \\
For every $(i,R)$-path $\pi \in \Pi$, define 
\EMPH{$\mathit{Expand}(\pi, 4C)$} to be the $4C$-neighborhood of $\pi$ \ul{in $R^+$}.\\
For every $\pi \in \Pi$, define the \EMPH{erased set} to be a subset of the vertices on $\pi$:
\[\textrm{Er}_\pi \coloneqq \pi \cap \bigcup_{\pi' \in \Pi \text{ with } \pi \prec \pi'} \mathit{Expand}(\pi', 4C) \]
and define the \EMPH{enlarged erased set} to be:
\[\textrm{EEr}_\pi \coloneqq \set{v \in \pi: \dist_\pi(v, \textrm{Er}_\pi) \le C}.\]
We add to \EMPH{$\cO$} the set of subpaths in $\pi \setminus \textrm{EEr}_\pi$.
(Note that $\pi$ could be shattered into arbitrarily many subpaths.)
\end{tcolorbox}

Note the subtlety that the enlarged erased set is defined using distances along the path $\pi$.
One way to visualize the construction is to draw the paths $\cpath a b$ down one by one according to the ordering $\preceq$, starting from the smallest path with respect to $\preceq$.
Before we draw the next path $\cpath a b$, we first ``erase'' any drawn section of a path intersecting the $C$-neighborhood \ul{in $R^+$} around $\cpath a b$ (where $R$ is the piece of $\cpath a b$).
Rinse and repeat until all paths in $\Pi$ are processed.  
At this point, parts of the originally-drawn shortest path 
have been erased, and the original path has turned into a collection of disjoint subpaths. 
Finally, enlarge all the ``erased'' segments of a path by expanding them a distance $C$ in either direction along that path. The resulting collection of disjoint subpaths is $\cO$.

\begin{lemma}
\label{lem:shortest-nearby}
    Let $P \in \cO$ be a path, and let $\pi \in \Pi$ be the $(i,R)$-path such that $P$ was added to $\cO$ because $P$ was a subpath of $\pi$. Then every vertex in $\cN_G(P, 2C)$ lies in $R^+$.
\end{lemma}
\begin{proof}
Suppose otherwise. Let $u$ be a vertex not in $R^+$ with $\dist_G(P,u) \le 2C$, and let $\Gamma$ be the shortest path in $G$ between $P$ and $u$.
Let $R_\Gamma$ be the highest piece
whose internal separator $S$ (but not necessarily $\pi_S$) contains a vertex in $\Gamma$.
Because $u$ is not in $R^+$, the definition of separator hierarchy implies that $\Gamma$ intersects some external separator of $R$; that is, $R_\Gamma$ is a proper ancestor of $R$ in the separator hierarchy.
Observe that the separator path $\pi_S$ is a path in $\Pi$ (indeed, by definition the endpoints of $\pi_S$ are portals at the maximum scale).
We conclude that $\pi \prec \pi_S$. We now show that $\dist_{R_\Gamma^+}(\pi_S, P) \le 4C$, which is a contradiction to the fact that $P$ is in $\cO$: some vertex of $P$ should have been erased by $\textit{Expand}(\pi_S, 4C)$ during the construction of $\cO$.
To this end, let choose an arbitrary vertex $v \in S \cap \Gamma$. Observe that $\Gamma$ is contained in $R_\Gamma^+$ (as otherwise $\Gamma$ would intersect some external separator of $R_\Gamma$, contradicting the choice of $R_\Gamma$), and so $\dist_{R_\Gamma^+}(v,P) \le \len \Gamma$. As $v$ is within distance $2C$ of $\pi_S$ in $R_\Gamma^+$ (by definition of $2C$-neighborhood separator hierarchy), we conclude:
\(
\dist_{R_\Gamma^+}(\pi_S, P) \le \dist_{R_\Gamma^+}(\pi_S, v) + \dist_{R_\Gamma^+}(v, P) \le 2C + \len \Gamma \le 4C.
\)
\end{proof}

The following corollary follows from the fact that $\pi$ (and hence $P$) is a shortest path in $R^+$.
\begin{corollary}
\label{cor:cO-shortest}
    For any path $P \in \cO$, we have that $P$ is a shortest path in the subgraph of $G$ induced by $\cN_G(P, 2C)$.
\end{corollary}

Finally, we prove $O(C)$-disjointness.
\begin{lemma}
\label{lem:C-disjoint-in-G}
For any two distinct paths $P_1, P_2 \in \cO$, we have $\dist_G(P_1,P_2) > 2C$.
\end{lemma}
\begin{proof}
By construction, $P_1$ (resp.\ $P_2$) is a subpath of some path $\pi_1 \in \Pi$ (resp.\ $\pi_2 \in \Pi$). Let $\Gamma$ be a shortest path in $G$ between $P_1$ and $P_2$, and we suppose for the sake of contradiction that $\len \Gamma \le 2C$. Suppose that $\pi_1$ (resp.\ $\pi_2$) is an $(i_1, R_1)$-path (resp.\ $(i_2, R_2)$-path).
There are two cases.

As the first case, suppose that $\pi_1 = \pi_2$. In this case, $2C$-disjointness will follow from the fact that we \emph{enlarged} the initial ``erased set''. For simplicity, let $\pi$ denote $\pi_1 = \pi_2$, and let $R$ denote $R_1 = R_2$. 
Because $P_1$ and $P_2$ are distinct, there is some vertex $v \in \mathrm{EEr}_\pi$ that lies between $P_1$ and $P_2$ along $\pi$. By definition of $\mathrm{EEr}_\pi$, there exists some $u \in \mathrm{Er}_\pi$ with $\dist_\pi(u,v) \le C$. But observe that \emph{every} vertex within distance $C$ (along $\pi$) from $u$ belongs to $\mathrm{EEr}_\pi$, and so $v$ belongs to a subpath of erased vertices with length at least $2C$. As the vertices of $P_1$ and $P_2$ are not erased, we conclude that every vertex in $P_1$ is at distance greater than $2C$ (in $\pi$) from every vertex in $P_2$. Because $\pi$ is a shortest path in $R^+$, we have $\dist_{R^+}(P_1,P_2) > 2C$, a contradiction.

As the second case, suppose that $\pi_1 \neq \pi_2$. In this case, $2C$-disjointness will follow from the initial ``erasing'' procedure.
Without loss of generality assume $\pi_1 \prec \pi_2$ in the total order on $\Pi$.
As $\Gamma$ has length at most $2C$, \Cref{lem:shortest-nearby} implies that $\Gamma$ is contained in $R_2^+$.
By construction, the subpaths of $\pi_1$ added to $\cO$ do not include any vertices in $\mathit{Expand}(\pi_2, 4C)$, ie does not include any vertices within distance $4C$ of $\pi_2$ in $R_2^+$. So $\len \Gamma > 4C$, a contradiction.
\end{proof}

The first two properties of \Cref{lem:string} follow from \Cref{cor:cO-shortest} and \Cref{lem:C-disjoint-in-G}.

\subsection{Detour analysis}
We now argue that, for any pair of vertices $(a,b)$ in $G$, there is
a ``broken'' path $P(a,b) =  Q_1 \circ P_1 \circ Q_2 \circ \ldots \circ Q_\ell$ satisfying the [low-hop] property of \Cref{lem:string}. To begin with, we focus on the case when $(a,b)$ is a canonical pair. In \Cref{SS:C-disjoint-paths}, we show how the Canonical Pairs Lemma lets us reduce the case that $(a,b)$ is arbitrary to the case that $(a,b)$ is a canonical pair.

Let \EMPH{$\canon a b$} be an arbitrary \EMPH{$(i_0,R_0)$}-pair.
We assume that $i_0$ is large enough that $2^{i_0} \ge \Omega(C \e^{-4} \log^5 n)$, for some sufficiently large hidden constant set later (determined in \Cref{lem:self-referential-threat}). 
Also recall that we defined $\e = \frac{\e_0}{\Theta(\log^3 n)}$
We show below that the procedure \EMPH{$\textsc{DetourPath}(\canon a b)$} returns a ``broken'' path $P(a,b)$ between $a$ and $b$
satisfying the properties of \Cref{lem:string}.
We emphasize that the procedure we are about to describe is \emph{not} in the algorithm we run to construct the collection of paths $\cO$; the construction of $P(a,b)$ happens only in the analysis.

\begin{definition}
\label{def:broken}
    A \EMPH{broken path} is a directed path $P$ together with a sequence of \EMPH{active subpaths} denoted \EMPH{$\act P$}. We treat each active subpath as a sequence of edges. The \EMPH{inactive subpaths} are the maximal subpaths of $P$ that are (edge-)disjoint from the active subpaths. An \EMPH{active vertex} is a vertex incident to an edge of an active subpath.
    For any subpath $P[x : y]$ of $P$, we define the active subpaths of $P[x : y]$ to be the (nonempty) active paths of $P$ restricted to $P[x:y]$, ie.,
    \[\EMPH{$\act {P[x : y]}$} \coloneqq \set{A \cap P[x : y] \mid A \in \act P \text{ and } A \cap P[x : y] \text{ contains at least 1 edge}}.\]
\end{definition}

We first describe the helper procedure \textsc{DetourAlong}, which modifies a broken path $P$ by ``detouring'' it along another path $\pi$. The procedure takes as input (1) a broken path $P$, (2) a path $\pi$ that we want to detour $P$ along, (3) a parameter $\Delta \ge 1$ called the detour threshold, and (4) a piece $R$ such that $P$ is a path in $R^+$ and $\pi$ is a shortest path in $R^+$. The procedure returns a new broken path (tagged with a new set of active subpaths). The pseudocode is in \Cref{fig:detour-along}, and an illustration is in \Cref{fig:detour-along-example}.

\begin{figure}[h!]
\small
\centering
\begin{algorithm}
\textul{$\textsc{DetourAlong}(P, \pi, \Delta, R)$}: \+
\\  \textbf{Input:} 
Broken path $P$ with endpoints $a$ and $b$ and active subpaths $\act P$;
path $\pi$ to detour $P$ along;\\
detour threshold $\Delta \ge 1$;
piece $R$ such that $P$ is a path in $R^+$ and $\pi$ is a shortest path in $R^+$
\\~
\\  $s \gets$ first vertex along $P$ such that $\dist_{R^+}(s, \pi) \le \Delta + 1$ \Comment{``first'' when orienting $P$ to start at $a$ and end at $b$}
\\  $s' \gets$ a vertex in $\pi$ that minimizes $\dist_{R^+}(s, s')$
\\  $t \gets$ last vertex along $P$ such that $\dist_{R^+}(t, \pi) \le \Delta + 1$
\\  $t' \gets$ a vertex in $\pi$ that minimizes $\dist_{R^+}(t, t')$
\\  if $s$ and $t$ exist: \Comment{If one exists so does the other; $s$ might be equal to $t$}\+
\\  $\textrm{spath}_{R^+}(s,s') \gets$ shortest path in $R^+$ between $s$ and $s'$
\\  $\textrm{spath}_{R^+}(t',t) \gets$ shortest path in $R^+$ between $t'$ and $t$
\\  $P \gets P[a:s] \circ \textrm{spath}_{R^+}(s,s') \circ \pi[s':t'] \circ \textrm{spath}_{R^+}(t',t) \circ P[t:b]$
\\  update $\act{P} \gets \act{P[a:s]} \cup \set{\pi[s':t']} \cup \act{P[t:b]}$\-
\\  return $P$ (with active subpaths $\act P$)
\end{algorithm}
\caption{The procedure \textsc{DetourAlong}.
}
\label{fig:detour-along}
\end{figure}

\begin{figure}[h!]
    \centering
\includegraphics[width=0.75\linewidth]{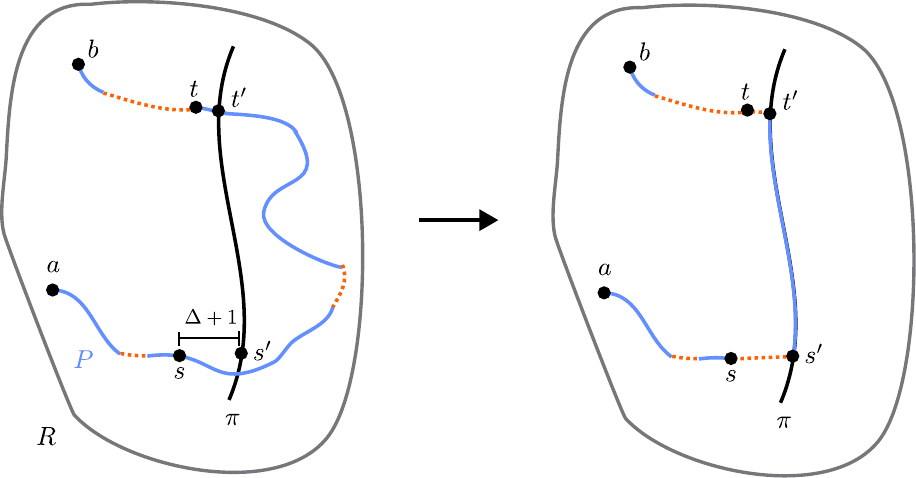}
    \caption{A stylized depiction of the result of $\textsc{DetourAlong}(P, \pi, \Delta, R)$. On the left: the piece $R$, the input broken path $P$ (with active subpaths drawn as solid blue lines, and inactive subpaths drawn as dashed orange lines) between vertices $a$ and $b$, and the path $\pi$ in $R$. The vertices $s, s',t,t'$ defined during the execution of $\textsc{DetourAlong}$ are marked. On the right: the broken path returned by \textsc{DetourAlong}.}
    \label{fig:detour-along-example}
\end{figure}
We now describe the procedure \textsc{DetourPath} in \Cref{fig:detour-try-again} with pseudocode. It takes as input the $(i_0, R_0)$-pair $\canon a b$ and returns a broken path $P(a,b)$ between $a$ and $b$.

\begin{figure}[h!]
\small
\centering
\begin{algorithm}
\textul{$\textsc{DetourPath}(\canon a b)$}: \+
\\  \textbf{Input:} $(i_0,R_0)$-pair $\canon a b$
\\~
\\ \Comment{Initialize broken path $P(a,b)$}
\\  $P(a,b) \gets$ canonical path in $\Pi$ associated with $\canon a b$, viewed as a directed path from $a$ to $b$
\\  initialize $\act {P(a,b)}$ to be $\set{P(a,b)}$
\\ \Comment{Detour $P(a,b)$}
\\  for each ancestor piece $R$ of $R_0$, from $R_0$ to root of hierarchy: \+
\\  $S \gets$ internal separator of $R$
\\  $\pi_S \gets$ internal separator path in $S$
\\  $P(a,b) \gets$ \textsc{DetourAlong}$(P(a,b), \pi_{S}, 2\e\cdot 2^{i_0}, R)$  \Comment{Detour $P(a,b)$ along $\pi_{S}$ within distance $(2\e\cdot 2^{i_0} )$} 
\\  for each path $\pi$ in $\Pi$ in the increasing $\preceq$-order, such that $\pi$ is an $(i,R)$-path with $2^i \ge \e 2^{i_0}$\+
\\      $P(a,b) \gets$ \textsc{DetourAlong}$(P(a,b), \pi, 7C, R)$  \Comment{Detour $P(a,b)$ along $\pi$ within distance $7C$} \-\-
\\ return $P(a,b)$
\end{algorithm}
\caption{The procedure \textsc{DetourPath}.
}
\label{fig:detour-try-again}
\end{figure}

We begin with an observation that the \textsc{DetourPath} algorithm is well-defined.
\begin{observation}
    Whenever we make some call $\textsc{DetourAlong}(P,\pi, \Delta, R)$ during the execution of $\textsc{DetourPath}$, we have that (i) $P$ is path in $R^+$ and (ii) $\pi$ is a shortest path in $R^+$. That is, we satisfy the preconditions on the input of \textsc{DetourPath}.
\end{observation}
\begin{proof}
    Part (i) of the observation follows from a straightforward inductive argument, using the fact that if path $P$ is contained in a piece $R$ then path $P$ is also contained in any ancestor piece $R'$ of $R$.
    To prove part (ii) of the observation, consider the calls to \textsc{DetourAlong} made while processing the piece $R$; part (ii) follows immediately from the fact that the separator path $\pi_S$ is a shortest path in $R$, and each $(\cdot, R)$-path $\pi$ is a shortest path in $R^+$, by definitions in \Cref{SS:canonical-pairs}. 
\end{proof}

The main guarantees on the output of \textsc{DetourPath} are described in the following lemma.
\begin{lemma}
\label{lem:canon-strings}
For any $(i_0,R_0)$-pair $\canon a b$ with $i_0$ large enough that $2^{i_0} \ge \Omega(C\e^{-4} \log^5 n)$ (see \Cref{rem:hidden-constant} for discussion on the hidden constant),
the broken path \EMPH{$P(a,b)$} returned by $\textsc{DetourPath}(\canon a b)$ satisfies:
\begin{enumerate}
\item 
Every active subpath of $P(a,b)$ is a subpath of some path in $\cO$. 

\item  
The total length of inactive subpaths of $P$ is at most $O(\e 2^{i_0} \cdot \log^2 n)$.

\item $P(a,b)$ consists of ${O}(\e^{-3} \log^5 n)$ active subpaths.%
\item We have $\len{P(a,b)} \le \dist_{R_0^+}(a,b) + O(\log^2 n) \cdot \e 2^{i_0}$
\end{enumerate}
\end{lemma}

Observe that the algorithm consists of two nested loops: the \EMPH{outer loop} iterates over the ancestor pieces $R$ of $R_0$, and the \EMPH{inner loop} iterates over canonical paths $\pi \in \Pi$ in piece $R$. To prove \Cref{lem:canon-strings}, we maintain the following invariants over each iteration of the outer loop; these invariants are stronger versions of the four properties of \Cref{lem:canon-strings}.
\begin{definition}
\label{def:outer-invariants}
Define $\EMPH{$\gamma$} \coloneqq \Theta(\e^{-3}\log^3 n)$, where the hidden constant will be specified later (see \Cref{rem:hidden-constant}).
Fix an $(i_0,R_0)$-pair $\canon a b$.
The \EMPH{outer-loop invariants} state that, after processing some ancestor piece $R$ of $R_0$ in the \EMPH{$\kappa$}-th iteration of the outer loop, the broken path \EMPH{$P(a,b)$} between $a$ and $b$ satisfies the following:
\begin{enumerate}
\item 
Denote
\[
\EMPH{$X_\kappa$} \coloneqq \bigcup \Set{\begin{array}{l}
\text{$7C$-neighborhood of $\pi'$ in $R'^+$ :}\\
\text{for all $(\cdot, R'$)-paths $\pi' \in \Pi$,
for all proper ancestors $R'$ of $R$} 
\end{array}
}.
\]
Then every active subpath of $P(a,b) \setminus X_\kappa$ is a subpath of some path in $\cO$. 

\item  The total length of inactive subpaths of $P$ is at most $\kappa \cdot (4\e \cdot 2^{i_0} + 16C \gamma)$.

\item $P(a,b)$ consists of at most $\kappa (2 \gamma + 2)$ active subpaths.%
\item We have $\lenR{R^+}{P(a,b)} \le \dist_{R_0^+}(a,b) + \kappa \cdot 16 \e 2^{i_0} + \kappa \gamma \cdot 56 C$.
\item During the $\kappa$-th iteration of the outer loop, there are at most $\gamma$ calls to $\textsc{DetourAlong}$ that change $P(a,b)$.
\end{enumerate}
\end{definition}

Assuming the correctness of these invariants, it is easy to prove \Cref{lem:canon-strings}.
\begin{proof}[of \Cref{lem:canon-strings}]
By \Cref{lem:separator-hierarchy}, the height of the separator hierarchy is $O(\log^2 n)$, so there are only $O(\log^2 n)$ iterations of the outer loop. In the last iteration, $\kappa = O(\log^2 n)$ and the piece  $R$ being processed is the root piece of the separator hierarchy (that is, $R = G$). To prove the 1st property of \Cref{lem:canon-strings}: By the 1st outer-loop invariant, $P(a,b) \setminus X_\kappa$ is in $\cO$ at the end of \textsc{DetourPath}$(\canon a b)$.
But the root piece has no proper ancestors, so $X_\kappa = \varnothing$; thus every subpath of $P(a,b)$ is a subpath of some path in $\cO$, as desired.
The 2nd and 3rd properties of \Cref{lem:canon-strings} follow from plugging in $\kappa = O(\log^2 n)$. The 4th property of \Cref{lem:canon-strings} follows from the fact $\kappa = O(\log^2 n)$ together with our assumption that $\e 2^{i_0} \ge \Omega(C\e^{-3} \log^5 n) \ge \gamma \cdot O(\log^2 n) \cdot 32 C$.
\end{proof}

It remains to establish that all loop invariants hold.

\begin{lemma}
\label{lem:outer-invariants}
The outer-loop invariants defined in \Cref{def:outer-invariants} hold throughout the execution of \textsc{DetourPath}. 
\end{lemma}

The proof is by induction on $\kappa$. Observe that the invariants trivially hold when $\kappa = 0$, i.e.\ immediately after $P(a,b)$ is initialized to be the canonical path associated with $\canon a b$. Hereafter we assume $\kappa \ge 1$. To prove the outer-loop invariants (1)--(4), our induction hypothesis is that all the outer-loop invariants hold at iteration $\kappa - 1$ \emph{and} that the outer-loop invariant (5) holds at iteration $\kappa$. To prove the outer-loop invariant (5), our induction hypothesis is that all the outer-loop invariants hold at iteration $\kappa -1$.

\subsubsection{Invariant 1: Detours as broken paths on $\cO$}
To show that every path in $P(a,b) \setminus X_\kappa$ is a subpath of some path in $\cO$, first we prove a few useful claims.

\begin{observation}
\label{obs:detour-separation}
    When we update $P^{\rm out} \gets \textsc{DetourAlong}(P, \pi, \Delta, R)$, the active subpaths of the new broken path $P^{\rm out}$ comprises a ``new'' subpath (namely, $\pi[s':t']$) as well as ``old'' subpaths (namely, the active subpaths of $P[a:s]$ and $P[t:b]$).
    Every vertex in an old subpath is at distance \ul{strictly greater than $\Delta$} (in $R^+$) from $\pi$.
\end{observation}
\begin{proof}
    The first sentence of the observation is immediate from the algorithm. We now elaborate on the last sentence. By choice of $s$, the only possible active vertex of $P[a:s]$ with distance $\le \Delta$ (in $\R^+$) of $\pi$ is $s$. By \Cref{def:broken}, active subpaths of $P[a:s]$ contain at least one edge. Thus, if $s$ is an active vertex of $P[a:s]$\footnote{Note that $s$ is an active vertex of the input path $P$, but a priori we don't know if $s$ is an active vertex of $P[a:s]$.}, there is an edge $(x,s)$ of $P[a:s]$ such that $x$ is an active vertex; in this case, $x$ is also an active vertex of $P$, and the assumption that $G$ is unweighted implies that $\dist_{R^+}(x,\pi) \le \Delta + 1$, contradicting the choice of $s$. A symmetric argument involving $t$ completes the proof.
\end{proof}

\begin{observation}
\label{obs:path-types}
After the $\kappa$-th iteration of the outer loop, 
the active subpaths in $P(a,b)$ can be partitioned into the following sets: 
\begin{itemize}
\item 
\EMPH{$\cP^{\rm new}$}: subpaths of $\pi_S$, and subpaths of some $(i,R)$-path $\pi$ in $\Pi$ with $2^i \ge \e\cdot 2^{i_0}$

\item
\EMPH{$\cP^{\rm old}$}: subpaths of $P^{\rm old}(a,b)$, where $P^{\rm old}(a,b)$ is the path $P(a,b)$ immediately after iteration $\kappa - 1$.
\end{itemize}
\end{observation}

\begin{lemma}
\label{lem:new-subpaths}
    The paths in $\cP^{\rm new} \setminus X_\kappa$ are subpaths of $\cO$.
\end{lemma}
\begin{proof}
    First observe that any path in $\cP^{\rm new}$ is a subpath of some $(i,R)$-path $\pi$ in $\Pi$ with $2^i \ge \e \cdot 2^{i_0}$.
    This follows from observing that $\pi_S$ is itself an $(i,R)$-path $\pi$ in $\Pi$ with $2^i \ge \e \cdot 2^{i_0}$: indeed, by definition of portals, both endpoints of $\pi_S$ are portals at the maximum scale $i = \lceil \log n \rceil$, and so $\pi_S$ is an $(i,R)$-path and $i \ge i_0$.

    Now, let $P^*$ be some path in $\cP^{\rm new} \setminus X_\kappa$. Path $P^*$ is a subpath of some $(i,R)$-path $\pi^*$ in $\Pi$ with $2^i \ge \e 2^{i_0}$.
    The description of the algorithm implies that $P^*$ is at distance greater than $7C$ (in $R^+$) from every $(\cdot, R)$-path $\pi$ in $\Pi$ with $\pi^* \prec \pi$.%
    \footnote{Indeed, consider the first time $P^*$ appears as a subpath in the detour path $P(a,b)$. Because the inner loop processes paths $\pi$ in increasing order according to $\prec$, the algorithm detours $P(a,b)$ along every $\pi$ with $\pi^* \prec \pi$ after $P^*$ first appears. After detouring against $\pi$, any vertex within distance $7C$ of $\pi$ gets deleted from $P^*$, by \Cref{obs:detour-separation}.}
    Moreover, the definition of $X_\kappa$ implies that $P^*$ is at least distance $7C$ away (in $R'^+$) from every $(\cdot, R')$-path with $R'$ a proper ancestor of $R$. We conclude that every vertex in $P^*$ is at least distance $7C$ away (in $R'^+$) from every $(\cdot, R')$-path $\pi$ with $\pi^* \prec \pi$ and $R'$ a (not necessarily proper) ancestor of $R$.

    We now argue this implies that no vertex in $P^*$ was erased when we added $\pi^*$ to $\cO$; that is, $P^*$ is a subpath of some path in $\cO$, as desired.
    Indeed, suppose that some vertex $v$ of $P^*$ was erased by some $(\cdot, R')$-path $\pi$ with $\pi^* \prec \pi$, meaning that $v \in \mathrm{EEr}_{\pi^*}$ as defined in the construction of $\cO$. 
    Then there exists some $u \in \pi^*$ with $\dist_{\pi^*}(v,u) \le C$ and $\dist_{R'^+}(u, \pi) \le 4C$. 

    From here there are two cases. First suppose that $R'$ is a (not necessarily proper) ancestor of $R$.
    Because $\pi^*$ is a path in $R^+$, and $R'^+$ contains $R^+$, we have $\dist_{R'^+}(v,u) \le \dist_{\pi^*}(v,u) \le C$.
    Triangle inequality implies that $\dist_{R'^+}(v, \pi') \le 5C < 7C$, contradicting our assumption.
    Now suppose that $R'$ is not an ancestor of $R$. Definition of $\prec$ implies that $R'$ and $R$ are not in an ancestor-descendant relationship. Let $\Gamma$ be the path in $R'^+$ between $u$ and $\pi'$ of length at most $4C$. By definition of separator hierarchy, $\Gamma$ intersects some separator $S_{\rm up}$ of an ancestor region $R_{\rm up}$ of $R$ (and WLOG, after the intersection point, $\Gamma$ remains in $R$). By definition of separator, this point of $\Gamma$ is within distance $2C$ of the separator path $\pi_{\rm up}$. Similarly to how we argued in the first paragraph of this proof, $\pi_{\rm up}$ is itself a path with $\pi^* \prec \pi_{\rm up}$. Thus we have $\dist_{R_{\rm up}^+}(u, \pi_{\rm up}) \le 4C + 2C = 6C$ and so $\dist_{R_{\rm up}^+}(v, \pi_{\rm up}) \le 7C$, a contradiction.
\end{proof}

\begin{lemma}
\label{lem:old-subpaths}
    The paths in $\cP^{\rm old} \setminus X_\kappa$ are subpaths of $\cO$.
\end{lemma}
\begin{proof}
By induction, we have that every path in
$\cP^{\rm old} \setminus X_{\kappa - 1}$ is a subpath of a path in $\cO$.
To prove the lemma, we show that the paths in $\cP^{\rm old} \setminus X_{\kappa}$ are precisely the same as the paths in $\cP^{\rm old} \setminus X_{\kappa - 1}$. To this end, let $P^*$ be a path in $\cP^{\rm old} \setminus X_{\kappa}$. We prove that $P^* \cap X_{\kappa - 1} = \varnothing$.

By definition of $X$, the set $X_{\kappa-1}$ contains all vertices in $X_\kappa$, plus the $7C$-neighborhoods (in $R^+$) of all paths $\pi \in \Pi$ with piece $R$.
Notice that $P^*$ is at distance greater than $7C$ (in $R^+$) from any $(i,R)$-path $\pi$ in $\Pi$ with $2^i \ge \e 2^{i_0}$, by the description of the algorithm. (Otherwise, part of $P^*$ would have been deleted when we detoured against $\pi_S$ or $\pi$, by \Cref{obs:detour-separation}.)
It remains to show that $P^*$ is at distance greater than $7C$ (in $R^+$) from any $(i,R)$-path $\pi$ with $2^i < \e 2^{i_0}$.
Assume for contradiction that there exists some $(i, R)$-path $\pi$ with $2^i < \e 2^{i_0}$, such that $\dist_{R^+}(P^*, \pi) \le 7C$.
As we assume $2^{i_0} \ge \Omega(C\e^{-3} \log^5 n) \ge 7C\e^{-1}$, clearly we also have $\dist_{R^+}(P^*, \pi) \le \e\cdot 2^{i_0}$.
Path $\pi$ has length less than $\e\cdot 2^{i_0}$, and one of its endpoints lies on the internal separator path $\pi_S$ of $R$, by definition of $(i,R)$-path.
Thus the distance from $P^{\rm old}(a,b)$ to $\pi_{S}$ is at most 
\[\dist_{R^+}(P^*, \pi_S) \le \dist_{R^+}(P^*, \pi) + \len \pi 
\le 2\e\cdot 2^{i_0}.\]
But this is a contradiction:
$P^*$ must be at distance strictly greater than $2\e\cdot 2^{i_0}$ from $\pi_S$ in $R^+$, because we detoured $P(a,b)$ against $\pi_S$ with $\Delta = 2\e \cdot 2^{i_0}$ (and so we can apply \Cref{obs:detour-separation}).
\end{proof}
All together, \Cref{obs:path-types} and \Cref{lem:new-subpaths,lem:old-subpaths} prove that the 1st outer-loop invariant holds.

\subsubsection{Invariants 2--4: Inactive length, number of subpaths, and length bound}
Here we prove the 2nd, 3rd, and 4th outer-loop invariant. Recall that we assume (as part of the induction hypothesis) that the outer-loop invariant (5) holds in the $\kappa$-th iteration: that is, $P(a,b)$ is modified at most $\gamma$ times by $\textsc{DetourAlong}$ during the $\kappa$-th outer iteration. The invariants 2--4 all follow from fairly straightforward applications of this fact.

\begin{claim}[Outer-loop invariant 2]
    The total length of inactive subpaths of $P$ is at most $\kappa \cdot (4\e \cdot 2^{i_0} + 16C \gamma)$.
\end{claim}
\begin{proof}
    Let $P^{\rm old}(a,b)$ denote the broken path $P(a,b)$ after iteration $\kappa -1$, and let $R_{\rm old}$ denote the piece processed in iteration $\kappa - 1$.
We assume inductively that the total length of inactive subpaths of $P^{\rm old}(a,b)$ is at most $(\kappa-1)(4\e \cdot 2^{i_0} + 16C\gamma)$.

Now observe that every call $P \gets \textsc{DetourAlong}(P^{\rm old}, \pi, \Delta, R)$ adds at most 2 inactive subpaths to the broken path $P$, each of length at most $\Delta + 1$; every other inactive edge in $P$ is also an inactive edge of $P^{\rm old}$.
In the $\kappa$-th iteration of the outer loop, $\textsc{DetourAlong}$ is called once with parameter $\Delta = 2\e \cdot 2^{i_0}$, and all other calls to $\textsc{DetourAlong}$ are made with parameter $\Delta = 7C$. By the induction hypothesis, there are at most $\gamma$ calls to $\textsc{DetourAlong}$ that affect $P(a,b)$ during the $\kappa$-th iteration. Thus the length of the inactive subpaths in $P(a,b)$ increases by at most $2\cdot(2\e \cdot 2^{i_0} + 1) +(\gamma-1)(2 \cdot (7C + 1)) < 4\e \cdot 2^{i_0} + 16 C \gamma$; that is, the total length of inactive subpaths of $P(a,b)$ is at most $\kappa(4\e \cdot 2^{i_0} + 16C \gamma)$.
\end{proof}

\begin{lemma}[Outer-loop invariant 3]
$P(a,b)$ has at most $\kappa(2\gamma + 2)$ active subpaths.
\end{lemma}
\begin{proof}
    We assume inductively that the outer-loop invariant 3 (\Cref{def:outer-invariants}) holds for iteration $\kappa-1$ of the outer loop. Thus, at the beginning of the $\kappa$-th outer iteration, we have that $P(a,b)$ has at most $(\kappa - 1) (2\gamma + 2)$ active subpaths.
    Each time $P(a,b)$ is modified by \textsc{DetourAlong}, the number of active subpaths increases by at most 2. (Specifically, $\pi[s':t']$ is added as an active subpath, and one old active subpath may be split into two new active subpaths when we chop into $P[a:s]$ and $P[t:b]$.) By induction hypothesis, $P(a,b)$ is modified at most $\gamma$ times in the $\kappa$-th iteration. Thus in the end, it has at most $\kappa (2\gamma + 2)$ active subpaths.
\end{proof}

\begin{lemma}[Outer-loop invariant 4]
\label{lem:outer-4}
At the end of the $\kappa$-th outer iteration, we have $\len{P(a,b)} \le \dist_{R_0^+}(a,b) + \kappa \cdot 16\e 2^{i_0} + \kappa \gamma \cdot 56 C$.
\end{lemma}
\begin{proof}
    We assume inductively that the 4th outer-loop invariant (\Cref{def:outer-invariants}) holds after the $(\kappa-1)$-th iteration.
    That is, at the start of iteration $\kappa$, we have
    \[\len{P(a,b)} \le \dist_{R_0^+}(a,b)  + (\kappa-1) \cdot 16 \e 2^{i_0} + (\kappa-1)\gamma \cdot 56C.\]
    We claim that each time $P(a,b)$ is modified by some call $P(a,b) \gets \textsc{DetourAlong}(P(a,b), \pi, \Delta, R)$, the length of $P(a,b)$ increases by at most $8\Delta$.
    (Actually, we prove that the length increases by at most $4 \Delta + 4$, but we upper-bound this quantity by $8\Delta$ for ease of notation.)
    This bound suffices to prove the lemma: path $P(a,b)$ is first detoured against the internal separator $\pi_S$ of $R$ (increasing its length by $+16 \e \cdot 2^{i_0}$) and then detoured at most $\gamma$ times by the inner loop (increasing its length by $+56 C \gamma$).
    
    To prove the bound on the length increase, consider some call $P \gets \textsc{DetourAlong}(P^{\rm old}, \pi, \Delta, R)$. Let $s$, $s'$, $t$, $t'$ be as defined in the \textsc{DetourAlong} procedure. As $\dist_{R^+}(s,s') \le \Delta + 1$ and $\dist_{R^+}(t,t') \le \Delta + 1$, triangle inequality implies that $\dist_{R^+}(s',t') \le \dist_{R^+}(s,t) + 2 \Delta + 2$. As $\pi$ is a shortest path in $R^+$, we conclude
    \begin{align*}
        \len{P} &\le \len{P^{\rm old}[a:s]} + \dist_{R^+}(s,s') + \dist_{R^+}(s',t') + \dist_{R^+}(t',t) + \len{P^{\rm old}[t:b]}\\
        &\le \len{P^{\rm old}[a:s]} + \dist_{R^+}(s,t) + \len{P^{\rm old}[t:b]} + 4\Delta + 4\\
        &\le \len{P^{\rm old}[a:s]} + \len{P^{\rm old}[s:t]} + \len{P^{\rm old}[t:b]} + 4\Delta + 4\\
        &= \len{P^{\rm old}} + 4\Delta + 4\\
        &\le \len{P^{\rm old}} + 8\Delta.
    \end{align*}
where the third-to-last line follows from the fact that $P^{\rm old}$ is a path in $R^+$.
\end{proof}

\subsubsection{Invariant 5: Number of modifications per outer iteration}

The goal of this section is to prove the 5th outer-loop invariant: during the $\kappa$-th iteration of the outer loop, there are at most $\gamma$ calls to $\textsc{DetourAlong}$ that change $P(a,b)$. 
Recall that $\gamma \coloneqq \Theta(\e^{-3} \log^3 n)$, as defined in \Cref{def:outer-invariants}.
For simplicity of notation, throughout the remainder of this section we write $P$ for $P(a,b)$.
Our inductive argument is interesting: we don't directly use the fact that the 5th outer-loop invariant holds in the $(\kappa-1)$-th iteration, but instead we use the fact that the length of $P$ is bounded during the $(\kappa-1)$-th iteration (that is, the 4th outer-loop invariant). Recall, in turn, that our proof of the length bound in iteration $\kappa-1$ relied on the fact that the 5th outer-loop invariant held during iteration $\kappa-1$.

Our proof of the 5th invariant uses a charging argument, for which we now introduce the notion of \EMPH{threateners}. Recall that $\canon{a}{b}$ is an $(i_0,R_0)$-pair.

\begin{definition}
\label{def:threatener}
Let $\canon{a'}{b'}$ be an $(i, R)$-pair, for some scale $i$. (Recall $R$ is an ancestor of $R_0$.)
We say that $\canon{a'}{b'}$ is a pair that \EMPH{$(\mThreat, \beta)$-threatens} $\canon a b$ if $\dist_{R^+}(a,a') \le \mThreat \cdot 2^{i} + \beta.$
\end{definition}
For the rest of the section, let $\EMPH{$\mThreat$} \coloneqq 2 + 2\e^{-1}$. 
For any $j > 0$, let $\EMPH{$\beta_j$}\coloneqq ((\kappa - 1)\gamma + j-1) \cdot 56C$.
We first need a slightly refined version of \Cref{lem:outer-4}.
\begin{claim}
\label{clm:inner-length}
    If $P$ has been modified by $j$ calls to $\textsc{DetourAlong}$ in the inner loop (during iteration $\kappa$ of the outer loop), we have $\len{P(a,b)} \le \dist_{R_0^+}(a,b) + \kappa \cdot 16\e 2^{i_0} + ((\kappa -1)\gamma + j) \cdot 56 C$.
\end{claim}
The proof follows verbatim from the first part of the proof of \Cref{lem:outer-4}. (In particular, the proof of \Cref{clm:inner-length} relies only on the fact that the 4th outer-loop invariant holds in iteration $\kappa-1$, which we assume holds by induction hypothesis). We now prove a lemma which will allow us to \EMPH{charge} each call of \textsc{DetourAlong} to a threatener.

\begin{lemma}
\label{lem:inner-threat}
    Suppose that $P$ has been modified by $j$ calls to \textsc{DetourAlong} in the inner loop (during iteration $\kappa$ of the outer loop, that processes piece $R$). 
    Let $X = \set{\canon {a_1} {b_1}, \ldots, \canon {a_j} {b_j}}$ denote the set of canonical pairs associated with the paths $\pi$ such that $P$ was modified while detouring along $\pi$.
    Every canonical pair in $X$ is a $(i,R)$-pair with $2^i \ge \e 2^{i_0}$ that $(\mThreat, \beta_j)$-threatens $\canon a b$.
\end{lemma}
\begin{proof}
    The proof is by induction on the number of calls made in the inner loop. 
    Suppose, by induction that $\set{\canon {a_1} {b_1}, \ldots, \canon{a_{j-1}} {b_{j-1}}}$ are all pairs that $(\mThreat, \beta_{j-1})$-threaten $\canon a b$. Clearly each one also $(\mThreat, \beta_{j})$-threatens $\canon a b$. 
    It remains to show that $\canon {a_j} {b_j}$ is a pair that $(\mThreat, \beta_{j})$-threatens $\canon a b$.  
    By description of the algorithm, $\canon {a_j} {b_j}$ is an \EMPH{$(i, R)$}-pair for some scale $i$ with $2^i \ge \e 2^{i_0}$. 
    Let \EMPH{$P^{\rm old}$} denote the path $P$ right before it gets detoured against the path $\EMPH{$\pi$} \coloneqq \pi(a_j, b_j)$. 
    Because $P$ is modified, there exists $s \in P_{\rm old}$ and $s' \in \pi$ such that 
    \[
    \dist_{R^+}(s, s') \le 7C+1.
    \]
    By definition, $\pi$ is a path in $R^+$ with length at most $2^{i}$ that contains both $s'$ and $a_j$, and so
    \[
    \dist_{R^+}(a_j, s') \le 2^{i}.
    \]
    Similarly, $P^{\rm old}$ is a (broken) path in $R^+$ that contains both $s$ and $a$; by \Cref{clm:inner-length}, $P^{\rm old}$ has length at most
    \begin{align*}
      \lenR{R^+}{P^{\rm old}} &\le \dist_{R_0^+}(a,b)  + \kappa \cdot 16\e 2^{i_0} + \beta_j\\
      &\le 2 \cdot 2^{i_0} + \beta_j
    \end{align*}
    where the last inequality follows from the fact that $\e = \frac{\e_0}{\Theta(\log^3 n)} < \frac{1}{16 \cdot O(\log^2 n)} \le \frac{1}{16\kappa}$ and also $\dist_{R_0^+}(a,b) \le 2^{i_0}$.
    Thus we have
    \(\dist_{R^+}(a,s) \le 2 \cdot 2^{i_0} + \beta_j.\)
    By triangle inequality, we conclude
        \[\dist_{R^+}(a_j, a) \le \dist_{R^+}(a_j, s') + \dist_{R^+}(s', s) + \dist_{R^+}(s, a) \le 2^i + (7C+1) + 2 \cdot 2^{i_0} + \beta_j.\]
    From the description of \textsc{DetourPath} we have that $2^{i} \ge \e\cdot 2^{i_0}$, and by assumption of \Cref{lem:canon-strings} we have $\e \cdot 2^{i_0} \ge 7C+1$. Thus we can write
    \[\dist_{R^+}(a_j, a) \le (2 + 2\e^{-1})\cdot 2^i + \beta_j = \mThreat \cdot 2^i + \beta_j\]
    as desired.
\end{proof}

It remains to show that there can't be too many canonical pairs that $(\mThreat, \beta_j$)-threaten $\canon a b$. 
But what value of $j$ should we choose when we try to bound the number of $(\mThreat, \beta_j)$-threatening pairs?
Below, in \Cref{lem:self-referential-threat} we show that there exists some small value $\theta$ such that there are at most $\theta$ pairs that are $(\mThreat, \beta_\theta)$-threateners; note the seemingly self-referential nature of the definition of $\theta$. To prove the lemma, we first need the following (standard) claim which bounds the number of portals that are close to any fixed vertex (see \cite[Claim 2.3]{chang2025distance}).
\begin{claim}
    \label{clm:length-threat}
    Let $R$ be a piece, let $i \in \mathbb{N}$ be a scale, and let $k \ge 1$. 
    For any vertex $v$ in $R^+$, the number of $(i, R)$-portals $p$ with $\dist_{R^+}(v, p) \le k \cdot 2^i$ is $O(k \cdot \e^{-1})$. 
\end{claim}

\begin{proof}
Let $S$ be the internal separator of $R$, and let $\pi_S$ be the separator path of $S$. Let $\xi = \set{p_0, \ldots, p_\beta}$ be the set of $(i, R)$-portals $p$ such that $\dist_{R^+}(v,p) \le k \cdot 2^i$, where $\beta \coloneqq |\xi|-1$. Notice that every $(i, R)$-portal lies on $\pi_S$.
Assume that $p_0$ is the first portal in $\xi$ when walking along $\pi_S$ (in an arbitrary direction), and that $p_\beta$ is the last portal in $\xi$ when walking along $\pi_S$.
It could be the case that $p_\beta$ is the endpoint of $\pi_S$, but observe that $p_{\beta-1}$ is not the final endpoint of $\pi_S$; thus,
property (2) from the definition of portals implies that $\dist_{\pi_S}(p_0, p_{\beta-1}) \ge \frac{\e}{2} \cdot 2^i \cdot (\beta - 1)$.
Because $\pi_S$ is a shortest path in $R^+$, we further have 
\[
\dist_{R^+}(p_0, p_{\beta-1}) \ge \frac{\e}{2} \cdot 2^i \cdot (\beta - 1).
\]
Triangle inequality implies that
\[\dist_{R^+}(p_0, p_{\beta-1}) \le \dist_{R^+}(p_0, v) + \dist_{R^+}(v, p_{\beta-1}) \le 2 k \cdot 2^i. \]
Combining these two inequalities, we find that $\beta = O(k \cdot \e^{-1}) + O(1) = O(k \cdot \e^{-1})$ as claimed.
\end{proof}

\begin{lemma}[``Self-referential'' Lemma]
\label{lem:self-referential-threat}
Suppose $\gamma = \Theta(\e^{-3} \log^{3} n)$ with a sufficiently large hidden constant.
Suppose $\canon a b$ is an $(i_0,R_0)$-pair with $2^{i_0} = \Omega(\gamma \cdot C \e^{-1} \log^2 n)$ with a sufficiently large constant. Let $R$ be an ancestor piece of $R_0$.
There exists an integer $\theta \le \gamma$ such that there are \emph{strictly fewer} than $\theta$ many $(i, R)$-pairs that $(\mThreat, \beta_\theta)$-threaten $\canon a b$ and satisfy $2^i \ge \e 2^{i_0}$.
\end{lemma}
\begin{proof}
    Define 
    \[
    \EMPH{$\theta$} \coloneqq \text{one more than the number of $(\cdot, R)$-pairs that $(\mThreat + 1, 0)$-threaten $\canon a b$.} 
    \]
    We will show below that $\theta \le \gamma = O(\e^{-3} \log^3 n)$. We now claim that (if $i_0$ is sufficiently large) every $(i,R)$-canonical pair that $(\mThreat, \beta_\theta)$-threatens $\canon a b$ and satisfies $2^i \ge \e 2^{i_0}$ \emph{also} $(\mThreat + 1, 0)$-threatens $\canon a b$. Indeed, suppose that $\canon {a'} {b'}$ is an $(i, R)$-canonical pair that $(\mThreat, \beta_\theta)$-threatens $\canon a b$ and satisfies $2^i \ge \e 2^{i_0}$. Then
    \begin{align*}
        \dist_{R^+}(a, a') &\le \mThreat \cdot 2^{i} + \beta_\theta\\
        &= \mThreat \cdot 2^{i} + ((\kappa - 1)\gamma + \theta-1) \cdot 56C\\
        &= \mThreat \cdot 2^{i} + O(\gamma \cdot C \log^2 n)\\
        &\le \mThreat \cdot 2^{i} + \e \cdot 2^{i_0}\\
        &\le (\mThreat +1) \cdot 2^{i}
    \end{align*}
    where the third-to-last line follows from the fact that $\kappa = O(\log^2 n)$ and $\theta \le \gamma$;
    the second-to-last inequality follows from the assumption that $2^{i_0} = \Omega(\gamma \cdot C \e^{-1} \cdot \log^2 n)$ is sufficiently large; and the last inequality follows from assumption that $\e 2^{i_0}\le 2^{i}$. We conclude that $\canon{a'}{b'}$ is a pair that $(t+1, 0)$-threatens $\canon a b$.

    To prove the lemma, it remains to show that $\theta \le \gamma$. There are $O(\log n)$ scales.
    For each scale $i$, \Cref{clm:length-threat} implies that the number of $(i, R)$-portals $a'$ with $\dist_{R^+}(a', a) \le (\mThreat + 1) \cdot 2^{i}$ is at most $O(\mThreat \e^{-1}) = O(\e^{-2})$.
    For each such portal $a'$, the number of $(i,R)$-pairs of the form $\canon {a'} {b'}$ is at most $O(\e^{-1} \log^2 n)$: indeed, there are $O(\log^2 n)$ ancestor pieces $R'$ of $R$, and for each $R'$ there are only $O(\e^{-1})$ scale-$i$ portals $b'$ on the internal separator of $R'$ such that $\dist_{R^+}(a',b') \le 2^{i}$ (by \Cref{clm:length-threat}). All together, there are only $O(\e^{-3}\log^3 n)$ canonical pairs that $(\mThreat+1, 0)$-threaten $\canon a b$. Provided that we choose the hidden constant in the definition $\gamma = \Theta(\e^{-3} \log^3 n)$ is large enough, we have $\theta \le \gamma$.
\end{proof}

\begin{corollary}[Outer-loop invariant 5]
\label{cor:inner-termination}
    During iteration $\kappa$ of the outer loop, the path $P$ is modified by at most $\gamma$ calls to $\textsc{DetourAlong}$.
\end{corollary}
\begin{proof}
    The path $P$ may be modified by the first call to $\textsc{DetourAlong}$ that detours $P$ along the internal separator  $\pi_R$. We claim that $P$ is modified by at most $\gamma-1$ calls to $\textsc{DetourAlong}$ in the inner loop. Indeed, otherwise there would be some instant at which $P$ had been modified by exactly $\theta$ iterations of the inner loop. At this point, \Cref{lem:inner-threat} would imply that there are $\theta$ canonical pairs that $(\mThreat, \beta_\theta)$-threaten $\canon{a}{b}$, contradicting \Cref{lem:self-referential-threat}.
\end{proof}

\begin{remark}[Choice of hidden constants for $\gamma$ and $2^{i_0}$]
\label{rem:hidden-constant}
In the statement of \Cref{def:outer-invariants}, we choose the hidden constant in $\gamma = \Theta(\e^{-3} \log^3 n)$ to be large enough to satisfy the constraints of \Cref{lem:self-referential-threat}.  In the statement of \Cref{lem:canon-strings}, we choose the hidden constant $2^{i_0} \ge \Omega(\gamma\e^{-1}C\log^2 n) = \Omega(C\e^{-4}\log^5 n)$ sufficiently large to satisfy the constraints of \Cref{lem:canon-strings,lem:self-referential-threat,lem:old-subpaths}.
We emphasize that these hidden constants are well-defined and not circular.
\begin{itemize}
    \item The choice of $\gamma$ depends only on the definition $\mThreat = 2 +  2\e^{-1}$, the number of scales $O(\log n)$, the height $O(\log^2 n)$ of the separator hierarchy defined in \Cref{lem:separator-hierarchy}, and the constant hidden in the number $O(k \cdot \e^{-1})$ of nearby portals in \Cref{clm:length-threat}.
    \item The hidden constant in the lower bound to $2^{i_0}$ depends on (i) the constraint $2^{i_0} \ge \e^{-1} \cdot O(\log^2 n) \cdot \gamma \cdot 56 C$ where $O(\log^2 n)$ is the height of the separator hierarchy defined in \Cref{lem:separator-hierarchy} (for the proofs of \Cref{lem:canon-strings,lem:self-referential-threat}), (ii) the constraint $2^{i_0} \ge 7C\e^{-1}$ (for the proof of \Cref{lem:old-subpaths}), and (iii) the constraint $2^{i_0} \ge (7C+1)\e^{-1}$ (for the proof of \Cref{lem:inner-threat}).
\end{itemize}  
All these constraints are satisfied for sufficiently large $\gamma$ and $2^{i_0}$.
\end{remark}

\subsection{Proof of \Cref{lem:string}}
\label{SS:C-disjoint-paths}

We finally prove \Cref{lem:string}, which we restate below.

\brokenpaths*

\begin{proof}[of \Cref{lem:string}]
We construct $\cO$ as described in \Cref{S:broken-paths}: define $\EMPH{$\e$} \coloneqq O(\log^{-3} n) \cdot \e_0$, construct a $2C$-neighborhood path separator hierarchy of height $O(\log^2 n)$, define canonical paths $\Pi$, and construct $\cO$ according to the $\textsc{ErasedPaths}$ procedure.
By \Cref{lem:C-disjoint-in-G}, the paths $\cO$ are pairwise at distance more than $2C$ from each other in $G$. We now prove the [low-hop] property.

\paragraph{Definition of $P$.} Let $u$ and $v$ be vertices in $G$.
If $\dist_G(u,v) \le \Theta(C\e^{-4}\log^5 n) = \Theta(C \e_0^{-4} \log^{17}n)$ for some sufficiently large hidden constant, then we take $P$ to be the shortest path between $u$ and $v$ in $G$ and we are done.
Otherwise, assuming $\dist_G(u,v) \ge \Omega(C \e^{-4} \log^5 n)$, the Canonical Pairs Lemma (\Cref{lem:canonical-pairs}) implies there is a path 
\EMPH{$P_0$} in $G$ between $u$ and $v$ that is the concatenation of $O(\log n)$ subpaths, two of which are called \EMPH{short subpaths} and have length $O(\e\cdot \dist_G(u,v) + C) = O(\e \cdot \dist_G(u,v))$, and the rest of which are between canonical pairs at scale $i$ where $2^i = \Theta(\dist_G(u,v)) \ge \Omega(C \e^{-4} \log^5 n)$; moreover, we have
\[
\len{P_0} \le (1+O(\log n)\cdot \e) \cdot \dist_G(u,v) + O(C \cdot \log n) \le (1+O(\log n)\cdot \e) \cdot \dist_G(u,v).
\]
For each canonical pair $\canon a b$ on $P_0$ we will find a broken path $P(a,b)$ according to the \textsc{DetourPath} procedure (\Cref{lem:canon-strings}); note that the Canonical Pairs Lemma (\Cref{lem:canonical-pairs}) implies that $\canon a b$ 
is at scale $i$ where $2^i = \Theta(\dist_G(u,v)) \ge \Omega(C \e^{-4} \log^5 n)$, so the assumption of \Cref{lem:canon-strings} holds.
We construct the final path \EMPH{$P$} between $u$ and $v$ by concatenating all the broken paths $P(a,b)$ together, and concatenating the two short subpaths of $P_0$. We write
\[
P = Q_1 \circ P_1 \circ Q_2 \circ P_2 \circ \dots \circ Q_\ell.
\] 
where the $P_i$'s are active subpaths, and the $Q_i$'s are maximal subpaths of inactive edges that separate the $P_i$'s. (In this definition, we treat the two short subpaths of $P_0$ as inactive subpaths.)

\paragraph{Analysis of $P$.} We now show that path $P$ satisfies all the properties claimed in the [low-hop] part of \Cref{lem:string}. Fix any canonical pair $\canon a b$. 
\Cref{lem:canon-strings}(1) implies all active subpaths in the detour path $P(a,b)$ are subpaths of paths in $\cO$; thus, all the subpaths $P_i$ of $P$ are subpaths of $\cO$, as desired.
\Cref{lem:canon-strings}(3) implies there are $O(\e^{-3}\log^5 n)$ active subpaths in $P(a,b)$, and so in total (after concatenating the $O(\log n)$ detour paths) the path $P$ contains $\ell = O(\e^{-3}\log^6 n) = O(\e_0^{-3} \cdot \log^{15} n)$ subpaths $P_i$, as desired.
To bound the length of the $Q_j$ subpaths, note that \Cref{lem:canon-strings}(2) implies that the total length of the inactive edges in $P(a,b)$ is at most 
\(
O(\log^2 n) \cdot O(\e \cdot 2^{i})
= O(\e \log^2 n) \cdot \dist_G(u,v),
\)
and the two short subpaths of $P_0$ contribute $O(\e \cdot \dist_G(u,v) + C) \le O(\e \dist_G(u,v))$ to the length of the inactive subpaths $Q_j$.
Thus the sum of lengths of all inactive subpaths $Q_j$ in $P$ (summed over all $O(\log n)$ canonical pairs) is at most 
\[
O(\e \log^3 n \cdot \dist_G(u,v)) \le \e_0 \cdot \dist_G(u,v).
\]
Finally, \Cref{lem:canon-strings}(4)
implies that the length of path $P(a,b)$ between $(i,R_j)$-pair in $R_j^+$ is at most
\[\len {P(a,b)} \le \dist_{R_j^+}(a,b) + O(\log^2 n) \cdot \e 2^{i} \le \dist_{R_j^+}(a,b) + O(\log^2 n) \cdot \e \dist_G(u,v).\]
Thus the total length of $P$, summing across all $O(\log n)$ canonical pairs, is at most
\begin{align*}
\len P &\le O(\e \cdot \dist_{G}(u,v)) + \sum_{\text{$(i,R_j)$-pair $\canon a b$ on $P_0$}} \left( \dist_{R_j^+}(a,b) + O(\log^2 n) \cdot \e \dist_G(u,v) \right)\\
&\le O(\e \cdot \dist_{G}(u,v)) + \len{P_0} + O(\e \log^3 n) \cdot \dist_G(u,v) \\
&\le (1 + O(\e \log n)) \dist_G(u,v) + O(\e \log^3 n) \cdot \dist_G(u,v) \\
&\le (1+\e_0)\cdot \dist_{G}(u,v).
\end{align*}
\aftermath

\end{proof}

\subsection{Construction of separator hierarchy}
\label{SSS:separator-hierarchy}

This subsection is dedicated to proving the following lemma.
\hierarchy*

Our key tool, which we apply recursively in a standard way, is to construct a shortest path separator that consists of $O(\log n)$ paths.
\stringSeparator*

Before proving \Cref{lem:string-path-separator}, we show that it implies \Cref{lem:separator-hierarchy}.
\begin{proof}[of \Cref{lem:separator-hierarchy}]
We aim to construct a $2C$-neighborhood separator hierarchy for $G$ with a recursive approach: roughly, we want to apply \Cref{lem:string-path-separator} to find a balanced separator $S$ for $G$ and use this separator as the root node of the hierarchy, and then recurse on each connected component of $G \setminus S$. There are two issues. First, the separator $S$ comprises $O(\log n)$ shortest paths and their neighborhoods --- but our definition of separator hierarchy only permits the internal separator of a node to be a single path and its neighborhood. To deal with this, we ``peel off'' the paths in $S$ one-by-one with a path of $O(\log n)$ nodes in the hierarchy. Second, if we recurse on a connected piece $R$ of $G \setminus S$, then the next set of separator paths we select will be shortest paths in $R$ --- but our definition of separator hierarchy requires that the internal separator of $R$ is a shortest path in $R^+$, not just in $R$. To deal with this issue, we recurse on $R \cup S$ instead of just recursing on $R$.

We now describe the construction formally. Our procedure takes as input two induced subgraphs $R$ and $R^+$, such that any path in $G$ containing a vertex in $R$ and $G \setminus R$ intersects $R^+$ (that is, $R^+ \setminus R$ represents the external separators of $R$), and it returns a tree representing a separator hierarchy for piece $R$. Calling this procedure with $R = R^+ = G$ produces a $2C$-neighborhood path separator hierarchy for $G$. Given inputs $R$ and $R^+$, initially we set the weight of every vertex in $R$ to be 1, and apply \Cref{lem:string-path-separator} on $R^+$ to find a set $\EMPH{$\cS$} = \set{S_1, \ldots, S_{k}}$ of $k = O(\log n)$ shortest paths in $R^+$ whose $2C$-neighborhood acts as a separator for $R^+$. For each $S_i \in \cS$, define $S_i^+ \subseteq V(G)$ to be the $2C$-neighborhood of $S_i$ in $R^+$.
We initialize a separator hierarchy $\cT$ as a path of $k = O(\log n)$ nodes: the first node of the path (ie, the root) is associated with the piece $R$ and internal separator $S_1^+$, and the $i$th node on the path is associated with the piece $R_i \coloneqq R \setminus \bigcup_{j < i} S_j^+$ and internal separator $S_i^+$. Observe that the union of $R_i$ and its external separators is $R_i^+ \coloneqq R_i \cup \left( \bigcup_{j < i} S_j^+\right) \cup (R^+ \setminus R) = R^+$. Clearly, for every piece $R_i$ the corresponding internal separator is a $2C$-neighborhood of one shortest path in $R_i^+ = R^+$.
Moreover, any path in $G$ that contains a vertex in $R_i$ and a vertex not in $R_i$ intersects a vertex in an external separator, ie $R^+_i \setminus R_i$.
To finish building $\cT$, we consider every connected components $C$ of $R_k \setminus S_k^+$; define $C^+ \coloneqq C \cup (R^+ \setminus R) \cup \left( \bigcup_{j \le k} S_j^+ \right)$ to be the union of $C$ with all its external separators; recursively compute a tree $\cT_C$ with inputs $C$ and $C^+$; and attach $\cT_C$ as a child subtree of the node $R_k$ in $\cT$.

A inductive argument implies that calling this procedure with $R = R^+ = G$ produces a valid $2C$-neighborhood path separator hierarchy for $G$ as in \Cref{def:separator}.
    Moreover, the height of the final separator hierarchy is $O(\log^2 n)$: the depth of the recursion is $O(\log n)$ (because the weight of the piece $R$ decreases by half in each iteration) and in each iteration we attach a path of $O(\log n)$ nodes to $\cT$.

\end{proof}

It remains to prove \Cref{lem:string-path-separator}. To do this, we combine two tools: the \emph{path aggregation} theorem
\cite{hathcock2025steiner}, and the standard shortest-path separator theorem in planar graphs \cite{lipton1979separator,thorup2004compact}.
\begin{lemma}[Path Aggregation {\cite[Theorem 1.1]{hathcock2025steiner}}]
\label{lem:path-aggregation}
    Let $H$ be an $n$-vertex graph, let $r \in V(H)$ be a designated root vertex, and let $X$ be an arbitrary subset of vertices. 
    For every vertex $x \in V(H)$, let $P_x$ be a simple $x$-to-$r$ path in $H$, and let $\cP = \set{P_x : x \in X}$.
    Then there is a tree $T$ of $H$, such for every $x \in V(H)$: the $x$-to-$r$ path in $T$ can be written as the concatenation of $O(\log n)$ subpaths, each of which is a subpath of a path in $\cP$.%
    \footnote{In fact, \cite{hathcock2025steiner} prove something a bit stronger---they allow $G$ to be a directed multigraph, and they bound the \emph{switching cost} of paths in $T$.}

\end{lemma}
We will apply this lemma when $X = V(H)$, in which case the tree $T$ produced is a spanning tree of $H$.

\begin{lemma}[Planar Shortest-Path Separators {\cite[Lemma 2.3]{thorup2004compact}}]
\label{lem:planar-separator}
    Let $H$ be a planar graph, and let $T$ be a spanning tree of $H$. For any set of non-negative vertex weight of $G$, there is a set of 3 vertex-to-root paths in $T$ whose removal shatters $H$ into connected components with at most half the weight of~$H$.
\end{lemma}

We can now prove \Cref{lem:string-path-separator}.
\begin{proof}[of \Cref{lem:string-path-separator}]
    Let $H$ be a $C$-distortion planar emulator for $G$, with $V(G) = V(H)$. Fix an arbitrary set of non-negative weights on the vertices.
    We begin by defining a set of paths $\cP$ on $G$ and a corresponding set of paths $\cP^H$ on $H$. Compute a single-source shortest path tree $T$ in $G$, rooted at some arbitrary vertex $v \in V(G)$. Define $\cP$ to be the set of $r$-to-$v$ paths in $T$, over all $v \in V(G)$. 
    For every path $P = (r=p_1, \ldots, p_\ell)$ in $\cP$, we define the \EMPH{image} of $P$, \EMPH{$P^H$}, to be the simple path obtained from shortcutting corresponding walk in $H$: that is, $P^H$ is obtained by (1) taking the the union of, for each $i \in [\ell-1]$, the shortest path in $H$ between $p_i$ and $p_{i+1}$, and then (2) removing all cycles. We say that $P$ is the \EMPH{preimage} of $P^H$ in $G$. Let $\EMPH{$\cP^H$} = \set{P^H : P \in \cP}$.

    Having constructed the set of paths $\cP^H$, we now apply path aggregation (\Cref{lem:path-aggregation}) on $\cP^H$: we conclude that there is a spanning tree \EMPH{$T^H$} of $H$, such that every path in $T^H$ lies in the union of $O(\log n)$ paths in $\cP^H$.
    By the planar shortest-path separator theorem (\Cref{lem:planar-separator}), there are 3 paths $\EMPH{$\Pi$} = \set{\pi_1,\pi_2,\pi_3}$ in $T^H$ such that every connected component of $H \setminus \Pi$ has at half the total weight on the vertices of $H$. For each path $\pi_i \in \Pi$, let $\EMPH{$\cP^H[\pi_i]$} \subseteq \cP^H$ denote a set of $O(\log n)$ paths in $\cP^H$ whose union contains $\pi_i$; note that $\cP^H[\pi_i]$ exists by construction of the tree $T^H$. Similarly, let $\EMPH{$\cP^H[\Pi]$} = \cP^H[\pi_1] \cup \cP^H[\pi_2] \cup \cP^H[\pi_3]$. Finally, let $S \subseteq V(G)$ be the set of vertices which are in the $2C$-neighborhood (in $G$) of the preimage of paths in $\cP^H[\Pi]$ --- that is,
    \[\EMPH{$S$} \coloneqq \bigcup_{P^H \in \cP^H[\Pi]} \mathcal{N}_G(P, 2C).\]

    Clearly $S$ is the $2C$-neighborhood of $O(\log n)$ shortest paths in $G$. It remains to show that deleting $S$ from $G$ shatters $G$ into connected components with at most half the weight of $G$. We first pause to prove a helpful claim:
    \begin{equation}
    \label{eq:c-neighborhood-sep}
        \text{If vertex $v$ is in the $C$-neighborhood (in $G$) of a vertex of some $\pi_i \in \Pi$, then $v \in S$.}
    \end{equation}
    Indeed, let $v'$ be a vertex on $\pi_i \in \Pi$; we show that the $C$-neighborhood of $v'$ is in $S$. By definition, the set of paths $\cP^H[\pi_i]$ contains some path $P^H$ that contains $v'$. Let $P$ be the preimage of $P^H$. The preimage $P$ need not contain $v'$, but we have that $\dist_G(v', P) \le \dist_H(v', P) \le C$ --- indeed, the path $P^H$ is obtained from a walk composed of subpaths of length most $C$ whose endpoints are in $P$ (and even after we delete vertices in $P^H$ to make it a simple path, it remains true every vertex in $P^H$ is within distance $C$ of a vertex of $P$ in $H$). By construction, the set $S$ contains the $2C$-neighborhood (in $G$) of every vertex in $P \in \cP^H[\pi_i]$, and so by triangle inequality $S$ contains the $C$-neighborhood (in $G$) of $v'$. This proves \eqref{eq:c-neighborhood-sep}.
    
    We can now show that every connected component of $G \setminus S$ has at most half the weight of $G$. Let $X \subset V(G)$ be a connected component in $G \setminus S$. Let $V(\Pi) \subseteq V(H)$ denote the set of vertices on some path in $\Pi$. We claim $X$ is a subset of some connected component in $H \setminus V(\Pi)$; by construction of $\Pi$, this implies that $X$ has at most half the total weight of $G$ (which is the same as the total weight of $H$, as $V(H) = V(G)$). Indeed, suppose for contradiction that $X$ contained two vertices $x$ and $y$ in different connected components of $H \setminus V(\Pi)$. Let $\pi$ be a path between $x$ and $y$ in $G[X]$. Let $\pi^H$ be the image of $\pi$ in $H$.
    By assumption, $x$ and $y$ are in different connected components, so $\pi^H$ contains some vertex $v'$ in $V(\Pi)$. Every vertex in $\pi^H$ is within distance $C$ (in $H$) of some vertex in $\pi$ (by definition of image, as we argued above), so we conclude that there exists a vertex $v$ on $\pi$ with $\dist_G(v,v') \le \dist_H(v,v') \le C$. But \eqref{eq:c-neighborhood-sep} implies that $v \in S$, and so $\pi \cap S \neq \varnothing$. This contradicts our assumption that $\pi$ is a path in $G \setminus S$.
\end{proof}

\subsection{Proof of the Canonical Pairs Lemma}
\label{S:canonical-pairs}

\begin{observation}
\label{obs:nearby-portal}
    Let $R$ be a piece with internal separator $S$, and let $i$ be a scale. Let $\Pi_i$ denote the scale-$i$ portals of $R$. Every vertex $v \in S$ satisfies $\dist_{S}(v, \Pi_i) \le \e 2^i + O(C)$.
\end{observation}
\begin{proof}
    By \Cref{def:separator}, the separator $S$ comprises the $2C$-neighborhood of $\pi_S$ in $R^+$. In particular, the vertex $v \in S$ satisfies $\dist_{R^+}(v,\pi_S) \le O(C)$; moreover, every vertex on the shortest path in $R^+$ between $v$ and $\pi_S$ is in $S$ by definition, so $\dist_S(v,\pi_S) \le O(C)$.
    A geometric series implies that every vertex in $\pi_S$ is within distance $\e 2^i$ of $\Pi_i$.
\end{proof}

\begin{proof}[of \Cref{lem:canonical-pairs}]
    Initialize \EMPH{$\tilde P_{u,v}$} to be a shortest path in $G$ between $u$ and $v$. 
    Let \EMPH{$R$} be the
    highest piece in the separator hierarchy such that $\tilde P_{u,v}$ intersects the internal separator \EMPH{$S$} of $R$, and let \EMPH{$p$} be a point in $S \cap \tilde P_{u,v}$.
    By definition of separator hierarchy, $R$ contains every vertex in $\tilde P_{u,v}$ (otherwise $\tilde P_{u,v}$ would contain some vertex in an external separator of $R$, contradicting the choice of $R$).
    In particular, this means $\dist_{R}(u,p) + \dist_{R}(p,v)  = \dist_{G}(u,v)$.
    Let \EMPH{$s$} be the scale such that $\dist_G(u,v) \in [2^{s-2},2^{s-1})$.
    Observe that $2^s = \Theta(\dist_G(u,v))$.
    We assume without loss of generality that $C \le 2^s / 100$: otherwise, we simply take $P \coloneqq \tilde P_{u,v}$, and $P$ consists of a single subpath of length $O(C)$.
    
    We now prove that there is a path $P_{u}$ in $G$ between $u$ and $p$ with length $\dist_{R}(u,p) + O(\e 2^s \cdot \log n)$, such that $P_{u}$ can be written as the concatenation of $O(\log n)$ subpaths: 2 of which have length at most $O(\e 2^s)$, and the rest of which is a shortest path (in $R_j^+$) between some $(\cdot, R_j)$-pair. 
    By symmetry, such a path $P_{v}$ also exists between $p$ and $v$. The lemma follows by concatenating $P \coloneqq P_{u} \circ P_{v}$;
    indeed, path $P$ has length at most $(1+O(\e \log n))\cdot\dist_G(u,v)$ because of the definition of $s$, and $P$ can be decomposed into $O(\log n)$ subpaths as desired.

    \medskip \noindent \textbf{Construction of \boldmath{$P_{u}$}.}
    Let \EMPH{$\tilde P_u$} denote the subpath of $\tilde P_{u,v}$ that runs between $u$ and $p$; throughout this proof we view $P_u$ to be constructed as starting at $u$ and ending at $p$. 
    Let \EMPH{$R_1$} be the lowest piece such that $u$ is in $R_1$; in particular, this means that $u$ is in the internal separator \EMPH{$S_1$} of $R_1$. Define $\EMPH{$x'_1$}\coloneqq u$. 
    By \Cref{obs:nearby-portal}, there exists an $(s,R_1)$-portal \EMPH{$x_1$} on $S_1$ with $\dist_{S_1}(u,x_1) \le \e 2^s + O(C) \le O(\e 2^s)$. 
    Now, for $j > 1$, we inductively define: let \EMPH{$x'_j$} be the first vertex on $\tilde P_u$ that is not in piece $R_{j-1}$; let \EMPH{$R_j$} be the lowest piece that contains $x'_j$ and let \EMPH{$S_j$} be the internal separator of $R_j$; and let \EMPH{$x_j$} be the $(s,R_j)$-portal on $S_j$ provided by \Cref{obs:nearby-portal} with $\dist_{S_j}(x_j, x_j') \le O(\e 2^s)$.
    We remark that the definition of separator hierarchy implies that piece $R_j$ is a proper ancestor of $R_{j-1}$; thus,
    after $\EMPH{$k$} = O(\log n)$ iterations, we have $R_k \gets R$ and the process terminates (as $\tilde P_u$ is contained in $R$). 
    We slightly abuse notation and define $\EMPH{$x_{k+1}'$} \coloneqq p$, $\EMPH{$R_{k+1}$} \coloneqq R$ and $\EMPH{$S_{k+1}$} \coloneqq S$, and $\EMPH{$x_{k+1}$}$ to be an $(s,R)$-portal on $S$ with $\dist_S(x_{k+1}, p) \le \e 2^s + O(C)$.
    
    Let \EMPH{$Q$} be the shortest path in $R_1^+$ between $u$ and $x_1$.
    For $j \in [k]$, define \EMPH{$P_j$} to be the shortest path in $R_j^+$ between $x_j$ and $x_{j+1}$.
    Let \EMPH{$Q'$} be the shortest path in $R^+$ between $x_{k+1}$ and $p$.
    We define \EMPH{$P_{u}$} to be the concatenation \(P_{u} = Q \circ P_1 \circ P_2 \ldots \circ P_k \circ Q'.\)

    \medskip \noindent \textbf{Stretch.}
    For every $j \in [k]$, let \EMPH{$\tilde P[x_j':x_{j+1}']$} denote the subpath of $\tilde P_u$ that runs from $x_{j}'$ to $x_{j+1}'$. The subpaths $\tilde P[x_{j}', x_{j+1}']$ are disjoint, and their concatenation is precisely the path $\tilde P_u$.
    We conclude that
    \[\sum_{j \in [k]} \len {\tilde P[x_j' : x_{j+1}']} \le \dist_{R}(u,p).\]
    Each subpath $\tilde P[x_{j}', x_{j+1}']$ is a path in $R_j^+$ (specifically, it lies entirely in $R_j$ except for the endpoint $x_{j+1}'$ which is in some external separator of $R_j$). Moreover, because $R_{j+1}$ is an ancestor%
    \footnote{note that if $j = k$, then $R_j = R_{j+1}$: even in this case, $R_j$ is a (non-proper) ancestor of $R_{j-1}$} 
    of $R_j$, the separators $S_j$ and $S_{j+1}$ are both subsets of $R_j^+$, and we have $\dist_{R_j^+}(x_{j}', x_{j}) \le O(\e 2^s)$ and $\dist_{R_j^+}(x_{j+1}', x_{j+1}) \le \dist_{S_{j+1}}(x_{j+1}', x_{j+1}) \le O(\e 2^s)$.
    Triangle inequality implies
    \[
    \len {P_j} = \dist_{R_j^+}(x_j,x_{j+1})
    \le O(\e 2^s) + \dist_{R_j^+}(x_{j}',x_{j+1}')  
    =  O(\e 2^s)+ \len{\tilde P[x_{j}', x_{j+1}']}.
    \]
    We conclude that 
    \begin{align*}
        \len {P_{u,p}} &= \len Q + \sum_{j \in [k]} \len {P_j} + \len {Q'}\\
        &\le  O(\e 2^s \cdot \log n) + \sum_{j\in[k]} \len {\tilde P[x_j',x'_{j+1}]}\\
        &\le O(\e 2^s \cdot \log n) + \dist_{R}(u,p)
    \end{align*}
    as desired.
    
    \medskip \noindent \textbf{Decomposition into subpaths.}
    Clearly $\len Q$ and $\len {Q'}$ are upper-bounded by $O(\e 2^s)$. Now consider a subpath $P_j$ for $j \in [k]$, which runs between $x_{j}$ and $x_{j+1}$. 
    We claim $\dist_{R_j^+}(x_{j}, x_{j+1}) \le 2^s$: indeed,
    \begin{align*}
    \dist_{R_j^+}(x_{j}, x_{j+1})
    &\le O(\e 2^s) + \len {\tilde P [x_{j}', x_{j+1}']}\\
    &\le O(\e 2^s) + \len {\tilde P}\\
    &\le (1+O(\e)) \cdot 2^{s-1}
    \le 2^s   
    \end{align*}
    where the second-to-last inequality follows from the fact that $\len {\tilde P} \le 2^{s-1}$, and the last inequality holds
    for sufficiently small $\e$.
    By definition, $x_{j}$ is an $(s,R_{j})$-portal, $x_{j+1}$ is an $(s,R_{j+1})$-portal, and $R_{j+1}$ is an ancestor piece of $R_{j}$.
    Thus $\canon {x_{j}} {x_{j+1}}$ is an $(s,R_j)$-pair. Thus, $P_j$ is a shortest path between some canonical pair, as desired.
\end{proof}

\subsection{From \Cref{lem:string} to $(1+\e, +\poly\log n)$-distortion}
\label{SS:approx-reduction}

We use the path-straightening technique introduced by Nguyen, Scott, and Seymour~\cite{nguyen2025asymptotic}.
The lemma below is a slight adaptation of statement of 2.2 in \cite{nguyen2025asymptotic} for finite graphs.
We first need to formally introduce the notion of \emph{quasi-isometry}, which is defined as follows.
Let $G$ and $H$ be edge-weighted graphs. A mapping $\phi: V(G)\to V(H)$ is an \EMPH{$(L,C)$-quasi-isometry}, if and only if:
\begin{itemize}
\item for all $u,v\in V(G)$, $\dist_H(\phi(u),\phi(v))\le L\cdot \dist_G(u,v)+C$;
\item for all $u,v\in V(G)$, $\dist_G(u,v)\le L\cdot \dist_H(\phi(u),\phi(v))+C$; and
\item for every $y\in V(H)$, there exists $u\in V(G)$, such that $\dist_H(y,\phi(u))\le C$.
\end{itemize}

\begin{theorem}[{\cite[2.2]{nguyen2025asymptotic}}]
\label{thm:NSS25-2.2}
Let $C\ge 4$, and let $\phi$ be a $(C-1,C)$-quasi-isometry from a graph $G$ to an unweighted graph $H$. 
Let $\pi$ be a shortest path in $G$, with vertices $p_1,\ldots,p_\ell$ in order. 
Then there exists
(i) an edge weight function $w: E(H)\to \mathbb{N}$ with $\max_{e\in E(H)} w(e)\le 32C^4$; 
(ii) a path $\pi'$ of $H$; and 
(iii) a set of vertices $\set{q_j}_{j\in J}$ in the order of $\pi'$ called \EMPH{lampposts}, indexed by $J$ which is a subset of $[\ell]$, such that:
\begin{enumerate}
    \item In the reweighted graph $(H,w)$, $\pi'$ is a shortest path between its endpoints, and for each pair $(j,j')\in J^2$, $\dist_{(H,w)}(q_j,q_{j'})=|j-j'|$;
    \item In the unweighted graph $H$,%
    \begin{itemize}\cramped
        \item for each $q_j$, $\dist_H(\phi(p_j),q_j) \le 2C$;
        \item for each $p_i$, there exists $q_j$ with $|j-i|\le C^2$ and $\dist_{H}(\phi(p_i),q_{j})<C^3$; and
        \item each vertex of $\pi'$ is within distance $3C$ from some $q_j$.
    \end{itemize}
\end{enumerate}
\end{theorem}

\begin{figure}[t]
\centering
\includegraphics[width=0.5\linewidth]{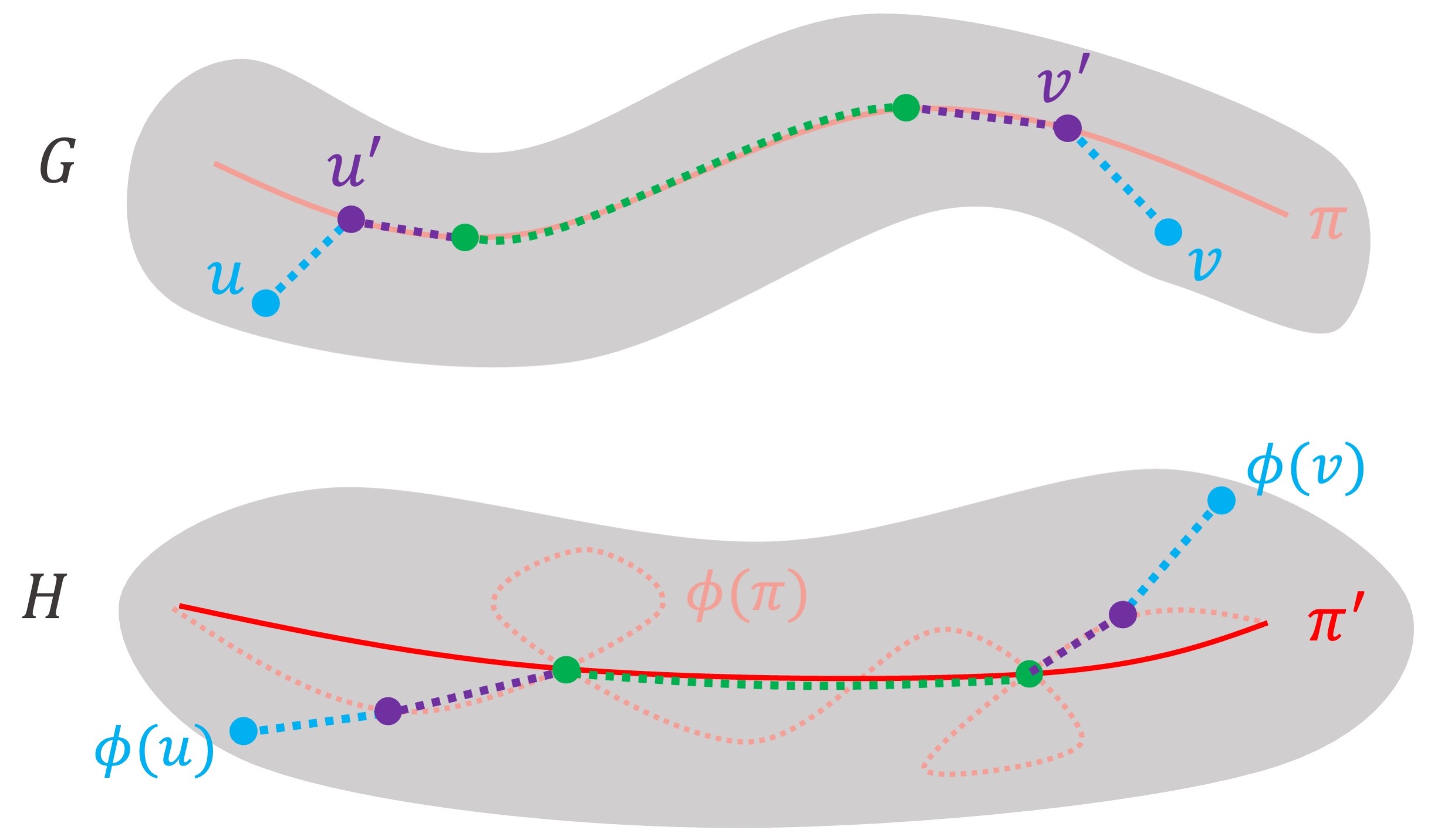}
\caption{A illustration of how $+\poly C$ distortion can be achieved for graph $G$ consisting of a $C$-neighborhood around a shortest path, by applying the \cite{nguyen2025asymptotic} theorem. Vertices $u$ and $v$ are two endpoints of a shortest path $\pi$ in $G$ (the light red path on top), where $G$ lies within the $C$-neighborhood of $\pi$; the image $\phi(\pi)$ in the $O(1)$-distortion emulator $H$ is then straightened to $\pi'$ (dark red). Lampposts are shown in green. All blue and purple segments are short (of length $\poly C$), and the distances between lampposts in $G$ and in $H$ are preserved.}
    \label{fig:straightening}
\end{figure}

\Cref{thm:NSS25-2.2} effectively ``straightens'' a shortest path in $G$. As described in the technical overview, it can be used to covert an $O(1)$-multiplicative distortion emulator of $G$ into a $+O(1)$-additive distortion emulator of $G$, for the toy example where $G$ is a $O(1)$-neighborhood around a shortest path. See \Cref{fig:straightening} for an illustration.

Let $G$ be an unweighted $n$-vertex graph with a distortion-$C_{\ref{thm:multToAdd}}$ embedding into a planar emulator \EMPH{$H_{\rm init}$} that has non-negative integer edge weights and $V(H_{\rm init}) = V(G)$.  
\begin{observation}
\label{obs:bounded-edges}
Without loss of generality, we may assume that every edge of $H_{\rm init}$ has weight at most $C_{\ref{thm:multToAdd}}$.    
\end{observation}
\begin{proof}
One can delete all edges from $H_{\rm init}$ with weight larger than $C$, without affecting the distortion guarantee. Indeed, by assumption every edge $(u,v)$ of $G$ can be mapped to a path in $H_{\rm init}$ between $u$ and $v$ of length at most $C$ (meaning that this path uses edges in $H_{\rm init}$ with weight at most $C$). Moreover, every path in $G$ can be viewed as a sequence of edges, so every path in $G$ of length $\ell$ can be mapped to a path in $H_{\rm init}$ that has uses only edges of length $\le C$ and has total length $\le C \cdot \ell$.
\end{proof}
Define \EMPH{$H$} to be the unweighted graph obtained from $H_{\rm init}$ by subdividing every weighted edge: specifically, we replace each edge $e$ of weight $w_{\rm init}(e)$ in $H_{\rm init}$ with a path of $w_{\rm init}(e)$ unweighted edges in $H$ (recall that $w_{\rm init}(e)$ is a non-negative integer by assumption).
The resulting graph $H$ is a $C_{\ref{thm:multToAdd}}$-distortion \emph{unweighted} planar emulator for $G$, and (by \Cref{obs:bounded-edges}) it has at most $O(C_{\ref{thm:multToAdd}} \cdot n) = O(n)$ vertices.
For the sake of clarity, let $\EMPH{$\phi$}:V(G) \to V(H)$ be the mapping from vertices in $G$ to vertices in $H$ which is a $C_{\ref{thm:multToAdd}}$-distortion embedding. In the rest of this section, we work with $H$ instead of $H_{\rm init}$.
 
Throughout this section, define $\EMPH{$C$} \coloneqq C_{\ref{thm:multToAdd}}+2$, and assume WLOG that $C \ge 4$.
Let $\cO$ be the set of paths in $G$ given by \Cref{lem:string}. 
Consider now a path $\pi\in \cO$ with $\pi=(p_1,\ldots,p_\ell)$. 
We define graph \EMPH{$H_\pi$} as the union of shortest paths between vertices $\phi(p_i)$ and $\phi(p_{i+1})$ in $H$ for every $1 \le i < \ell$. (Clearly, each such path has length at most $C-2$ since $\phi$ is a distortion-$(C-2)$ embedding.)

\begin{claim}
\label{clm:quasi-isometry}
$\phi$ is a $(C-1,C/2 - 1)$-quasi-isometry (and so also a $(C-1,C)$-quasi-isometry)
from $\pi$ to $H_\pi$.
\end{claim}
\begin{proof}
We check the three conditions for quasi-isometry.

First, we show that for every $x \in V(H_\pi)$ there exists $y \in V(\pi)$ such that $\dist_{H_\pi}(x, \phi(y)) \le C/2 - 1$. Since every vertex $x$ of $H_\pi$ lies on the shortest path between some $(\phi(p_i),\phi(p_{i+1}))$ pair, and since such a shortest path has length at most $C-2$, in graph $H_\pi$ (and therefore in $H$), $x$ must lie within distance $(C-2)/2$ from either $\phi(p_i)$ or $\phi(p_{i+1})$. 

Second, we fix vertices $u$ and $v$ in $\pi$ and show that $\dist_{H_\pi}(\phi(u),\phi(v)) \le (C-2) \dist_\pi(u,v)$.
Since $\phi$ is a distortion-$(C-2)$ embedding from $G$ to $H$, for each $i\ge 1$, $\dist_H(\phi(p_i),\phi(p_{i+1}))\le (C-2)\cdot \dist_G(p_i,p_{i+1})$, and so
\[
\begin{split}
\dist_{H_\pi}(\phi(p_i),\phi(p_j))& \le \sum_{t=i}^{j-1}\dist_{H_\pi}(\phi(p_i),\phi(p_{i+1}))=\sum_{t=i}^{j-1}\dist_{H}(\phi(p_i),\phi(p_{i+1}))\\
&\le (C-2)\cdot \sum_{t=i}^{j-1}\dist_G(p_i,p_{i+1})
\le (C-2)\cdot \sum_{t=i}^{j-1}\dist_{\pi}(p_i,p_{i+1})
\le (C-2)\cdot \dist_{\pi}(p_i,p_j)
\end{split}
\]
as desired.

Finally, we show $\dist_{\pi}(u,v) \le (C-1) \cdot \dist_{H_\pi}(\phi(u),\phi(v))$. Consider a shortest path $[x_1, \ldots, x_{\ell+1}]$ between $\phi(u)$ and $\phi(v)$ in $H_\pi$. As $H_\pi$ is unweighted, this path has length $\ell$; that is, $\dist_{H_\pi}(\phi(u),\phi(v)) = \ell$. We aim to show that $\dist_{\pi}(u,v) \le (C-1) \cdot \ell$. As we argued above, for every vertex $x_i$ there exists some $\EMPH{$y_i$} \in V(\pi)$ such that $\dist_{H_\pi}(x_i, \phi(y_i)) \le C/2 -1$; we take $y_0 \coloneqq u$ and $y_{\ell+1} \coloneqq v$.  Observe that for every $i \in [\ell]$, we have that \[\dist_{G}(y_i, y_{i+1}) \le \dist_{H_\pi}(\phi(y_i), \phi(y_{i+1}))
\le \dist_{H_\pi}(\phi(y_i), x_i) + 1 + \dist_{H_\pi}(x_{i+1}, \phi(y_{i+1}))
\le C-1.\]
Now consider the walk \EMPH{$P$} in $G$ between $u$ and $v$, formed by concatenating the shortest paths (in $G$) from $y_i$ to $y_{i+1}$ for all $i \in [\ell]$. We have that $\len P \le (C-1) \cdot \ell$. Moreover, because $\dist_G(y_i, y_{i+1}) \le C-1$, every vertex along the walk $P$ is within distance $(C-1)/2$ in $G$ of some vertex of $\pi$; that is, the walk $P$ clearly lies within the subgraph $G[\cN_G(\pi, 2(C-2))]$ induced by vertices in the $2C$-neighborhood of $\pi$ in $G$. As \Cref{lem:string} asserts that $\pi$ is a shortest path in $G[\cN_G(\pi, 2(C-2))]$, we have $\dist_\pi(u,v) \le \len P \le (C-1)\cdot \ell$.
\end{proof}

\begin{claim}
\label{clm:disjoint}
The collection of graphs $H_\pi$ over every $\pi$ in $\cO$, as subgraphs of $H$, are vertex-disjoint.
\end{claim}
\begin{proof}
Assume for contradiction that some vertex $x$ in $H$ belongs to both $H_{\pi}$ and $H_{\pi'}$, for distinct paths $\pi$ and $\pi'$ in $\cO$. 
From \Cref{clm:quasi-isometry} and the definition of a $(C-1,C/2 - 1)$-quasi-isometry, in graph $H$, $x$ lies within distance $C/2 - 1$ from $\phi(p)$ for some $p\in \pi$, 
and $x$ also lies within distance $C/2 - 1$ from $\phi(p')$ for some $p'\in \pi'$, together implying that $\dist_{H}(\phi(p),\phi(p')) < C$. 
However, from \Cref{lem:string} [$O(C)$-disjointness] property, $\dist_G(\pi,\pi')\ge 2C_{\ref{lem:string}}= 2(C-2) \ge C$, 
and since $\phi$ is a non-contracting embedding from $G$ to $H$ by \Cref{thm:main}, $\dist_H(\phi(\pi),\phi(\pi')) \ge C$, a contradiction.
\end{proof}

From \Cref{clm:quasi-isometry}, $\phi$ is a $(C-1,C)$-quasi-isometry from $\pi$ to $H_\pi$.
For each path $\pi\in \cO$, we
apply \Cref{thm:NSS25-2.2} to $\phi$ and $\pi$, and obtain (i) an edge weighting $w_{H_\pi}$ on edges of $H_\pi$ with maximum weight $32C^4$; and (ii) a path $\pi'$ satisfying the conditions in \Cref{thm:NSS25-2.2}. 
Collecting all edge weightings $\set{w_{H_\pi}}_{\pi\in \cO}$, we define an edge weight function \EMPH{$w$} on $H$ as follows:
\begin{itemize}
\item Define $\EMPH{$C'$} \coloneqq 150\cdot C^7$.
\item For each edge $e$ that lies in some (unique) graph $H_\pi$, set $w(e) \coloneqq w_{H_\pi}(e)$. 
\item For each edge $e$ that does not belong to any $H_\pi$ but has an endpoint  in some $H_\pi$, set $w(e) \coloneqq C'$.
\item For every other edge $e$, set $w(e) \coloneqq 1$.
\end{itemize}
Notice that \Cref{clm:disjoint} is needed for the edge weighting $w$ is well-defined; this is to ensure that every edge $e$ belongs to at most one graph $H_\pi$.

\begin{observation}
The maximum weight $w(e)$ over all edges $e$ in $H$ is $C'$.
\end{observation}

\subsubsection{Proof of \Cref{thm:multToAdd}}
Now we are ready to prove the second main theorem.
We now show that the mapping $\phi$ from $G$ to $H$, now viewed as a mapping from $G$ to the reweighted graph $(H,w)$ with edge weighting $w$ defined above, 
is almost\footnote{Distances may contract by an additive $O(1)$ amount. At the end of the proof, we add new Steiner vertices to get rid of any contraction} a $(1+\e_0, +\poly\log n)$-embedding.

\paragraph{The upper bound proof.}
We first show that for each pair $u,v\in V(G)$,
\[
\dist_{(H,w)} (\phi(u),\phi(v))
\le \Paren{1+O(C^8 \cdot \e_0) } \cdot \dist_{G}(u,v) + O\Paren{ C^9\cdot\e_0^{-4}\cdot \log^{17} n}.
\]
From \Cref{lem:string} [low-hop] property, there is a path $P$ between vertices $u$ and $v$ in $G$,
such that either $\norm{P} < O( C \cdot \e_0^{-4}\cdot \log^{17} n)$,
or $\norm{P} \le (1+\e_0)\cdot\dist_G(u,v)$ and
$P = Q_1 \circ P_1 \circ Q_2 \circ P_2 \circ \dots \circ Q_\ell$,
where $\ell \le O(\e_0^{-3} \log^{15} n)$, and each $P_i$ is a subpath of some path in $\cO$, and $\sum_{i=1}^{\ell} \norm{Q_i} \le \e_0\cdot \norm{P}$.

We first deal with the case where $\len P \ge \Omega(C \cdot \e_0^{-4} \cdot \log^{17}n)$. We have
$\norm{P}=\sum_{i=1}^{\ell} \norm{Q_i} + \sum_{i=1}^{\ell-1} \norm{P_i}$. 
Denote by $x_i$ and $y_i$ the starting and ending points of $Q_i$, respectively. Denote $\EMPH{$\tilde x_i$} \coloneqq \phi(x_i)$ and $\EMPH{$\tilde y_i$} \coloneqq \phi(y_i)$. 
Denote $\EMPH{$\tilde u$} \coloneqq \phi(u)$ and $\EMPH{$\tilde v$} \coloneqq \phi(v)$, so
\[
\dist_{(H,w)}(\tilde u,\tilde v)\le \sum_{i=1}^{\ell}\dist_{(H,w)}(\tilde x_i,\tilde y_i) + \sum_{i=1}^{\ell-1}\dist_{(H,w)}(\tilde y_i,\tilde x_{i+1}).
\]
Each edge $e$ of $H$ is given a weight $w(e)$ at most $C'$. 
Since $\phi$, when viewed as a mapping from $G$ to the unweighted $H$, is a distortion-$(C-1)$ embedding, we have for each $i$, $\dist_{(H,w)}(\tilde x_i,\tilde y_i)\le C'\cdot \dist_{H}(\tilde x_i,\tilde y_i)\le C'\cdot (C-1)\cdot \dist_{G}(x_i, y_i)$, and so
\[
\sum_{i=1}^{\ell}\dist_{(H,w)}(\tilde x_i,\tilde y_i)\le 
C'\cdot (C-1)\cdot \sum_{i=1}^{\ell}\dist_{G}(x_i,y_i)
\le 
C'\cdot (C-1)\cdot \sum_{i=1}^{\ell} \,\norm{Q_i}
= O(C'\cdot C\cdot \e_0)\cdot \norm{P}.
\]

Consider a path $P_i$, which is a subpath of $\pi$ in $\cO$, with endpoints $y \coloneqq y_i$ and $x \coloneqq x_{i+1}$. 
When we applied \Cref{thm:NSS25-2.2} to $\phi$ and $\pi$, we obtained a path $\pi'$ in $H_{\pi}$ (and thus in $H$).
Denote $\pi = (p_1,\ldots,p_\ell)$; there exists $i$ and $i'$ such that $y = p_i$ and $x = p_{i'}$ holds, so $\dist_{\pi}(x,y) = i'-i$. 
By definition of the lampposts $\set{q_j}_{j\in J}$, from \Cref{thm:NSS25-2.2} there exists 
\begin{itemize}
    \item $q_j$ with $|j-i|\le C^2$ and $\dist_{H}(\phi(p_i),q_{j})<C^3$, and
    \item $q_{j'}$ with $|j'-i'|\le C^2$ and $\dist_{H}(\phi(p_{i'}),q_{j'})<C^3$,
\end{itemize}
such that in the reweighted graph $(H,w)$, $\dist_{(H,w)}(q_j,q_{j'})=|j-j'|$.
Note that 
\[
|j-j'|\le |j-i|+|i-i'|+|i'-j'|\le |i-i'|+2C^2.
\]
Therefore,
\[
\begin{split}
\dist_{(H,w)}(\phi(p_i),\phi(p_{i'}))
& \le 
\dist_{(H,w)}(\phi(p_i),q_{j})
+\dist_{(H,w)}(q_j,q_{j'})
+\dist_{(H,w)}(q_{j'},\phi(p_{i'}))\\
& \le C^3\cdot 32C^4  + |j-j'|+ C^3\cdot 32C^4\\
& \le 64C^7 + |i-i'|+2C^2\\
& \le \dist_\pi(p_i,p_{i'})+ 66\cdot C^7.
\end{split}
\]
In other words, for each path $P_i$, we lose an additive distortion of $O(C^7)$ on its length after embedding $P_i$ into the reweighted graph $(H,w)$. Therefore, because $\ell \le O(\e_0^{-3} \log^{15}n)$,
\[
\sum_{i=1}^{\ell-1}\dist_{(H,w)}(\tilde y_i,\tilde x_{i+1})
\le 
\sum_{i=1}^{\ell-1}\dist_{P_i}(y_i,x_{i+1}) + O\Paren{C^7\cdot\e_0^{-3} \log^{15} n}.
\]
Altogether,
\[
\begin{split}
\dist_{(H,w)}(\tilde u,\tilde v) &\le \sum_{i=1}^{\ell}\dist_{(H,w)}(\tilde x_i,\tilde y_i) + \sum_{i=1}^{\ell-1}\dist_{(H,w)}(\tilde y_i,\tilde x_{i+1})\\
&\le O(C'\cdot C\cdot \e_0)\cdot \norm{P} + \sum_{i=1}^{\ell-1} \dist_{P_i}(y_i,x_{i+1}) + O\Paren{C^7\cdot\e_0^{-3} \log^{15} n}\\
&\le \Paren{1+O(C^8\cdot \e_0) } \cdot\dist_{G}(u,v) + O\Paren{C^7\cdot\e_0^{-3} \log^{15} n}
\end{split}
\]
in the case that $\norm{P} \ge \Omega(C\cdot\e_0^{-4}\log^{17} n)$.
On the other hand, if $\norm{P} < O(C\cdot\e_0^{-4} \cdot \log^{17} n)$, then
\[\dist_{(H,w)}(\tilde u, \tilde v) \le
O(C' \cdot C) \cdot \len{P} = O(C^{9} \cdot \e_0^{-4} \cdot \log^{17} n).
\]
because each edge in $P$ can be mapped to a path of at most $C$ edges in the unweighted $H$, and the maximum edge weight in the reweighted graph $(H,w)$ is $C'$. This proves the desired upper bound.

\paragraph{The lower bound proof.}
Consider a pair of vertices $(u,v)$ in $G$, and denote $\tilde u \coloneqq \phi(u)$ and $\tilde v \coloneqq \phi(v)$. 
We will show that
\[
\dist_{(H,w)}(\tilde u,\tilde v)\ge \dist_{G}(u,v)-O(C^7).
\]

Recall that we have reweighted the edges in disjoint subgraphs $\set{H_\pi}_{\pi\in \cO}$. 
Consider the shortest path $P^*$ in $(H,w)$ between $\tilde u$ and $\tilde v$. 
If $P^*$ does not intersect any such subgraph $H_\pi$, then according to our edge weighting $w$, all edges of $P^*$ are given weight at least $1$, and so $\dist_{(H,w)}(\tilde u,\tilde v)\ge \dist_{H}(\tilde u,\tilde v)$. 
Since $\phi$ is a distortion-$C$ embedding from $G$ to $H$, $\dist_{H}(\tilde u,\tilde v)\ge \dist_{G}(u,v)$. 
Altogether, in this case $\dist_{(H,w)}(\tilde u,\tilde v)\ge \dist_{G}(u,v)$.
Therefore, the remaining case is when $P^*$ intersects some subgraph $H_\pi$.

We view $P^*$ as being directed from $\tilde u$ to $\tilde v$. 
Assume that $P^*$ intersects subgraphs $H_{\pi_1},\ldots,H_{\pi_k}$ sequentially. (Note that it may be the case that $\pi_i = \pi_j$ for some $i \neq j$; if $P^*$ wanders out of $H_{\pi_i}$ and then back in, we count the second time that $P^*$ enters $H_{\pi_i}$ as a second subgraph in the list.) For each $1\le i\le k$, denote by $x_i$ the first vertex of $H_{\pi_i}\cap P^*$, and denote by $y_i$ the last vertex of $H_{\pi_i}\cap P^*$ such that the subpath of $P^*$ between $x_i$ and $y_i$ is fully contained in $H_{\pi_i}$. We have
\[
\dist_{(H,w)}(\tilde u,\tilde v)=\dist_{(H,w)}(\tilde u,x_1)+\Paren{\sum_{i=1}^k\dist_{(H,w)}(x_i,y_i)}+\Paren{\sum_{i=1}^{k-1}\dist_{(H,w)}(y_i,x_{i+1})}+\dist_{(H,w)}(y_k,\tilde v).
\]
From our edge weighting $w$, each edge of $P^*[y_i,x_{i+1}]$ is given weight at least $1$, and at least one edge is given weight $C'$, so for each $i$, $\dist_{(H,w)}(y_i,x_{i+1})\ge \dist_{H}(y_i,x_{i+1})+(C'-1)$.

For each $i$, recall that from \Cref{clm:quasi-isometry},
$\phi$ is a $(C-1,C)$-quasi-isometry from $\pi_i$ to $H_{\pi_i}$, so there is some $a_i\in V(\pi_i)$ with $\dist_H(x_i,\phi(a_i))\le C$, and there is some $b_i\in V(\pi_i)$ with $\dist_H(y_i,\phi(b_i))\le C$.
Because the edge weighting $w_{H_{\pi_i}}$ on $H_{\pi_i}$ has maximum edge weight $32 C^4$, we have $\dist_{H,w}(x_i,\phi(a_i))\le 32 C^5$ and $\dist_{H,w}(y_i,\phi(b_i))\le 32 C^5$.
Therefore, 
\[
\dist_{(H,w)}(y_i,x_{i+1})\ge \dist_{H}(y_i,x_{i+1})+(C'-1)\ge \dist_{H}(\phi(b_i),\phi(a_{i+1}))+(C'-1)-64C^5\ge \dist_{G}(b_i,a_{i+1})+C'/2.
\]
On the other hand, from \Cref{thm:NSS25-2.2} and similar analysis in the upper bound proof,
\[
\dist_{(H,w)}(\phi(a_i),\phi(b_i))\ge \dist_{\pi_i}(a_i,b_i) - 66\cdot C^7\ge \dist_{G}(a_i,b_i) - 66\cdot C^7,
\]
and so (using $C \ge 4$), we have
\[\dist_{(H,w)}(x_i,y_i) \ge \dist_{(H,w)}(\phi(a_i),\phi(b_i)) - 64C^5 \ge \dist_{G}(a_i,b_i) - 70\cdot C^7.\]
With similar approach for the pairs $(\tilde u, x_1)$ and $(y_k,\tilde v)$, we get that
\[
\begin{split}
\dist_{(H,w)}(\tilde u,\tilde v) & =\dist_{(H,w)}(\tilde u,x_1) + \Paren{\sum_{i=1}^k\dist_{(H,w)}(x_i,y_i)} + \Paren{\sum_{i=1}^{k-1} \dist_{(H,w)}(y_i,x_{i+1})} + \dist_{(H,w)}(y_k,\tilde v)\\
& \ge 
\dist_{(H,w)}(\tilde u,x_1) + \Paren{\sum_{i=1}^k (\dist_{G}(a_i,b_i)-70\cdot C^7)} + \Paren{\sum_{i=1}^{k-1} (\dist_{G}(b_i,a_{i+1})+C'/2)} + \dist_{(H,w)}(y_k,\tilde v)\\
& \ge 
\dist_{G}(u,a_1)+\Paren{\sum_{i=1}^k \dist_{G}(a_i,b_i)} + \Paren{\sum_{i=1}^{k-1}\dist_{G}(b_i,a_{i+1})} + \dist_{G}(b_k, v) - k\cdot 70\cdot C^7 - O(C^7) + (k-1)\cdot \frac{C'}{2}\\
&\ge \dist_G(u,v)-O(C^7),
\end{split}
\]
where the last inequality is due to $C'=150\cdot C^7$.

To sum up, we have constructed a planar graph $(H,w)$ and an embedding $\phi: V(G)\to V(H)$, such that for each pair $u,v\in V(G)$,
\[
\dist_{G} (u,v)-O(C^7)\le \dist_{(H,w)} (\phi(u),\phi(v))
\le \Paren{1+O(C^8 \cdot \e_0) } \cdot \dist_{G}(u,v) + O\Paren{ C^9\cdot\e_0^{-4}\cdot \log^{17} n}.
\]
To complete the proof of \Cref{thm:multToAdd}, we modify the graph $H$ and the mapping $\phi$ as follows.
For each vertex $u$, denote its image in $H$ by $\tilde u=\phi(u)$. We create a new vertex $\tilde u'$ and connect it to $\tilde u$ by a single edge of length $\Theta(C^7)$. The resulting graph is denoted by $H'$. Clearly, $H'$ is also planar (as the new edge $(\tilde u,\tilde u')$ can be drawn within a tiny neighborhood of $\tilde u$). 
We then change the mapping $\phi$ into a mapping $\phi'$ that maps each vertex $u$ to the new vertex $\tilde u'$ (instead of $\tilde u$).
It is then easy to verify that the distance between every pair $\phi'(u),\phi'(v)$ in $H'$ is $\Theta(C^7)$ greater than the distance between $\phi(u),\phi(v)$ in $H$. 
Consequently, we ensure that for every pair $u,v\in V(G)$,
\[
\begin{split}
\dist_{G} (u,v)\le \dist_{H'} (\phi'(u),\phi'(v))
& \le \Paren{1+O(C^8 \cdot \e_0) } \cdot \dist_{G}(u,v) + O\Paren{ C^9\cdot\e_0^{-4}\cdot \log^{17} n} + \Theta(C^7)\\
& \le \Paren{1+O(C^8 \cdot \e_0) } \cdot \dist_{G}(u,v) + O\Paren{ C^9\cdot\e_0^{-4}\cdot \log^{17} n}
\end{split}
\]
as claimed in \Cref{thm:multToAdd}, after rescaling $\e_0 \gets \e_0 / \Theta(C^8)$.

\section{Applications with \boldmath{$(1+\e, +O(1))$}-Distortion}
\label{S:applications}

\subsection{Tree cover and distance oracles}

In this section we prove \Cref{thm:tree-cover}. Note that the assumption of \Cref{thm:tree-cover} is that the input graph $G$ has an $O(1)$-distortion planar emulator; by applying Steiner point removal for planar graphs (as in \Cref{cor:small-emulator}), we may assume that $G$ has an $O(1)$-distortion planar emulator $H$ with $V(H) = V(G)$.

\subsubsection{Sparse cover}
Recall that a cluster of $G$ is a connected subset of vertices in $G$. A \EMPH{$(\beta, s, \Delta)$-sparse partition cover} of a graph $G$ is a collection $\cC$ of clusters, such that (1) every cluster has (weak) diameter at most $\Delta$; 
(2) for every $v \in V(G)$ there exists a cluster $C \in \cC$ such that all vertices within distance $\Delta/\beta$ to $v$ are in $C$; 
(3) the clusters in $\cC$ can be partitioned into $s$ partitions $\cP_1, \ldots, \cP_s$ (where each $\cP_i$ is a partition of $V(G)$ into clusters).
For any $\Delta > 0$, planar graphs admit $(O(1), O(1), \Delta)$-sparse partition cover. String graphs are also known to admit $(O(1), O(1), \Delta)$-sparse partition cover. See \cite[Theorem 1.3]{abrishami2025burling} for an alternative proof of this fact using \emph{asymptotic dimension}.

For the sake of completeness, we remark that the existence of an $O(1)$-distortion planar emulator for $G$ can be used to easily construct $(O(1), O(1), \Delta)$-sparse partition cover for $G$, through the following folklore observation.
\begin{observation}[Folklore]
Let $H$ be a $C$-distortion emulator for graph $G$ with $V(H) = V(G)$, and let $\cC_H$ be any $(\beta, s, \Delta)$-sparse partition cover of $H$.
Then there is a $(C \cdot \beta, s, \Delta)$-sparse partition cover for $G$. 
\end{observation}
\begin{proof}
For a graph $G$, a vertex $v\in V(G)$, and a non-negative number $r$, let \EMPH{$B_G(v, r)$} be defined as the set of vertices within distance $r$ from $v$. 
We construct a set of clusters $\cC$ in $G$ by iterating over $\cC_H$: for each $C_H \in \cC_H$, add every connected component of $G[C_H]$ as a cluster in $\cC$.
Observe that every cluster in $\cC$ is connected and has (weak) diameter at most $\Delta$, and $\cC$ can be partitioned into $s$ partitions of $V(G)$. Moreover, observe that for any $v \in V(G)$, by definition of emulator, every vertex in the ball $\EMPH{$B_G^{\rm small}$} \coloneqq B_G(v, \Delta/(C \cdot \beta))$ %
is contained in the ball $\EMPH{$B_H^{\rm big}$} \coloneqq B_H(v, \Delta/\beta)$. %
Because there is some cluster $C_H \in \cC_H$ that contains the ball $B_H^{\rm big}$, we have that $C_H$ contains $B_G^{\rm small}$; in particular, every vertex of $B_G^{\rm small}$ is in the same connected component of the induced graph $G[C_H]$, so some cluster of $\cC$ contains $B_G^{\rm small}$.
\end{proof}

\subsubsection{Reduction to additive tree cover}
We say a tree cover $\cT$ is \EMPH{$\ell$-bounded} if every tree in $\cT$ has diameter at most $\ell$.
\cite{CCLMST23} show (in their Lemma 1.7) that if one can construct an $O(\diam(G))$-bounded $(1, +\e\cdot \diam(G))$-tree cover of size $\tau(\e)$ for any planar graph $G$, then one can construct a $(1+\e, +0)$-tree cover of size $O(\tau(\e) \log (1/\e))$ for any planar graph $G$.
Their proof relied only on the fact that the class of planar graphs is closed under induced subgraphs, and the fact that planar graphs have \emph{hierarchical padded partition}---which itself follows from the existence of sparse partition cover \cite{KLMN05}. As such, the same reduction applies for string graphs (and more generally for any graph class $\mathcal{G}$ that admits $O(1)$-distortion planar emulators and is closed under induced subgraphs).
It is immediate\footnote{Specifically, in Claim 8.7 of \cite{CCLMST23} the distance bound becomes $O(r_i + 1)$ instead of $O(r_i)$, and thus the distance bound in Equation 5 becomes $(1+O(r^2 \e)) \cdot \dist_X(x, y) + O(1)$ instead of $(1+O(r^2 \e)) \cdot \dist_X(x, y)$; everything else remains the same. We omit the full details.} from the reduction of \cite{CCLMST23} that, if one can only construct $(1, +\e \diam(G) + O(1))$-tree covers for any $G \in \mathcal{G}$, then one can construct $(1+\e, +O(1))$ tree covers.

\subsubsection{Constructing the additive tree cover}
In this section, we construct additive-distortion tree cover for graphs with $O(1)$-distortion planar emulators; by the discussion in the previous section, this is enough to prove \Cref{thm:tree-cover}.
\begin{lemma}
\label{lem:additive-tree-cover}
    Let $\e > 0$. Let $G$ be a graph with a $C$-distortion planar emulator $H$ with $V(G) = V(H)$. Let $\Delta$ be the diameter of $G$. There is an $O(\Delta)$-bounded $(1, +\e \Delta + O(C))$-tree cover for $G$ with size $O(C^3 \cdot \e^{-3})$.
\end{lemma}
We use, as a black box, the result of \cite{CCLMST23} that any planar graph $H$ admits a small-size $(1, +\e \diam(H))$-tree cover. The tree cover constructed by \cite{CCLMST23} is stronger than just guaranteeing $(1, +\e \diam(H))$-distortion; it has a certain \EMPH{root preservation} property, as summarized in the following lemma.
\begin{lemma}[Root preservation property of {\cite[Claim 7.1]{CCLMST23}}]
\label{lem:planar-tree-additive}
    Let $\e > 0$, and let $H$ be a planar graph with diameter at most $\Delta_H$. There is a set $\mathcal{F}_H$ of $O(\e^{-3})$ forests in $H$, where each tree in each forest has a designated root vertex such that:
    \begin{itemize}
        \item \textnormal{[Root preservation.]} For any pair of vertices $(u,v)$ in $H$ and any walk $P$ between $u$ and $v$ in $H$ with $\norm{P} \le \Delta_H$,
    there is some tree $T$ in a forest of $\mathcal{F}$, with root $r$, such that $V(T)$ contains $u$ and $v$, and
    $\dist_H(r, P) \le \e \Delta_H.$
    \end{itemize}
\end{lemma}
Our \Cref{lem:planar-tree-additive} is a bit different from what is explicitly stated in \cite[Claim 7.1]{CCLMST23}. Namely, only a \emph{consequence} of the [root preservation] property is stated in \cite{CCLMST23}:
they state that, if we let $P$ to be the shortest path between $u$ and $v$, then $d_T(u,r) + d_T(r, v) \le d_H(u,v) + O(\e \Delta_H)$. Additionally, \cite{CCLMST23} take $\Delta_H$ to be exactly equal to the diameter of $H$, rather than allowing $\Delta_H$ to be an upper bound on the diameter of $H$. Nevertheless, it is immediate from the proof of \cite[Claim 7.1]{CCLMST23} that our strengthened version of the claim (\Cref{lem:planar-tree-additive}) also holds.
Using  this lemma, we now construct the additive tree cover for planar graphs.

\begin{proof}[of \Cref{lem:additive-tree-cover}]
    Let $H$ be a $C$-distortion planar emulator for $G$ with $V(H) = V(G)$.
    Observe that the diameter of $H$ is at most $\Delta_H \coloneqq C \cdot \Delta$.
    
    Let \EMPH{$\mathcal{F}_H$} be the set of forests on $H$ provided by \Cref{lem:planar-tree-additive}.
    We construct a corresponding set of forests \EMPH{$\mathcal{F}$} on $G$. For every forest $F_H \in \mathcal{F}_H$, for every tree $T_H$ in $F_H$, we define a corresponding tree $T$ as follows: 
    $T$ is simply a (edge-weighted) star with vertex set $V(T_H)$ and root $r$, where the weight of each edge $(r, v)$ in $T$ is set to $\dist_G(v, r)$.
    Define the forest $F$ to be the disjoint union of the trees $T$ corresponding to the trees $T_H$ in $F_H$.
    The set $\mathcal{F}$ is the set of all forests constructed this way.
    We now construct a set of trees \EMPH{$\cT$} from $\cF$: for every forest $F \in \cF$, choose an arbitrary vertex $v$ in $F$, and for every tree $T$ in $F$ (other than the tree containing $v$), add a weight-$\Delta$ edge between $v$ and an arbitrary vertex in $T$.
    
    We claim that $\cT$ is an $O(\Delta)$-bounded $(1, +O(C \e \Delta + C))$-tree cover for $G$ with size $O(\e^{-3})$. This is sufficient to prove the lemma, after rescaling $\e \gets \e / O(C)$.
    First observe that every tree $T \in \cT$ is a dominating tree of $G$. Indeed, we have $V(T) = V(G)$; and every edge $(u,v)$ in $T$ has weight at least $\dist_G(u,v)$, because by construction the weight of $(u,v)$ is either $\Delta$ or $\dist_G(u,v)$. Additionally, the size of $\cT$ is the same as the size of $\cF$, which is $O(\e^{-3})$.
    It remains to show that, for every pair of vertices $(u,v)$ in $G$, there is a tree $T \in \cT$ with small distortion bound.
    To this end, let $\EMPH{$P$} = [p_1, p_2, \ldots, p_\ell]$ be the shortest path between $u$ and $v$ in $G$. Let \EMPH{$P_H$} be the corresponding walk in $H$, which is the shortest walk in $H$ that walks from $p_1$ to $p_2$ to $p_3$ and so on; we denote this as $P_H = p_1 \rightsquigarrow p_2 \rightsquigarrow \ldots \rightsquigarrow p_\ell$.
    Because $G$ is unweighted, every path $p_i \rightsquigarrow p_{i+1}$ in $H$ has length at most $C$; thus, every vertex $v$ in $P_H$ satisfies $\dist_G(v, P) \le \dist_H(v, P) \le C$.
    Moreover, $P_H$ has length at most $C \cdot \norm{P} \le C \cdot \Delta = \Delta_H$.
    
    Here is the key idea: by the [root preservation] property $\cF_H$ applied on vertices $u$, $v$ and path $P_H$, there is some tree $T_H$ with root $r$ in some forest $F_H \in \cF_H$, such that $u$ and $v$ are in $T_H$, and $\dist_H(r, P_H) \le \e \Delta_H$. (Note that we can apply the [root preservation] property because $P_H$ has length at most $\Delta_H$.) In the corresponding tree $T \in \cT$, we have $\dist_T(u,v) = \dist_G(u,r) + \dist_G(r, v)$. We claim that
    \[\dist_G(r, u) + \dist_G(r, v) \le  \dist_G(u,v) + O(C\e \Delta + C) = \dist_G(u,v) + O(C \e \Delta + C).\]
    Indeed, let $p_H \in P_H$ be the vertex on $P_H$ that minimizes $\dist_G(r,P_H)$; we have $\dist_G(r, p_H) \le \dist_H(r,p_H) \le \e \Delta_H$.
    Let $p \in P$ be the vertex in $P$ that minimizes $\dist_G(p, p_H)$; we have $\dist_G(p, p_H) \le C$.
    By triangle inequality,
    \[\dist_G(r, p) \le \e \Delta_H + C.\]
    By triangle inequality, $\dist_G(u, r) \le \dist_G(u, p) + \e \Delta_H + C$ and $\dist_G(r, v) \le \dist_G(p, v) + \e \Delta_H + C$.
    Finally, observe that $\dist_G(u,p) + \dist_G(p, v) = \dist_G(u,v)$ because $p$ lies is on a shortest path $P$ between $u$ and $v$ in $G$.
    We conclude that $\dist_G(r,u) + \dist_G(r,v) \le  \dist_G(u,v) + O(\e \Delta_H + C) = \dist_G(u,v) + O(C\e \Delta + C)$ as desired.
\end{proof}

\subsection{Bounded-treewidth embedding}

In this section we prove \Cref{thm:tw-embedding}.
Following \cite{CCLMST23}, we use a tree cover with additional properties --- this includes the [root preservation] property discussed in \Cref{lem:planar-tree-additive} in the previous section (recall that this [root preservation] property is stronger than explicitly stated in \cite{CCLMST23}, but nonetheless is proven in \cite{CCLMST23}), as well as low-hop and cluster disjointness properties.
\begin{lemma}[{\cite[Claim 7.1]{CCLMST23}}]
\label{lem:planar-additive}
    Let $\e > 0$, and let $H$ be a planar graph with diameter at most $\Delta$. There is a set $\mathcal{F}_H$ of $O(\e^{-3})$ forests of rooted trees in $H$, and a partition $\cC$ of $H$ into connected clusters of (strong) diameter $\e \Delta$, such that:
    \begin{itemize}
    \item \textnormal{[Low-hop.]} For every pair of vertices $u$ and $v$ in $H$, there is a path in $H$ between $u$ and $v$ that intersects at most $O(\e^{-1})$ clusters.
    \item \textnormal{[Root preservation.]} %
    For any pair of vertices $(u,v)$ in $H$ and any walk $P$ between $u$ and $v$ in $H$ with $\norm{P} \le \Delta$, 
    there is some tree $T$ in a forest of $\mathcal{F}$, with root $r$, such that $V(T)$ contains $u$ and $v$, and
    $\dist_H(r, P) \le \e \Delta.$
    \item \textnormal{[Cluster disjointness.]} Every forest $\cF$ is a spanning forest (i.e., it is a subgraph of $H$). Further, no two trees in any forest of $\cF$ intersect the same cluster in $\cC$.
    \end{itemize}
\end{lemma}
\begin{proof}[of \Cref{thm:tw-embedding}]

Let $H$ be a $C$-distortion planar emulator of $G$ with $V(H) = V(G)$ and $C = O(1)$. Let $\Delta$ be the diameter of $G$, and observe that the diameter of $H$ is at most $C \cdot \Delta$. Let $\cC$ be the clusters and $\cF$ be the set of forests obtained by applying \Cref{lem:planar-additive} on $H$.
For every cluster $L$ in $\cC$, choose an arbitrary vertex $v_L$ in $L$ to be the \EMPH{center vertex} of $L$.
We define a graph \EMPH{$\hat G$} as follows.
\begin{enumerate}
    \item Initialize $\hat G$ as the \emph{cluster graph} with vertex set $\set{v_L : L \in \cC}$, obtained from $H$ by contracting every cluster of $\cC$ into its center vertex.
    \item    For each cluster $L$ in $\cC$, for every non-center vertex $v$ in $L$, add $v$ to $\hat G$ and add an edge between $v$ and the center vertex $v_L$; in other words, replace each vertex $v_L$ in $\hat G$ with a star centered at $v_L$. Observe that $V(\hat G) = V(H) = V(G)$.
    \item  Now augment $\hat G$ with additional edges. For every tree $T_H$ in every forest in $\cF_H$, and for every vertex $v$ in $T_H$, add an edge in $\hat G$ between the root of $T_H$ and $v$.
    \item Set the weight of every edge $(u,v)$ in $\hat G$ to be $\dist_G(u,v)$.
\end{enumerate}
The only difference between this graph and the one in \cite{CCLMST23} is that in the graph constructed in \cite{CCLMST23} sets the weight according to distances in $H$.
They show that the treewidth of $\hat G$ is at most $O(\e^{-4})$ \cite[Lemma 7.7]{CCLMST23}.
It remains to show the distortion bound for $\hat G$: that is, for every pair of vertices $u,v$ in $G$, 
\[\dist_G(u,v) \le \dist_{\hat G}(u,v) \le \dist_G(u,v) + O(C\e \Delta + C).\]
Rescaling $\e \gets \e/\Theta(C)$ then proves the lemma. 
The lower bound $\dist_G(u,v) \le \dist_{\hat G}(u,v)$ is immediate from the setting of edge weights.
The upper bound $\dist_{\hat G}(u,v) \le \dist_G(u,v) + O(C\e \Delta+C)$ follows from an identical argument as in the proof of \Cref{lem:additive-tree-cover}: the [root preservation] property implies
that there is some tree $T_H$ with root $r$ in some forest of $\cF_H$ such that $u$ and $v$ are in $T_H$, and $\dist_G(u,r) + \dist_G(r, v) \le \dist_G(u,v) + O(C\e \Delta+C)$, and (by construction of $\hat G$) we have $\dist_{\hat G}(u,v) \le \dist_G(u,r) + \dist_G(r,v)$.
\end{proof}

{
\small
\bibliographystyle{alphaurl}
\bibliography{main}

@inproceedings{CCLMST23,
  title = {Covering planar metrics (and beyond): O(1) trees suffice},
  author = {Chang, Hsien-Chih and Conroy, Jonathan and Le, Hung and Milenkovic, Lazar and Solomon, Shay and Than, Cuong},
  booktitle = {2023 IEEE 64th Annual Symposium on Foundations of Computer Science (FOCS)},
  pages = {2231--2261},
  year = {2023},
  organization = {IEEE},
  timestamp = {Fri, 04 Jul 2025 01:00:00 +0200},
  biburl = {https://dblp.org/rec/conf/focs/ChangCLMST23.bib},
  bibsource = {dblp computer science bibliography, https://dblp.org},
  doi = {10.1109/FOCS57990.2023.00139},
  _bib2doi_selected = {dblp:/rec/conf/focs/ChangCLMST23.bib},
  _bib2doi_confirmed = {true},
}

@article{KM91,
  title = {String graphs requiring exponential representations},
  author = {Jan Kratochvíl and Jiří Matoušek},
  journal = {Journal of Combinatorial Theory, Series B},
  volume = {53},
  number = {1},
  pages = {1-4},
  year = {1991},
  issn = {0095-8956},
  doi = {10.1016/0095-8956(91)90050-T},
  _bib2doi_old_doi = {https://doi.org/10.1016/0095-8956(91)90050-T},
  _bib2doi_finished = {true},
}

@article{BLT14,
  title = {Sparse Covers for Planar Graphs and Graphs That Exclude a Fixed Minor},
  author = {Busch, Costas and LaFortune, Ryan and Tirthapura, Srikanta},
  year = {2014},
  month = {jul},
  journal = {Algorithmica},
  volume = {69},
  number = {3},
  pages = {658--684},
  issn = {0178-4617, 1432-0541},
  doi = {10.1007/s00453-013-9757-4},
  urldate = {2022-04-28},
  langid = {english},
  annotation = {TLDR: For every n node graph that excludes a minor of a fixed size, this work presents an algorithm that yields a cover with radius no more than 4{$\gamma$} and degree O(logn), a significant improvement over previous results.},
  timestamp = {Sun, 02 Jun 2019 01:00:00 +0200},
  biburl = {https://dblp.org/rec/journals/algorithmica/BuschLT14.bib},
  bibsource = {dblp computer science bibliography, https://dblp.org},
  _bib2doi_selected = {dblp:/rec/journals/algorithmica/BuschLT14.bib},
  _bib2doi_confirmed = {true},
}

@inproceedings{CFKL20,
  title = {On light spanners, low-treewidth embeddings and efficient traversing in minor-free graphs},
  author = {{Cohen-Addad}, Vincent and Filtser, Arnold and Klein, Philip N. and Le, Hung},
  booktitle = {2020 IEEE 61st Annual Symposium on Foundations of Computer Science (FOCS)},
  pages = {589--600},
  year = {2020},
  organization = {IEEE},
  timestamp = {Sun, 05 May 2024 01:00:00 +0200},
  biburl = {https://dblp.org/rec/conf/focs/Cohen-AddadFKL20.bib},
  bibsource = {dblp computer science bibliography, https://dblp.org},
  doi = {10.1109/FOCS46700.2020.00061},
  _bib2doi_selected = {dblp:/rec/conf/focs/Cohen-AddadFKL20.bib},
  _bib2doi_confirmed = {true},
}

@inproceedings{CCLMST24,
  title = {Shortcut partitions in minor-free graphs: Steiner point removal, distance oracles, tree covers, and more},
  author = {Chang, Hsien-Chih and Conroy, Jonathan and Le, Hung and Milenkovi{\'c}, Lazar and Solomon, Shay and Than, Cuong},
  booktitle = {Proceedings of the 2024 Annual ACM-SIAM Symposium on Discrete Algorithms (SODA)},
  pages = {5300--5331},
  year = {2024},
  organization = {SIAM},
  timestamp = {Fri, 04 Jul 2025 01:00:00 +0200},
  biburl = {https://dblp.org/rec/conf/soda/ChangCLMST24.bib},
  bibsource = {dblp computer science bibliography, https://dblp.org},
  doi = {10.1137/1.9781611977912.191},
  _bib2doi_selected = {dblp:/rec/conf/soda/ChangCLMST24.bib},
  _bib2doi_confirmed = {true},
}

@article{CS19,
  title = {Approximate shortest paths and distance oracles in weighted unit-disk graphs},
  author = {Chan, Timothy M. and Skrepetos, Dimitrios},
  journal = {Journal of Computational Geometry},
  volume = {10},
  number = {2},
  pages = {3--20},
  year = {2019},
  timestamp = {Thu, 10 Sep 2020 01:00:00 +0200},
  biburl = {https://dblp.org/rec/journals/jocg/ChanS19a.bib},
  bibsource = {dblp computer science bibliography, https://dblp.org},
  doi = {10.20382/jocg.v10i2a2},
  _bib2doi_selected = {dblp:/rec/journals/jocg/ChanS19a.bib},
  _bib2doi_confirmed = {true},
}

@inproceedings{friggstad2025qptas,
  author = {Friggstad, Zachary and Rezapour, Mohsen and Salavatipour, Mohammad R. and Sun, Hao},
  title = {{A QPTAS for Facility Location on Unit Disk Graphs}},
  booktitle = {19th International Symposium on Algorithms and Data Structures (WADS 2025)},
  pages = {27:1--27:18},
  series = {Leibniz International Proceedings in Informatics (LIPIcs)},
  isbn = {978-3-95977-398-0},
  issn = {1868-8969},
  year = {2025},
  volume = {349},
  editor = {Morin, Pat and Oh, Eunjin},
  publisher = {Schloss Dagstuhl -- Leibniz-Zentrum f{\"u}r Informatik},
  address = {Dagstuhl, Germany},
  url = {https://drops.dagstuhl.de/entities/document/10.4230/LIPIcs.WADS.2025.27},
  urn = {urn:nbn:de:0030-drops-242586},
  doi = {10.4230/LIPIcs.WADS.2025.27},
  annote = {Keywords: Facility Location, Unit Disk Graphs, Approximation Algorithms},
  timestamp = {Fri, 29 Aug 2025 01:00:00 +0200},
  biburl = {https://dblp.org/rec/conf/wads/FriggstadRSS25.bib},
  bibsource = {dblp computer science bibliography, https://dblp.org},
  _bib2doi_selected = {dblp:/rec/conf/wads/FriggstadRSS25.bib},
  _bib2doi_confirmed = {true},
}

@inproceedings{CH23,
  title = {Constant-hop spanners for more geometric intersection graphs, with even smaller size},
  author = {Chan, Timothy M. and Huang, Zhengcheng},
  booktitle = {Proc. 39th International Symposium on Computational Geometry (SoCG 2023)},
  year = {2023},
  organization = {Schloss Dagstuhl--Leibniz-Zentrum f{\"u}r Informatik},
  doi = {10.4230/LIPIcs.SoCG.2023.23},
  _bib2doi_finished = {true},
}

@article{CT23,
  title = {Hop-spanners for geometric intersection graphs},
  volume = {14},
  url = {https://jocg.org/index.php/jocg/article/view/4475},
  doi = {10.20382/jocg.v14i2a3},
  number = {2},
  journal = {Journal of Computational Geometry},
  author = {Conroy, Jonathan and Tóth, Csaba},
  year = {2023},
  month = {Dec.},
  pages = {26--64},
  _bib2doi_finished = {true},
}

@inproceedings{BFI23,
  title = {Coresets for Clustering in Geometric Intersection Graphs},
  author = {Bandyapadhyay, Sayan and Fomin, Fedor V\/. and Inamdar, Tanmay},
  booktitle = {39th International Symposium on Computational Geometry (SoCG 2023)},
  pages = {10--1},
  year = {2023},
  organization = {Schloss Dagstuhl--Leibniz-Zentrum f{\"u}r Informatik},
  doi = {10.4230/LIPIcs.SoCG.2023.10},
  _bib2doi_finished = {true},
}

@inproceedings{BKKNP22,
  title = {Towards Sub-Quadratic Diameter Computation in Geometric Intersection Graphs},
  author = {Bringmann, Karl and {Kisfaludi-Bak}, S{\'a}ndor and K{\"u}nnemann, Marvin and Nusser, Andr{\'e} and Parsaeian, Zahra},
  booktitle = {38th International Symposium on Computational Geometry},
  pages = {1--16},
  year = {2022},
  organization = {Schloss Dagstuhl},
  timestamp = {Wed, 01 Jun 2022 01:00:00 +0200},
  biburl = {https://dblp.org/rec/conf/compgeom/BringmannKKNP22.bib},
  bibsource = {dblp computer science bibliography, https://dblp.org},
  doi = {10.4230/LIPIcs.SoCG.2022.21},
  _bib2doi_selected = {dblp:/rec/conf/compgeom/BringmannKKNP22.bib},
  _bib2doi_confirmed = {true},
}

@article{CS19_apsp,
  title = {All-Pairs Shortest Paths in Geometric Intersection Graphs},
  author = {Chan, Timothy M. and Skrepetos, Dimitrios},
  journal = {Journal of Computational Geometry},
  volume = {10},
  number = {1},
  pages = {27--41},
  year = {2019},
  timestamp = {Thu, 10 Sep 2020 01:00:00 +0200},
  biburl = {https://dblp.org/rec/journals/jocg/ChanS19.bib},
  bibsource = {dblp computer science bibliography, https://dblp.org},
  doi = {10.20382/jocg.v10i1a2},
  _bib2doi_selected = {dblp:/rec/journals/jocg/ChanS19.bib},
  _bib2doi_confirmed = {true},
}

@inproceedings{CGL24,
  title = {Computing Diameter+2 in Truly-Subquadratic Time for Unit-Disk Graphs},
  author = {Chang, Hsien-Chih and Gao, Jie and Le, Hung},
  booktitle = {40th International Symposium on Computational Geometry (SoCG 2024)},
  pages = {38--1},
  year = {2024},
  organization = {Schloss Dagstuhl--Leibniz-Zentrum f{\"u}r Informatik},
  doi = {10.4230/LIPIcs.SoCG.2024.38},
  _bib2doi_finished = {true},
}

@incollection {CCGKLZ25,
    AUTHOR = {Chan, Timothy M. and Chang, Hsien-Chih and Gao, Jie and
              Kisfaludi-Bak, S\'andor and Le, Hung and Zheng, Da Wei},
     TITLE = {Truly subquadratic time algorithms for diameter and related
              problems in graphs of bounded {VC}-dimension},
 BOOKTITLE = {2025 {IEEE} 66th {A}nnual {S}ymposium on {F}oundations of
              {C}omputer {S}cience---{FOCS} 2025},
     PAGES = {2728--2765},
 PUBLISHER = {IEEE Comput. Soc. Press, Los Alamitos, CA},
      YEAR = {2025},
      ISBN = {979-8-3315-7132-0},
   MRCLASS = {68R10},
  MRNUMBER = {5048107},
       DOI = {10.1109/FOCS63196.2025.00140}
}

@article{YXD12,
  title = {Compact and low delay routing labeling scheme for unit disk graphs},
  author = {Yan, Chenyu and Xiang, Yang and Dragan, Feodor F.},
  journal = {Computational Geometry},
  volume = {45},
  number = {7},
  pages = {305--325},
  year = {2012},
  publisher = {Elsevier},
  timestamp = {Thu, 11 Feb 2021 00:00:00 +0100},
  biburl = {https://dblp.org/rec/journals/comgeo/YanXD12.bib},
  bibsource = {dblp computer science bibliography, https://dblp.org},
  doi = {10.1016/j.comgeo.2012.01.015},
  _bib2doi_selected = {dblp:/rec/journals/comgeo/YanXD12.bib},
  _bib2doi_confirmed = {true},
}

@inproceedings{catusse2010planar,
  title = {Planar hop spanners for unit disk graphs},
  author = {Catusse, Nicolas and Chepoi, Victor and Vax{\`e}s, Yann},
  booktitle = {International Symposium on Algorithms and Experiments for Sensor Systems, Wireless Networks and Distributed Robotics},
  pages = {16--30},
  year = {2010},
  organization = {Springer},
  timestamp = {Sat, 09 Apr 2022 01:00:00 +0200},
  biburl = {https://dblp.org/rec/conf/algosensors/CatusseCV10.bib},
  bibsource = {dblp computer science bibliography, https://dblp.org},
  doi = {10.1007/978-3-642-16988-5_2},
  _bib2doi_selected = {dblp:/rec/conf/algosensors/CatusseCV10.bib},
  _bib2doi_confirmed = {true},
}

@article{biniaz2020plane,
  title = {Plane hop spanners for unit disk graphs: Simpler and better},
  author = {Biniaz, Ahmad},
  journal = {Computational Geometry},
  volume = {89},
  pages = {101622},
  year = {2020},
  publisher = {Elsevier},
  timestamp = {Mon, 15 Jun 2020 01:00:00 +0200},
  biburl = {https://dblp.org/rec/journals/comgeo/Biniaz20.bib},
  bibsource = {dblp computer science bibliography, https://dblp.org},
  doi = {10.1016/j.comgeo.2020.101622},
  _bib2doi_selected = {dblp:/rec/journals/comgeo/Biniaz20.bib},
  _bib2doi_confirmed = {true},
}

@inproceedings{li2002distributed,
  title = {Distributed construction of a planar spanner and routing for ad hoc wireless networks},
  author = {Li, Xiang-Yang and Calinescu, Gruia and Wan, Peng-Jun},
  booktitle = {Proceedings. Twenty-First Annual Joint Conference of the IEEE Computer and Communications Societies},
  volume = {3},
  pages = {1268--1277},
  year = {2002},
  organization = {IEEE},
  timestamp = {Sun, 06 Oct 2024 01:00:00 +0200},
  biburl = {https://dblp.org/rec/conf/infocom/LiCW02.bib},
  bibsource = {dblp computer science bibliography, https://dblp.org},
  doi = {10.1109/INFCOM.2002.1019377},
  _bib2doi_selected = {dblp:/rec/conf/infocom/LiCW02.bib},
  _bib2doi_confirmed = {true},
}

@inproceedings{HZ24,
  author = {Harb, Elfarouk and Huang, Zhengcheng and Zheng, Da Wei},
  title = {Shortest Path Separators in Unit Disk Graphs},
  booktitle = {32nd Annual European Symposium on Algorithms (ESA 2024)},
  pages = {66:1--66:14},
  series = {Leibniz International Proceedings in Informatics (LIPIcs)},
  isbn = {978-3-95977-338-6},
  issn = {1868-8969},
  year = {2024},
  volume = {308},
  editor = {Chan, Timothy and Fischer, Johannes and Iacono, John and Herman, Grzegorz},
  publisher = {Schloss Dagstuhl -- Leibniz-Zentrum f{\"u}r Informatik},
  address = {Dagstuhl, Germany},
  url = {https://drops.dagstuhl.de/entities/document/10.4230/LIPIcs.ESA.2024.66},
  urn = {urn:nbn:de:0030-drops-211375},
  doi = {10.4230/LIPIcs.ESA.2024.66},
  conference = {ESA 2024},
  arxiv = {https://doi.org/10.48550/arXiv.2407.15980},
  timestamp = {Tue, 24 Sep 2024 01:00:00 +0200},
  biburl = {https://dblp.org/rec/conf/esa/HarbHZ24.bib},
  bibsource = {dblp computer science bibliography, https://dblp.org},
  _bib2doi_selected = {dblp:/rec/conf/esa/HarbHZ24.bib},
  _bib2doi_confirmed = {true},
}

@inproceedings{lee2017separators,
  title = {Separators in Region Intersection Graphs},
  author = {Lee, James R.},
  booktitle = {8th Innovations in Theoretical Computer Science Conference (ITCS 2017)},
  pages = {1--1},
  year = {2017},
  organization = {Schloss Dagstuhl--Leibniz-Zentrum f{\"u}r Informatik},
  doi = {10.4230/LIPIcs.ITCS.2017.1},
  _bib2doi_finished = {true},
}

@article{filtser2024scattering,
  title = {Scattering and sparse partitions, and their applications},
  author = {Filtser, Arnold},
  journal = {ACM Transactions on Algorithms},
  volume = {20},
  number = {4},
  pages = {1--42},
  year = {2024},
  publisher = {ACM New York, NY},
  timestamp = {Wed, 06 Nov 2024 00:00:00 +0100},
  biburl = {https://dblp.org/rec/journals/talg/Filtser24.bib},
  bibsource = {dblp computer science bibliography, https://dblp.org},
  doi = {10.1145/3672562},
  _bib2doi_selected = {dblp:/rec/journals/talg/Filtser24.bib},
  _bib2doi_confirmed = {true},
}

@inproceedings{hershkowitz20221,
  title = {{$O(1)$} {Steiner} Point Removal in Series-Parallel Graphs},
  author = {Hershkowitz, D. Ellis and Li, Jason},
  booktitle = {30th Annual European Symposium on Algorithms (ESA 2022)},
  pages = {66--1},
  year = {2022},
  organization = {Schloss Dagstuhl--Leibniz-Zentrum f{\"u}r Informatik},
  doi = {10.4230/LIPIcs.ESA.2022.66},
  _bib2doi_finished = {true},
}

@inproceedings{lokshtanov20241,
  title = {A 1.9999-approximation algorithm for vertex cover on string graphs},
  author = {Lokshtanov, Daniel and Panolan, Fahad and Saurabh, Saket and Xue, Jie and Zehavi, Meirav},
  booktitle = {40th International Symposium on Computational Geometry (SoCG 2024)},
  volume = {293},
  pages = {72},
  year = {2024},
  doi = {10.4230/LIPIcs.SoCG.2024.72},
  publisher = {Schloss Dagstuhl - Leibniz-Zentrum f{\"{u}}r Informatik},
  _bib2doi_finished = {true},
}

@inproceedings{kumar2022point,
  title = {Point separation and obstacle removal by finding and hitting odd cycles},
  author = {Neeraj Kumar and Daniel Lokshtanov and Saket Saurabh and Subhash Suri and Jie Xue},
  year = {2022},
  doi = {10.4230/LIPIcs.SoCG.2022.52},
  booktitle = {38th International Symposium on Computational Geometry, SoCG 2022, Berlin, Germany, June 7-10, 2022},
  publisher = {Schloss Dagstuhl - Leibniz-Zentrum f{\"{u}}r Informatik},
  volume = {224},
  pages = {52:1--52:14},
  _bib2doi_selected = {dblp:/rec/conf/compgeom/0004L0S022.bib},
  _bib2doi_confirmed = {true},
  _bib2doi_finished = {true},
}

@inproceedings{kumar2021constant,
  title = {A constant factor approximation for navigating through connected obstacles in the plane},
  author = {Kumar, Neeraj and Lokshtanov, Daniel and Saurabh, Saket and Suri, Subhash},
  booktitle = {Proceedings of the 2021 ACM-SIAM Symposium on Discrete Algorithms (SODA)},
  pages = {822--839},
  year = {2021},
  organization = {SIAM},
  timestamp = {Thu, 15 Jul 2021 01:00:00 +0200},
  biburl = {https://dblp.org/rec/conf/soda/KumarLSS21.bib},
  bibsource = {dblp computer science bibliography, https://dblp.org},
  doi = {10.1137/1.9781611976465.52},
  _bib2doi_selected = {dblp:/rec/conf/soda/KumarLSS21.bib},
  _bib2doi_confirmed = {true},
}

@inproceedings{chang2025distance,
  title = {Distance Approximating Minors for Planar and Minor-Free Graphs},
  author = {Chang, Hsien-Chih and Conroy, Jonathan},
  year = {2025},
  doi = {10.1109/FOCS63196.2025.00039},
  booktitle = {66th {IEEE} Annual Symposium on Foundations of Computer Science, {FOCS} 2025, Sydney, Australia, December 14-17, 2025},
  publisher = {{IEEE}},
  pages = {734--754},
  _bib2doi_selected = {dblp:/rec/conf/focs/ChangC25.bib},
  _bib2doi_confirmed = {true},
  _bib2doi_finished = {true},
}

@article{nguyen2025asymptotic,
  title = {Asymptotic structure. {II.} {Path-width} and additive quasi-isometry},
  author = {Nguyen, Tung and Scott, Alex and Seymour, Paul},
  journal = {arXiv preprint arXiv:2509.09031},
  year = {2025},
  url = {https://arxiv.org/abs/2509.09031},
  _bib2doi_finished = {true},
}

@article{hathcock2025steiner,
  title = {The {S}teiner Path Aggregation Problem},
  author = {Chen, Da Qi and Hathcock, Daniel and Hershkowitz, D. Ellis and Ravi, R.},
  journal = {Information Processing Letters},
  pages = {106608},
  year = {2025},
  publisher = {Elsevier},
  doi = {10.1016/j.ipl.2025.106608},
  _bib2doi_finished = {true},
}

@article{thorup2004compact,
  title = {Compact oracles for reachability and approximate distances in planar digraphs},
  author = {Thorup, Mikkel},
  journal = {Journal of the ACM (JACM)},
  volume = {51},
  number = {6},
  pages = {993--1024},
  year = {2004},
  publisher = {ACM New York, NY, USA},
  timestamp = {Wed, 14 Nov 2018 00:00:00 +0100},
  biburl = {https://dblp.org/rec/journals/jacm/Thorup04.bib},
  bibsource = {dblp computer science bibliography, https://dblp.org},
  doi = {10.1145/1039488.1039493},
  _bib2doi_selected = {dblp:/rec/journals/jacm/Thorup04.bib},
  _bib2doi_confirmed = {true},
}

@article{davies2025string,
  title={String graphs are quasi-isometric to planar graphs},
  author={Davies, James},
  journal={arXiv preprint arXiv:2510.19602},
  year={2025},
  url={https://arxiv.org/abs/2510.19602}
}

@article{abrishami2025burling,
  title = {Burling graphs in graphs with large chromatic number},
  author = {Abrishami, Tara and Bria{\'n}ski, Marcin and Davies, James and Du, Xiying and Masa{\v{r}}{\'\i}kov{\'a}, Jana and Rz{\k{a}}{\.z}ewski, Pawe{\l} and Walczak, Bartosz},
  journal = {arXiv preprint arXiv:2510.19650},
  year = {2025},
  url = {https://arxiv.org/abs/2510.19650},
  _bib2doi_finished = {true},
}

@article{KLMN05,
  title = {Measured descent: {A} new embedding method for finite metrics},
  author = {Robert Krauthgamer and James R. Lee and Manor Mendel and Assaf Naor},
  journal = {Geometric and Functional Analysis GAFA},
  volume = {15},
  number = {4},
  pages = {839--858},
  year = {2005},
  publisher = {Springer},
  note = {Appeared in FOCS'04.},
  doi = {10.1007/s00039-005-0527-6},
  _bib2doi_finished = {true},
}

@inproceedings{fox2019embedding,
  title = {Embedding planar graphs into low-treewidth graphs with applications to efficient approximation schemes for metric problems},
  author = {{Fox-Epstein}, Eli and Klein, Philip N. and Schild, Aaron},
  booktitle = {Proceedings of the Thirtieth Annual ACM-SIAM Symposium on Discrete Algorithms},
  pages = {1069--1088},
  year = {2019},
  organization = {SIAM},
  timestamp = {Thu, 15 Jul 2021 01:00:00 +0200},
  biburl = {https://dblp.org/rec/conf/soda/Fox-EpsteinKS19.bib},
  bibsource = {dblp computer science bibliography, https://dblp.org},
  doi = {10.1137/1.9781611975482.66},
  _bib2doi_selected = {dblp:/rec/conf/soda/Fox-EpsteinKS19.bib},
  _bib2doi_confirmed = {true},
}

@inproceedings{filtser2022low,
  title = {Low treewidth embeddings of planar and minor-free metrics},
  author = {Filtser, Arnold and Le, Hung},
  booktitle = {2022 IEEE 63rd Annual Symposium on Foundations of Computer Science (FOCS)},
  pages = {1081--1092},
  year = {2022},
  organization = {IEEE},
  timestamp = {Sun, 05 May 2024 01:00:00 +0200},
  biburl = {https://dblp.org/rec/conf/focs/FiltserL22.bib},
  bibsource = {dblp computer science bibliography, https://dblp.org},
  doi = {10.1109/FOCS54457.2022.00105},
  _bib2doi_selected = {dblp:/rec/conf/focs/FiltserL22.bib},
  _bib2doi_confirmed = {true},
}

@article{awerbuch1992routing,
  title = {Routing with polynomial communication-space trade-off},
  author = {Awerbuch, Baruch and Peleg, David},
  journal = {SIAM Journal on Discrete Mathematics},
  volume = {5},
  number = {2},
  pages = {151--162},
  year = {1992},
  publisher = {SIAM},
  timestamp = {Sat, 25 Apr 2020 01:00:00 +0200},
  biburl = {https://dblp.org/rec/journals/siamdm/AwerbuchP92.bib},
  bibsource = {dblp computer science bibliography, https://dblp.org},
  doi = {10.1137/0405013},
  _bib2doi_selected = {dblp:/rec/journals/siamdm/AwerbuchP92.bib},
  _bib2doi_confirmed = {true},
}

@inproceedings{arya1995euclidean,
  title = {Euclidean spanners: short, thin, and lanky},
  author = {Arya, Sunil and Das, Gautam and Mount, David M. and Salowe, Jeffrey S and Smid, Michiel},
  booktitle = {Proceedings of the twenty-seventh annual ACM symposium on Theory of computing},
  pages = {489--498},
  year = {1995},
  timestamp = {Tue, 06 Nov 2018 00:00:00 +0100},
  biburl = {https://dblp.org/rec/conf/stoc/AryaDMSS95.bib},
  bibsource = {dblp computer science bibliography, https://dblp.org},
  doi = {10.1145/225058.225191},
  _bib2doi_selected = {dblp:/rec/conf/stoc/AryaDMSS95.bib},
  _bib2doi_confirmed = {true},
}

@article{gupta2005traveling,
  title = {Traveling with a pez dispenser (or, routing issues in {MPLS})},
  author = {Gupta, Anupam and Kumar, Amit and Rastogi, Rajeev},
  journal = {SIAM Journal on Computing},
  volume = {34},
  number = {2},
  pages = {453--474},
  year = {2005},
  publisher = {SIAM},
  timestamp = {Sat, 30 Sep 2023 01:00:00 +0200},
  biburl = {https://dblp.org/rec/journals/siamcomp/GuptaKR04.bib},
  bibsource = {dblp computer science bibliography, https://dblp.org},
  doi = {10.1137/S0097539702409927},
  _bib2doi_selected = {dblp:/rec/journals/siamcomp/GuptaKR04.bib},
  _bib2doi_confirmed = {true},
}

@inproceedings{bartal2019covering,
  title = {Covering metric spaces by few trees},
  author = {Yair Bartal and Nova Fandina and Ofer Neiman},
  year = {2019},
  doi = {10.4230/LIPIcs.ICALP.2019.20},
  booktitle = {46th International Colloquium on Automata, Languages, and Programming, {ICALP} 2019, Patras, Greece, July 9-12, 2019},
  publisher = {Schloss Dagstuhl - Leibniz-Zentrum f{\"{u}}r Informatik},
  volume = {132},
  pages = {20:1--20:16},
  _bib2doi_selected = {dblp:/rec/conf/icalp/BartalFN19.bib},
  _bib2doi_confirmed = {true},
  _bib2doi_finished = {true},
}

@inproceedings{abraham2012fully,
  title = {Fully dynamic approximate distance oracles for planar graphs via forbidden-set distance labels},
  author = {Abraham, Ittai and Chechik, Shiri and Gavoille, Cyril},
  booktitle = {Proceedings of the forty-fourth annual ACM symposium on Theory of computing},
  pages = {1199--1218},
  year = {2012},
  timestamp = {Tue, 06 Nov 2018 00:00:00 +0100},
  biburl = {https://dblp.org/rec/conf/stoc/AbrahamCG12.bib},
  bibsource = {dblp computer science bibliography, https://dblp.org},
  doi = {10.1145/2213977.2214084},
  _bib2doi_selected = {dblp:/rec/conf/stoc/AbrahamCG12.bib},
  _bib2doi_confirmed = {true},
}

@article {kerr23,
    AUTHOR = {Kerr, Alice},
     TITLE = {Tree approximation in quasi-trees},
   JOURNAL = {Groups Geom. Dyn.},
  FJOURNAL = {Groups, Geometry, and Dynamics},
    VOLUME = {17},
      YEAR = {2023},
    NUMBER = {4},
     PAGES = {1193--1233},
      ISSN = {1661-7207,1661-7215},
   MRCLASS = {51F30 (20E08 20F65)},
  MRNUMBER = {4641367},
MRREVIEWER = {Behnam\ Esmayli},
       DOI = {10.4171/ggd/733},
       URL = {https://doi.org/10.4171/ggd/733},
}

@article{cdn+-caaeg-2012,
  title = {Constant Approximation Algorithms for Embedding Graph Metrics into Trees and Outerplanar Graphs},
  author = {Chepoi, V. and Dragan, F. F. and Newman, I. and Rabinovich, Y. and Vax{\`e}s, Y.},
  year = {2012},
  month = {jan},
  journal = {Discrete \& Computational Geometry},
  volume = {47},
  number = {1},
  pages = {187--214},
  issn = {1432-0444},
  doi = {10.1007/s00454-011-9386-0},
  urldate = {2024-11-24},
  langid = {english},
  annotation = {TLDR: A simple factor 6 algorithm for approximating the optimal multiplicative distortion of embedding a graph metric into a tree metric is presented, improving and simplifying the factor 100 and 27 algorithms of B\v adoiu et al.},
  timestamp = {Sat, 09 Apr 2022 01:00:00 +0200},
  biburl = {https://dblp.org/rec/journals/dcg/ChepoiDNRV12.bib},
  bibsource = {dblp computer science bibliography, https://dblp.org},
  _bib2doi_selected = {dblp:/rec/journals/dcg/ChepoiDNRV12.bib},
  _bib2doi_confirmed = {true},
}

@article{lipton1979separator,
  title = {A separator theorem for planar graphs},
  author = {Lipton, Richard J and Tarjan, Robert Endre},
  journal = {SIAM Journal on Applied Mathematics},
  volume = {36},
  number = {2},
  pages = {177--189},
  year = {1979},
  doi = {10.1137/0136016},
  _bib2doi_finished = {true},
}

@mastersthesis{weng2025thesis,
  title = {Tree Covers of Unit Disk Graphs},
  author = {Weng, Puheng},
  year = {2025},
  doi = {10.7282/t3-nb7c-3w28},
  urldate = {2025-08-26},
  langid = {english},
  school = {Rutgers University - School of Graduate Studies},
  _bib2doi_finished = {true},
}

@article{SSS03,
  title = {Recognizing String Graphs in NP},
  author = {Schaefer, Marcus and Sedgwick, Eric and {\v S}tefankovi{\v c}, Daniel},
  year = {2003},
  month = {sep},
  journal = {Journal of Computer and System Sciences},
  volume = {67},
  number = {2},
  pages = {365--380},
  issn = {00220000},
  doi = {10.1016/S0022-0000(03)00045-X},
  urldate = {2019-03-22},
  langid = {english},
  timestamp = {Thu, 17 Jun 2021 01:00:00 +0200},
  biburl = {https://dblp.org/rec/journals/jcss/SchaeferSS03.bib},
  bibsource = {dblp computer science bibliography, https://dblp.org},
  _bib2doi_selected = {dblp:/rec/journals/jcss/SchaeferSS03.bib},
  _bib2doi_confirmed = {true},
}

@article{SSS11,
  title = {Spiraling and Folding: The Word View},
  shorttitle = {Spiraling and Folding},
  author = {Schaefer, Marcus and Sedgwick, Eric and {\v S}tefankovi{\v c}, Daniel},
  year = {2011},
  month = {jul},
  journal = {Algorithmica},
  volume = {60},
  number = {3},
  pages = {609--626},
  issn = {0178-4617, 1432-0541},
  doi = {10.1007/s00453-009-9362-8},
  urldate = {2019-05-02},
  langid = {english},
  timestamp = {Thu, 17 Jun 2021 01:00:00 +0200},
  biburl = {https://dblp.org/rec/journals/algorithmica/SchaeferSS11.bib},
  bibsource = {dblp computer science bibliography, https://dblp.org},
  _bib2doi_selected = {dblp:/rec/journals/algorithmica/SchaeferSS11.bib},
  _bib2doi_confirmed = {true},
}

@inproceedings{SSS07,
  title = {Spiralling and Folding: The Topological View},
  author = {Schaefer, Marcus and Sedgwick, Eric and {\v S}tefankovi{\v c}, Daniel},
  year = {2007},
  pages = {73--76},
  langid = {english},
  url = {http://cccg.ca/proceedings/2007/03b3.pdf},
  booktitle = {Proceedings of the 19th Annual Canadian Conference on Computational Geometry, {CCCG} 2007, August 20-22, 2007, Carleton University, Ottawa, Canada},
  publisher = {Carleton University, Ottawa, Canada},
  _bib2doi_selected = {dblp:/rec/conf/cccg/SchaeferSS07.bib},
  _bib2doi_confirmed = {true},
  _bib2doi_finished = {true},
}

@misc{KM26,
  title = {Approximation Schemes for Subset TSP and Steiner Tree on Geometric Intersection Graphs},
  author = {{Kisfaludi-Bak}, S{\'a}ndor and Marx, D{\'a}niel},
  year = 2026,
  month = mar,
  number = {arXiv:2603.26397},
  eprint = {2603.26397},
  primaryclass = {cs},
  publisher = {arXiv},
  doi = {10.48550/arXiv.2603.26397},
  urldate = {2026-03-31},
  archiveprefix = {arXiv}
}
}

\newpage
\appendix
\section{Proof of \Cref{lem:blt}}
\label{ap:bltproof}

We first state a simple well-known fact relating the outer face of a planar graph $H_0$ to crossing paths within $H_0$. This fact is frequently used for proofs involving paths in planar graphs, e.g. it is used to prove the Monge property of distances that lie on a face of a planar graph. For completeness we include a short proof.

\begin{figure}
    \centering
    \includegraphics[width=0.5\linewidth]{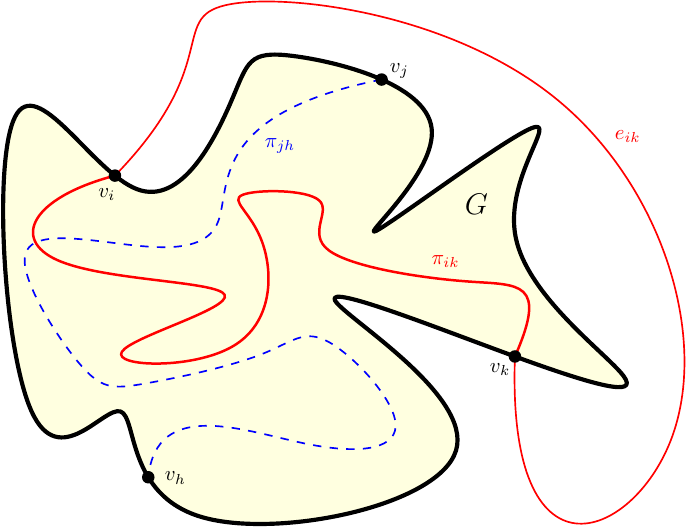}
    \caption{Figure for \Cref{lem:crossing_outer} depicting $v_j$ on the inside and $v_h$ on the outside of the curve $C$ formed with the path $\pi_{ik}$ and an arbitrary curve $e_{ik}$ in the outer face. Any walk from $v_j$ to $v_h$ must intersect the path $\pi_{ik}$.}
    \label{fig:crossing-outer}
\end{figure}

\begin{lemma}
\label{lem:crossing_outer}
Let $H_0$ be a planar graph with at least $2$ vertices. Fix a plane embedding of $H_0$ and let $Q = [v_1, v_2, \dots, v_\ell=v_1]$ be the closed walk around the outer face of $H_0$ starting at any outer vertex $v_1$.
For any $1\le i \le j \le k \le h \le \ell-1$, let $\pi_{ik}$ be a path from $v_i$ to $v_k$, and $\pi_{jh}$ be an arbitrary walk from $v_j$ to $v_h$.
The path $\pi_{ik}$ and walk $\pi_{jh}$ must intersect, i.e. share at least one vertex.
\end{lemma}
\begin{proof}
We assume without loss of generality that the path $\pi_{ik}$ has no vertices of the outer face in its interior, since we may consider the largest index $i' \le j$ such that $v_{i'}$ is on $\pi_{ik}$ and the smallest index $j \le k'$ such that $v_{k'}$ is on $\pi_{ik}$ and run the following argument on the subpath $\pi_{i'k'} \subset \pi_{ik}$.
We may also assume that $v_i$, $v_j$, $v_k$, and $v_h$ are all distinct, otherwise $\pi_{ik}$ and $\pi_{jh}$ trivially intersect at one of the end points.

Consider the simple closed curve \EMPH{$C$} formed by the the path $\pi_{ik}$ from $v_i$ to $v_k$ followed by any curve $e_{ik}$ in the outer face between $v_i$ and $v_k$. By the Jordan curve theorem, $C$ divides the plane into two regions; call the one region containing $v_j$ the \EMPH{inside} of $C$, and the other region the \EMPH{outside} of $C$.
It suffices to show that $v_h$ lies outside of $C$ as in \Cref{fig:crossing-outer}, as then any walk $\pi_{jh}$ from $v_j$ to $v_h$ must intersect $C$, and thus intersect $\pi_{ik}$.
Since $C$ does not intersect any vertices of the outer face except $v_i$ and $v_k$, the subwalk of $Q$ of the outer face between $v_i$ and $v_k$ containing $v_j$ must be on the inside of the $C$, while the subwalk of $Q$ between $v_k$ and $v_i$ containing $v_h$ must be outside of $C$.
\end{proof}

For convenience we begin by restating \Cref{lem:blt} here. As mentioned, the proof is implicit in \cite{BLT14}.
\bltlemma*

\begin{proof}
See \Cref{fig:outer-bounded} for an illustration of the three facts we wish to prove.

We begin by proving (1), that $(H,W)$ has either $1$ or $2$ critical vertices.
If $W = \varnothing$, then clearly $(H,W)$ has $1$ (arbitrary) critical vertex.
For simplicity, we will first assume that the outer face of $H_0$ is a simple cycle with $\ell$ vertices.
Since $H$ and $W$ both contain at least one outer vertex, there must be some adjacent vertex $p_0$ and $p_1$ on the outer face of $H_0$ such that $p_0\in W$ and $p_1 \in H$.
Observe that by definition $p_1$ is a critical vertex. Let us index the outer vertices $p_1, p_2, \dots, p_\ell = p_0$ of by the order they appear on the outer face starting with $p_1$.
To prove that $(H,W)$ has at most $2$ critical vertices, we claim that the outer vertices contained in $H$ must form a contiguous interval $p_1, p_2, \dots, p_a$ with $1\le a\le \ell-1$, and 
thus $p_a$ is the unique other critical vertex.
Suppose not, and there are vertices $p_i, p_{j+1}\in H$ and $p_{i+1}, p_j \in W$ with $1\le i < j\le \ell$. Since $H$ is connected, there is a path in $H$ from $p_i$ to $p_{j+1}$. Similarly, there is a path in $G[W]$ from $p_i+1$ to $p_j$. This is a contradiction, as any such paths must intersect by \Cref{lem:crossing_outer}. 
If the outer face of $H_0$ is not a simple cycle, and instead has cut vertices, the only difference is that the outer face defines a closed walk $Q = [p_1, p_2, \dots, p_\ell = p_0, p_1]$, where the contiguous subwalk $Q' = [p_1, p_2, \dots, p_a]$ is contained in $H$.
Since $Q$ is a walk, we may have $p_1 = p_a$ so $p_1$ is the unique critical vertex.

Next we prove (2). Let $\pi$ be an arbitrary path between $p_1$ and $p_a$. Let $\Gamma$ be a path between $s\in W$ and a vertex $p_i\in H$ that is an outer vertex of $H_0$, so $1\le i\le a$. 
Since $W$ is connected, there is a path $\pi_s$ from $p_0$ to $s$ that is completely contained in $W$. 
Then by \Cref{lem:crossing_outer}, the concatenated walk $\pi_s \circ \Gamma$ must intersect $\pi$. Furthermore, since $\pi_s$ consists solely of vertices in $W$, the intersection occurs at a vertex of $\Gamma$.

Finally we prove (3), that $H'$ is outer bounded by $W'$ in $H_0$.
The connected set $W'$ contains at least one outer vertex since it contains all of $\pi$, and $\pi$ is a (possibly trivial) path between outer vertices of $H_0$. Additionally, $H'$ has an outer vertex by assumption. 
Finally we need to verify that $H'$ is a connected component of $H_0\setminus W'$.
By definition, $H'$ is connected and vertex-disjoint from $W'$, so
to show that $H'$ is induced by a connected component of $H_0 \setminus W'$, it suffices to show the following: all edges in $H_0$ with exactly one end point in $H'$ have the other endpoint in  $W'$. 
The edges leaving $V(H')$ cannot connect to other vertices of $V(H\setminus W')$ since $H'$ is induced by a connected component of $H\setminus W'$. Nor can the edges leaving $H'$ be incident to vertices in $W$, by (2). Since all edges in $H_0$ with one end point in $V(H)$ have an end point in either $V(H)$ or $W$, we conclude that in $H_0$, all edges with exactly one end point in $V(H')$ are in $W'$.
Thus we conclude that $H'$ is outer bounded by $W'$.
\end{proof}

\end{document}